\documentclass[11pt,oneside]{article}
\usepackage[a4paper,margin=1in]{geometry}
\usepackage{amsmath,amssymb,bm,mathtools,mathrsfs,amsthm}
\usepackage{microtype}
\usepackage{mathrsfs}
\usepackage{booktabs}
\usepackage[hidelinks]{hyperref}
\usepackage[T1]{fontenc}
\usepackage{lmodern}
\usepackage{todonotes}

\numberwithin{equation}{section}
\newtheorem{theorem}{Theorem}[section]
\newtheorem{proposition}[theorem]{Proposition}
\theoremstyle{definition}

\usepackage{tikz-cd}

\title{\textbf{Fock-Space Formulation of the Boltzmann Collision Operator for Maxwell Molecules}}

\author{Ilya Karlin\thanks{ikarlin@ethz.ch}\\
Department of Mechanical and Process Engineering\\
ETH Zurich, CH-8092 Zurich, Switzerland}

\date{\today}

\begin{document}

\maketitle

\begin{abstract}
We develop a representation-independent symmetric-Fock formulation of the nonlinear Boltzmann collision operator for Maxwell molecules.  Starting from the standard velocity-space operator, we give a short derivation of Bobylev's Fourier identity and isolate its characteristic structure: two linearly related Fourier arguments followed by multiplication.  We then prove, independently of collision geometry and of any coordinate realization, a canonical lift--fusion intertwining theorem: one-particle linear maps lift to the two incoming Fock legs and symmetric-algebra multiplication fuses them into a single outgoing state.  Bargmann coordinates make this intrinsic construction especially transparent, but do not define it; Bobylev's collision maps subsequently select the Maxwell vertex from the already graded lift--fusion class.  The resulting bilinear map obeys the exact grading
\[
\widehat{\mathcal Q}(\mathscr H_p,\mathscr H_q)\subseteq\mathscr H_{p+q},
\]
which gives an intrinsic origin for the triangular structures of Maxwell kinetics.  Its one-vacuum restriction yields the complete Wang Chang--Uhlenbeck spectrum, while moment adaptation makes the nonequilibrium hierarchy strictly triangular and every finite compatible Fock jet an exact autonomous factor with global exponential relaxation.  For the inhomogeneous problem, we extend realization covariance to the nonlinear two-leg vertex and evaluate the moving-frame connection independently through Hermite-function and Fourier realizations.  The two routes give the same abstract Fock propagation and the same first Chapman--Enskog source; its traceless level-$2$ and vector level-$3$ sectors are relaxed by the exact Maxwell rates, giving $\Pr=2/3$.  This provides an explicit computational demonstration that coordinate realizations may be chosen for economy without entering the structural Fock formulation.
\end{abstract}



\section{Introduction}
\label{sec:maxwell_introduction}

The Boltzmann collision operator for Maxwell molecules occupies a special position in kinetic theory.  In velocity variables it is a fully nonlinear binary collision integral, but the absence of a relative-speed factor in the collision kernel gives it an exceptional Fourier representation.  Bobylev showed that the Fourier transform of the gain term factorizes into a product of two characteristic functions evaluated at collision-related wave vectors \cite{Bobylev1975Fourier,Bobylev1988MaxwellReview}.  This Fourier structure underlies much of the exact and spectral theory of the homogeneous Maxwell equation.

Several other coordinate formulations display closely related triangular or successive structures.  Maxwell's moment method and the work of Ikenberry and Truesdell give exceptionally simple production laws for low-order moments \cite{Maxwell1867DynamicalTheory,IkenberryTruesdell1956I,Truesdell1956II,TruesdellMuncaster1980Maxwell}; Ernst obtained successively solvable ordinary and Sonine moment equations \cite{Ernst1979MaxwellMoments}; and Glangetas, Li, and Xu derived an explicit triangular nonlinear system in Laguerre--spherical coordinates \cite{GlangetasLiXu2016Triangular}.  The recurrence of essentially the same degree hierarchy in different coordinate systems raises the representation-theoretic question pursued here: is triangularity created by a convenient basis, or does it belong to the collision map itself?

Our route begins with the standard nonlinear Maxwell collision operator rather than with Fock machinery.  Section~\ref{sec:maxwell_velocity_fourier} recalls the velocity-space model and derives Bobylev's identity directly.  The derivation exposes the elementary coordinate operation that is characteristic of Maxwell molecules: two functions are evaluated at two linear images of a common Fourier variable and the results are multiplied.  Only then do we introduce the minimal symmetric-Fock framework and ask for the abstract operation represented by this substitution--product rule.

The central structural result is the canonical lift--fusion theorem of Sec.~\ref{subsec:maxwell_lift_fusion_theorem}.  For arbitrary one-particle linear maps $A$ and $B$, their bosonic lifts act independently on the two incoming Fock legs and multiplication in the symmetric tensor algebra combines the results into one outgoing state.  We call this canonical two-to-one multiplication \emph{fusion}.  The intrinsic map
\begin{equation}
\mathcal V_{A,B}=\mathsf m\circ[\Gamma(A)\otimes\Gamma(B)]
\end{equation}
is defined without choosing coordinates; in Bargmann coordinates its image is simply
\begin{equation}
(F,G)\longmapsto F(A^T\bm z)G(B^T\bm z).
\end{equation}
Injectivity of the Bargmann realization gives uniqueness.  Thus Bargmann space is computationally privileged here, as in the Hermite-function realization of the Lebowitz--Frisch--Helfand model \cite{Karlin2026LFHFock}, but neither realization defines the underlying Fock structure.  Bobylev's geometry enters only afterward, to select the particular $A$ and $B$ appropriate to Maxwell collisions and to supply the angular average and loss term.

Because the canonical lifts preserve symmetric-tensor degree while fusion adds degrees, the resulting Maxwell vertex satisfies
\begin{equation}
\boxed{
\widehat{\mathcal Q}(\mathscr H_p,\mathscr H_q)
\subseteq
\mathscr H_{p+q}.
}
\end{equation}
Here $\mathscr H_r=\operatorname{Sym}^r(V)$ is the intrinsic Fock level of symmetric-tensor degree $r$; only after a Hermite realization is chosen does this degree appear as total Hermite order $r$.  Fourier, moment/Sonine, Hermite, and Laguerre--spherical triangular rules are therefore coordinate manifestations of a single representation-independent grading.

Two dynamical consequences are developed in detail.  Fixing one incoming leg to the Maxwellian vacuum gives a linear restriction that preserves each Fock level; in three dimensions its rotational decomposition yields the complete Wang Chang--Uhlenbeck spectral family \cite{WangChangUhlenbeck1952,WangChangUhlenbeck1970,AltermanFrankowskiPekeris1962,LernerMorimotoPravdaStarovXu2013}.  After moment adaptation removes the collision-invariant low-order directions, the nonlinear hierarchy becomes strictly triangular.  Every finite compatible Fock jet is then an exact autonomous factor of the homogeneous dynamics and relaxes exponentially to its Maxwellian jet without a smallness assumption on the retained amplitudes.

The present paper shares its local realization, moving-frame, and compatibility machinery with the Fock formulation of the Lebowitz--Frisch--Helfand kinetic model \cite{Karlin2026LFHFock}, but the inhomogeneous problem is also used here as a direct test of representation independence.  We first extend covariance from a one-input local operator to the nonlinear two-input Maxwell vertex.  We then construct the complete moving-frame connection by two coordinate routes: the Hermite-function route used in Ref.~\cite{Karlin2026LFHFock} and an economical Fourier route.  Although the intermediate coordinate connectors look different, both reduce to the same ladder-operator connection and hence to the same abstract first Chapman--Enskog source.  The source lies in the traceless quadratic and contracted cubic sectors; the only Maxwell-specific collision input is the pair of exact one-vacuum eigenvalues, whose ratio gives $\Pr=2/3$.

The algebraic statements are exact on the finite-particle Fock core.  A global Hilbert-Fock state additionally requires $f\in L^2(W^{-1}d\bm v)$ and can fail for heavy-tailed distributions, although finite jets remain meaningful when the corresponding moments exist \cite{CannoneKarch2010,ChoMorimotoWangYang2016}.  The dissipative statement needed here is supplied instead by the exact finite-jet relaxation theorem below, so no separate Fock reformulation of the Boltzmann $H$-functional is required.

The paper is organized as follows.  We first formulate the Maxwell collision model in velocity variables and derive its Fourier representation.  We then introduce the minimal Fock framework and the four canonical coordinate realizations needed below, with the velocity-space Hermite material compressed by reference to the Lebowitz--Frisch--Helfand formulation \cite{Karlin2026LFHFock}.  The canonical lift--fusion theorem is proved before any Maxwell-specific Fock construction.  We next return to Bobylev's geometry to construct the collision vertex, establish its intrinsic grading, derive the one-vacuum spectrum and compatible hierarchy, and prove the finite-jet relaxation theorem.  The final sections establish covariance of the nonlinear vertex and treat the inhomogeneous equation by parallel Hermite and Fourier connection routes before deriving the common Chapman--Enskog limit.  A brief outlook only notes that the three-dimensional Carleman representation may provide a natural starting point beyond Maxwell molecules; its Fock analysis is left to future work.

\section{Maxwell molecules in velocity and Fourier space}
\label{sec:maxwell_velocity_fourier}

\subsection{The nonlinear collision operator in velocity variables}
\label{subsec:maxwell_velocity_model}

Let $f(\bm v,t)$ be the one-particle velocity distribution of a spatially homogeneous monatomic gas.  Up to Sec.~\ref{sec:maxwell_inhomogeneous} we consider the spatially homogeneous problem.  We begin with the ordered bilinear form of the Boltzmann collision operator.  This polarization is essential rather than notational: the physical quadratic term $Q(f,f)$ has already identified the two incoming states, whereas the bilinear map $Q(f,g)$ exposes two independent collision legs and hence the tensor-product domain on which they can later be lifted separately.  For two velocity distributions $f$ and $g$,
\begin{equation}
Q(f,g)(\bm v)
=
\int_{\mathbb R^d}\!d\bm v_*
\int_{S^{d-1}}\!d\bm\sigma\,
B(|\bm q|,{\bm e}_q\!\cdot\!\bm\sigma)
\Big[
 f(\bm v')g(\bm v_*')-f(\bm v)g(\bm v_*)
\Big],
\label{eq:maxwell_velocity_bilinear_Q}
\end{equation}
where
\begin{equation}
\bm q:=\bm v-\bm v_*,
\qquad
{\bm e}_q:=\frac{\bm q}{|\bm q|},
\end{equation}
and the post-collisional velocities in the $\bm\sigma$ parametrization are
\begin{equation}
\bm v'
=\frac{\bm v+\bm v_*}{2}+\frac{|\bm q|}{2}\bm\sigma,
\qquad
\bm v_*'
=\frac{\bm v+\bm v_*}{2}-\frac{|\bm q|}{2}\bm\sigma.
\label{eq:maxwell_postcollisional_sigma}
\end{equation}
The physical collision term is the diagonal value $Q(f,f)$, and the homogeneous Boltzmann equation is
\begin{equation}
\partial_t f=Q(f,f).
\label{eq:maxwell_homogeneous_velocity_equation}
\end{equation}
We use the angular-cutoff setting throughout the finite-level analysis below.

For inverse-power collision models the kernel is often written schematically as
\begin{equation}
B(|\bm q|,\cos\vartheta)=|\bm q|^\gamma b(\cos\vartheta).
\end{equation}
Maxwell molecules are the case $\gamma=0$, so that
\begin{equation}
B(|\bm q|,{\bm e}_q\!\cdot\!\bm\sigma)
=b({\bm e}_q\!\cdot\!\bm\sigma)
\label{eq:maxwell_kernel_velocity}
\end{equation}
is independent of the relative speed.  Kinetically, this is the defining simplification of Maxwell molecules: the collision dependence left after the elastic kinematics is purely angular.  In the Fock formulation it will have a direct algebraic manifestation.  No additional relative-speed factor is present to mix Fock levels, and the Bobylev collision geometry closes inside the exactly graded lift--fusion class constructed below.

\subsection{Derivation of Bobylev's Fourier identity}
\label{subsec:maxwell_bobylev_derivation}

Bobylev introduced the Fourier-transform method for Maxwell molecules in 1975 \cite{Bobylev1975Fourier}; a detailed later exposition and proof can be found, for example, in Desvillettes \cite{Desvillettes2003FourierBoltzmann}.  We include the short derivation because the factorization it reveals is the point of departure for the Fock construction.

Use the Fourier convention
\begin{equation}
\widetilde f(\bm k)
=\mathcal F[f](\bm k)
:=\int_{\mathbb R^d}e^{-i\bm k\cdot\bm v}f(\bm v)\,d\bm v.
\label{eq:maxwell_fourier_transform}
\end{equation}
The standard pre/post-collisional change of variables rewrites the ordered collision integral in weak form as
\begin{equation}
\int_{\mathbb R^d}Q(f,g)(\bm v)\,\varphi(\bm v)\,d\bm v
=
\int f(\bm v)g(\bm v_*)
 b(\widehat{\bm q}\!\cdot\!\bm\sigma)
 \big[\varphi(\bm v')-\varphi(\bm v)\big]
 \,d\bm\sigma\,d\bm v_*\,d\bm v.
\label{eq:maxwell_weak_form_ordered}
\end{equation}
Set $\varphi(\bm v)=e^{-i\bm k\cdot\bm v}$.  From Eq.~\eqref{eq:maxwell_postcollisional_sigma},
\begin{equation}
e^{-i\bm k\cdot\bm v'}
=
e^{-\frac{i}{2}\bm k\cdot(\bm v+\bm v_*)}
 e^{-\frac{i}{2}|\bm q|\,\bm k\cdot\bm\sigma}.
\label{eq:maxwell_gain_exponential_before_exchange}
\end{equation}
The remaining angular integral contains a zonal kernel, i.e. a function only of a scalar product.  Spherical convolution with such a kernel is symmetric under interchange of the two distinguished directions.  Consequently,
\begin{equation}
\int_{S^{d-1}}
 b(\widehat{\bm q}\!\cdot\!\bm\sigma)
 e^{-\frac{i}{2}|\bm q|\,\bm k\cdot\bm\sigma}
 \,d\bm\sigma
=
\int_{S^{d-1}}
 b(\bm e_k\!\cdot\!\bm\sigma)
 e^{-\frac{i}{2}|\bm k|\,\bm q\cdot\bm\sigma}
 \,d\bm\sigma,
\qquad
\bm e_k:=\frac{\bm k}{|\bm k|}.
\label{eq:maxwell_bobylev_angular_exchange}
\end{equation}
This angular interchange is the decisive Maxwell-molecule step.  Because no additional factor $|\bm q|^\gamma$ is present, the exponential on the right can be combined with the center-of-mass factor in Eq.~\eqref{eq:maxwell_gain_exponential_before_exchange}:
\begin{equation}
e^{-\frac{i}{2}\bm k\cdot(\bm v+\bm v_*)}
 e^{-\frac{i}{2}|\bm k|\bm\sigma\cdot(\bm v-\bm v_*)}
=
e^{-i\bm k_+\cdot\bm v}
e^{-i\bm k_-\cdot\bm v_*},
\label{eq:maxwell_bobylev_factorization_step}
\end{equation}
where
\begin{equation}
\bm k_\pm
:=\frac12\bigl(\bm k\pm|\bm k|\bm\sigma\bigr).
\label{eq:maxwell_kpm}
\end{equation}
The gain integral therefore factorizes into two Fourier transforms.  The loss term factorizes immediately after the angular integration, whose value is independent of the direction of $\bm q$.  Thus
\begin{equation}
\boxed{
\mathcal F[Q(f,g)](\bm k)
=
\int_{S^{d-1}}
 b(\bm e_k\!\cdot\!\bm\sigma)
 \Big[
 \widetilde f(\bm k_+)\widetilde g(\bm k_-)
 -\widetilde f(\bm k)\widetilde g(\bm0)
 \Big]d\bm\sigma .
}
\label{eq:maxwell_bobylev}
\end{equation}
We denote the right-hand side by $\mathcal Q_F(\widetilde f,\widetilde g)(\bm k)$, so that
\begin{equation}
\mathcal F[Q(f,g)]
=
\mathcal Q_F(\widetilde f,\widetilde g).
\label{eq:maxwell_bobylev_operator_identity}
\end{equation}
The value at $\bm k=\bm0$ is understood by continuity.  On the diagonal, Eq.~\eqref{eq:maxwell_homogeneous_velocity_equation} becomes
\begin{equation}
\partial_t\widetilde f
=
\mathcal Q_F(\widetilde f,\widetilde f).
\label{eq:maxwell_bobylev_diagonal_evolution}
\end{equation}

The structural content of Eq.~\eqref{eq:maxwell_bobylev} is worth isolating before introducing any Fock notation.  For each collision direction $\bm\sigma$, the gain term performs two linear substitutions of the same Fourier variable and then multiplies the two resulting functions.  The problem addressed below is to identify the representation-independent operation whose coordinate image is precisely this substitution--product rule.

\section{Minimal Fock framework and coordinate realizations}
\label{sec:maxwell_fock}

Bobylev's identity has reduced the Maxwell gain term to a particularly simple coordinate operation.  To identify its abstract representative we need only a small amount of symmetric-Fock machinery.  The bosonic Fock construction and second quantization are classical \cite{Fock1932,Berezin1966,Hall2013}; the more general realization, covariance, and moving-frame framework used here was developed for kinetic equations in the Fock formulation of the Lebowitz--Frisch--Helfand model \cite{Karlin2026LFHFock}.  The present section records only the definitions needed to make the collision construction self-contained.

\subsection{Symmetric Fock space, grading, and intertwining}
\label{subsec:maxwell_fock_framework}

Let $V\simeq\mathbb C^d$ denote the vector space associated with the $d$ velocity components.  Its symmetric Fock space is
\begin{equation}
\mathscr H
=
\bigoplus_{r=0}^{\infty}\mathscr H_r,
\qquad
\mathscr H_r:=\operatorname{Sym}^r(V).
\label{eq:maxwell_symmetric_fock_def}
\end{equation}
The subspace $\mathscr H_r$ will be called the $r$th \emph{Fock level}.  This degree is intrinsic: it is symmetric-tensor rank and is defined before any Hermite, Fourier, or Bargmann coordinates are chosen.  The abstract vacuum, grading, and canonical operators belong to $\mathscr H$ itself; a realization introduced below only assigns coordinate representatives to these already-defined objects.

We use Dirac notation only as compact vector-space notation.  The vacuum $|0\rangle$ spans $\mathscr H_0$, and the canonical creation and annihilation operators obey
\begin{equation}
[\hat a_\alpha,\hat a_\beta^\dagger]=\delta_{\alpha\beta},
\qquad
[\hat a_\alpha,\hat a_\beta]=
[\hat a_\alpha^\dagger,\hat a_\beta^\dagger]=0,
\qquad
\hat a_\alpha|0\rangle=0.
\label{eq:maxwell_CCR_intro}
\end{equation}
For a multi-index $\bm n=(n_1,\ldots,n_d)$, the normalized number states are
\begin{equation}
|\bm n\rangle
=
\prod_{\alpha=1}^{d}
\frac{(\hat a_\alpha^\dagger)^{n_\alpha}}{\sqrt{n_\alpha!}}
|0\rangle,
\qquad
|\bm n|:=\sum_\alpha n_\alpha,
\label{eq:maxwell_number_states_intro}
\end{equation}
and $\mathscr H_r=\operatorname{span}\{|\bm n\rangle:|\bm n|=r\}$.  The number operator $\hat N=\hat a_\alpha^\dagger\hat a_\alpha$ therefore measures the same intrinsic degree.

An admissible coordinate realization is an invertible linear map $\mathcal R:\mathscr H\to\mathscr F$ into a chosen function space.  For any nonvanishing coordinate factor $h$, we write $S_h$ for multiplication, $(S_hF)(x)=h(x)F(x)$.  A one-input operator is represented by intertwining,
\begin{equation}
\mathcal R\widehat{\mathcal O}=\mathcal O\mathcal R,
\qquad
\widehat{\mathcal O}=\mathcal R^{-1}\mathcal O\mathcal R.
\label{eq:maxwell_linear_intertwining_intro}
\end{equation}
For the canonical algebra define
\begin{equation}
\Omega_{\mathcal R}:=\mathcal R|0\rangle,
\qquad
a^{(\mathcal R)}_\alpha:=\mathcal R\hat a_\alpha\mathcal R^{-1},
\qquad
a^{(\mathcal R)\dagger}_\alpha:=\mathcal R\hat a_\alpha^\dagger\mathcal R^{-1}.
\label{eq:maxwell_realization_vacuum_generators}
\end{equation}
Then vacuum consistency is automatic:
\begin{equation}
a^{(\mathcal R)}_\alpha\Omega_{\mathcal R}
=\mathcal R\hat a_\alpha|0\rangle=0,
\label{eq:maxwell_realization_vacuum_consistency}
\end{equation}
and the whole coordinate number basis is generated by
\begin{equation}
\mathcal R|\bm n\rangle
=
\prod_{\alpha=1}^{d}
\frac{\bigl(a^{(\mathcal R)\dagger}_\alpha\bigr)^{n_\alpha}}{\sqrt{n_\alpha!}}\,\Omega_{\mathcal R}.
\label{eq:maxwell_realization_basis_from_vacuum}
\end{equation}
Thus a realization transports not only operator formulas but also the distinguished Fock vacuum from which the canonical basis is generated.

The collision map is bilinear rather than a one-input operator.  The corresponding pull-back is equally elementary.
\begin{proposition}[Bilinear pull-back]
\label{prop:maxwell_bilinear_pullback}
Let $\mathcal Q:\mathscr F\times\mathscr F\to\mathscr F$ be bilinear and let $\mathcal R$ be an admissible realization.  Then
\begin{equation}
|\widehat{\mathcal Q}[\phi,\psi]\rangle
:=\mathcal R^{-1}\mathcal Q(\mathcal R|\phi\rangle,\mathcal R|\psi\rangle)
\label{eq:maxwell_general_bilinear_pullback}
\end{equation}
defines the unique bilinear Fock map satisfying
\begin{equation}
\mathcal R|\widehat{\mathcal Q}[\phi,\psi]\rangle
=\mathcal Q(\mathcal R|\phi\rangle,\mathcal R|\psi\rangle).
\label{eq:maxwell_general_bilinear_intertwining}
\end{equation}
\end{proposition}
\begin{proof}
Apply $\mathcal R$ to the definition.  Uniqueness follows from invertibility and bilinearity is inherited from $\mathcal Q$.
\end{proof}
The bilinear viewpoint is the tensorial entry point of the construction: by the universal property of the tensor product, a bilinear map on two incoming states is equivalently a linear map on $\mathscr H\otimes\mathscr H$.  We keep the lighter two-input notation here and introduce the explicit tensor-product representative only when the two collision legs are transformed independently.

\subsection{Velocity Hermite realizations: compact summary}
\label{subsec:maxwell_square_root}
Let
\begin{equation}
W_{\bm u,\theta}(\bm v)
=\frac{1}{(2\pi\theta)^{d/2}}
\exp\!\left[-\frac{|\bm v-\bm u|^2}{2\theta}\right],
\qquad
\bm\xi:=\frac{\bm v-\bm u}{\sqrt\theta}.
\label{eq:maxwell_local_W_xi}
\end{equation}
We use the probabilists' Hermite polynomials, adapted to the Gaussian weight $e^{-\xi^2/2}$,
\begin{equation}
\operatorname{He}_n(\xi)
=(-1)^n e^{\xi^2/2}\frac{d^n}{d\xi^n}e^{-\xi^2/2},
\qquad
\operatorname{He}_{\bm n}(\bm\xi)
=\prod_{\alpha=1}^d\operatorname{He}_{n_\alpha}(\xi_\alpha),
\qquad
\bm n!:=\prod_{\alpha=1}^d n_\alpha!.
\label{eq:maxwell_probabilists_hermite}
\end{equation}
They are related to the physicists' convention by
\begin{equation}
\operatorname{He}_n(x)=2^{-n/2}H_n(x/\sqrt2),
\qquad
H_n(x)=(-1)^n e^{x^2}\frac{d^n}{dx^n}e^{-x^2}.
\label{eq:maxwell_hermite_convention_relation}
\end{equation}
The Gaussian orthogonality is
\begin{equation}
\int_{\mathbb R^d}\frac{e^{-|\bm\xi|^2/2}}{(2\pi)^{d/2}}
\operatorname{He}_{\bm n}(\bm\xi)\operatorname{He}_{\bm m}(\bm\xi)\,d^d\xi
=\bm n!\,\delta_{\bm n\bm m}.
\label{eq:maxwell_hermite_orthogonality}
\end{equation}
This convention differs only by oscillator scaling from the physicists' Hermite convention used in the Lebowitz--Frisch--Helfand formulation \cite{Karlin2026LFHFock}; see also Grad's multidimensional Hermite construction \cite{Grad1949Hermite}.  With these definitions, the two velocity-space realizations are
\begin{align}
\mathcal R_H|\bm n\rangle
&=\sqrt W\,
\frac{\operatorname{He}_{\bm n}(\bm\xi)}{\sqrt{\bm n!}},
&\Omega_H&=\sqrt W,
\label{eq:maxwell_H_basis}
\\
\mathcal R_P|\bm n\rangle
&=\frac{\operatorname{He}_{\bm n}(\bm\xi)}{\sqrt{\bm n!}},
&\Omega_P&=1.
\label{eq:maxwell_P_basis}
\end{align}
They are related by $\mathcal R_H=S_{\sqrt W}\mathcal R_P$, and the same physical distribution is reconstructed as
\begin{equation}
f
=\sqrt W\,\mathcal R_H|\chi\rangle
=W\,\mathcal R_P|\chi\rangle.
\label{eq:maxwell_physical_reconstruction_P}
\end{equation}
The factors $W$ and $\sqrt W$ belong to the coordinate realization and reconstruction, not to the abstract Fock state.

The canonical generators are
\begin{align}
a^{(H)}_\alpha
&=\partial_{\xi_\alpha}+\frac{\xi_\alpha}{2},
&
a^{(H)\dagger}_\alpha
&=-\partial_{\xi_\alpha}+\frac{\xi_\alpha}{2},
\label{eq:maxwell_ladder_H}
\\
a^{(P)}_\alpha
&=\partial_{\xi_\alpha},
&
a^{(P)\dagger}_\alpha
&=\xi_\alpha-\partial_{\xi_\alpha}.
\label{eq:maxwell_ladder_P}
\end{align}
The two pairs are coordinate representatives of the same abstract operators.

\subsection{Fourier and Bargmann realizations}
\label{subsec:maxwell_fourier_bargmann}

Among the equivalent realizations, the Fourier--Bargmann pair is singled out below only by the Maxwell collision calculation.  Bargmann coordinates will make the lift--fusion operation shortest to evaluate, but the Fock algebra, grading, and collision vertex are not defined by this choice.

Using the Fourier convention of Eq.~\eqref{eq:maxwell_fourier_transform},
\begin{equation}
\widetilde W_{\bm u,\theta}(\bm k)
=\exp\!\left[-i\bm k\cdot\bm u-\frac{\theta|\bm k|^2}{2}\right].
\label{eq:maxwell_W_tilde}
\end{equation}
The Fourier transform of the Maxwellian-weighted Hermite polynomial is
\begin{equation}
\mathcal F\!\left[
W\frac{\operatorname{He}_{\bm n}(\bm\xi)}{\sqrt{\bm n!}}
\right]
=
\widetilde W(\bm k)
\frac{\bm z^{\bm n}}{\sqrt{\bm n!}},
\qquad
\bm z:=-i\sqrt\theta\,\bm k.
\label{eq:maxwell_fourier_hermite}
\end{equation}
This defines the Fourier realization
\begin{equation}
\mathcal R_F:=\mathcal F S_W \mathcal R_P
=\mathcal F S_{\sqrt W}\mathcal R_H,
\qquad
\mathcal R_F|\chi\rangle=\widetilde f,
\qquad
\Omega_F=\widetilde W.
\label{eq:maxwell_RF}
\end{equation}

Removing the Fourier vacuum gives the Bargmann realization
\begin{equation}
\mathcal R_B:=S_{\widetilde W}^{-1}\mathcal R_F,
\qquad
\mathcal R_B|\bm n\rangle
=\frac{\bm z^{\bm n}}{\sqrt{\bm n!}},
\qquad
\Omega_B=1.
\label{eq:maxwell_RB_basis_vacuum}
\end{equation}
For a general state,
\begin{equation}
\mathcal R_B|\chi\rangle
=: \mathscr B_\chi(\bm z)
=
\frac{\widetilde f(\bm k)}{\widetilde W(\bm k)}
=
\sum_{\bm n}\chi_{\bm n}\frac{\bm z^{\bm n}}{\sqrt{\bm n!}}.
\label{eq:maxwell_bargmann_symbol}
\end{equation}
The Bargmann generators take their simplest form,
\begin{equation}
a^{(B)}_\alpha=\partial_{z_\alpha},
\qquad
a^{(B)\dagger}_\alpha=z_\alpha,
\label{eq:maxwell_ladder_B}
\end{equation}
while conjugation by $\widetilde W$ gives the Fourier pair
\begin{equation}
a^{(F)}_\alpha
=\partial_{z_\alpha}-z_\alpha-\frac{u_\alpha}{\sqrt\theta},
\qquad
a^{(F)\dagger}_\alpha=z_\alpha.
\label{eq:maxwell_ladder_F_ann}
\end{equation}
Indeed $a^{(F)}_\alpha\widetilde W=0$, so Gaussian stripping simultaneously sends the Fourier vacuum $\widetilde W$ to the Bargmann vacuum $1$ and conjugates the annihilator to $\partial_{z_\alpha}$.

\paragraph{Four basic canonical realizations.}
The vacua and canonical generators are collected in one place because the same abstract state will be viewed in different coordinates later.  In each row $a^{(X)}_\alpha\Omega_X=0$ and Eq.~\eqref{eq:maxwell_realization_basis_from_vacuum} generates the coordinate basis.
\begin{center}
\small
\renewcommand{\arraystretch}{1.45}
\begin{tabular}{@{}p{0.17\textwidth} p{0.20\textwidth} p{0.58\textwidth}@{}}
\toprule
realization $\mathcal R_X$ & vacuum $\Omega_X=\mathcal R_X|0\rangle$ & canonical generators $a_\alpha^{(X)}$,  $a_\alpha^{(X)\dagger}$ \\
\midrule
$\mathcal R_H$ & $\sqrt W$ &
$\begin{aligned}
a^{(H)}_\alpha&=\partial_{\xi_\alpha}+\xi_\alpha/2,\
a^{(H)\dagger}_\alpha=-\partial_{\xi_\alpha}+\xi_\alpha/2
\end{aligned}$ \\
$\mathcal R_P$ & $1$ &
$\begin{aligned}
a^{(P)}_\alpha&=\partial_{\xi_\alpha},\
a^{(P)\dagger}_\alpha=\xi_\alpha-\partial_{\xi_\alpha}
\end{aligned}$ \\
$\mathcal R_F$ & $\widetilde W$ &
$\begin{aligned}
a^{(F)}_\alpha&=\partial_{z_\alpha}-z_\alpha-u_\alpha/\sqrt\theta,\
a^{(F)\dagger}_\alpha=z_\alpha
\end{aligned}$ \\
$\mathcal R_B$ & $1$ &
$\begin{aligned}
a^{(B)}_\alpha&=\partial_{z_\alpha},\
a^{(B)\dagger}_\alpha=z_\alpha
\end{aligned}$ \\
\bottomrule
\end{tabular}
\end{center}
The relations among the maps are
\begin{equation}
\mathcal R_H=S_{\sqrt W}\mathcal R_P,
\qquad
\mathcal R_F=\mathcal F S_W \mathcal R_P=\mathcal F S_{\sqrt W}\mathcal R_H,
\qquad
\mathcal R_B=S_{\widetilde W}^{-1}\mathcal R_F.
\label{eq:maxwell_four_realization_chain}
\end{equation}
It is useful to view the four realizations as two parallel pairs.  In velocity space, $\mathcal R_H$ retains the Gaussian vacuum whereas $\mathcal R_P$ strips it off, leaving Hermite polynomials.  In Fourier space, $\mathcal R_F$ retains the Fourier Maxwellian vacuum whereas $\mathcal R_B$ strips it off, leaving Bargmann monomials.  Thus the four realizations are not four unrelated constructions but two analogous vacuum-dressed/vacuum-stripped pairs:
\begin{equation}
\mathcal R_H=S_{\sqrt W}\mathcal R_P,
\qquad
\mathcal R_F=S_{\widetilde W}\mathcal R_B.
\label{eq:maxwell_parallel_realization_pairs}
\end{equation}
None of these coordinate choices is the Fock space itself.  Bargmann is privileged in the present collision calculation only because symmetric-tensor multiplication becomes ordinary polynomial multiplication there.

Laguerre--spherical functions are not a fifth basic realization here.  They are coordinate expressions, within the Hermite realizations, of the intrinsic trace/STF states introduced later.

\subsection{Algebraic Fock core, Hilbert domain, and finite moment jets}
\label{subsec:maxwell_fock_domain}

The collision algebra below is first defined on the algebraic Fock space
\begin{equation}
\mathscr H_{\rm alg}
:=\bigoplus_{r\ge0}^{\rm alg}\mathscr H_r,
\label{eq:maxwell_algebraic_fock_core}
\end{equation}
i.e. finite sums of homogeneous levels.  On this core all polynomial identities and finite-level collision constructions are exact.

The Hilbert completion is more restrictive.  Since $\mathcal R_H$ is the orthonormal Hermite-function realization and $f=\sqrt W \mathcal R_H|\chi\rangle$,
\begin{equation}
\|\chi\|_{\mathscr H}^2
=\int_{\mathbb R^d}\frac{|f(\bm v)|^2}{W_{\bm u,\theta}(\bm v)}\,d\bm v<\infty
\label{eq:maxwell_fock_hilbert_domain}
\end{equation}
is required for a global Hilbert-Fock state.  This need not hold for distributions with substantially slower than Gaussian tails \cite{Tang1993Hermite,SarnaGiesselmannTorrilhon2020}.

Finite jets require less.  Whenever the moments through degree $r$ exist, the Hermite coefficients
\begin{equation}
\chi_{\bm n}
=\int_{\mathbb R^d}f(\bm v)
\frac{\operatorname{He}_{\bm n}(\bm\xi)}{\sqrt{\bm n!}}\,d\bm v,
\qquad |\bm n|\le r,
\label{eq:maxwell_finite_moment_coefficients}
\end{equation}
are well defined and determine
\begin{equation}
J_r f:=\sum_{|\bm n|\le r}\chi_{\bm n}|\bm n\rangle.
\label{eq:maxwell_fock_jet}
\end{equation}
For Maxwell molecules the exact grading proved below makes the collision production of every existing finite jet meaningful independently of whether the full coefficient sequence belongs to the Hilbert completion.  Characteristic-function formulations extend still farther \cite{CannoneKarch2010,ChoMorimotoWangYang2016}; what may then fail is the identification of $\widetilde f/\widetilde W$ with an entire square-summable Bargmann function.  No completion beyond the domains stated here is needed for the finite-level results of this paper.

\section{Canonical lift--fusion intertwining on symmetric Fock space}
\label{subsec:maxwell_lift_fusion_theorem}

Section~\ref{subsec:maxwell_bobylev_derivation} has already shown what must be represented: at fixed collision direction the Maxwell gain term takes two generating functions, changes the argument of each by a linear map, and multiplies the results.  We now solve this representation problem \emph{without using any collision geometry}.  This separation is important.  The Fock operation will be determined once and for all for arbitrary linear maps $A$ and $B$; only in the following section will Bobylev's particular collision maps be inserted.

Below, we use only the minimal set of elementary operations pertinent to second quantization.  First, a linear map on the one-particle vector space is applied independently to every index of a symmetric tensor.  Second, two symmetric tensors are joined into one tensor of the summed degree.  The standard terminology for the first operation is the bosonic lift or second quantization; the second is simply multiplication in the symmetric algebra.  No quantum-mechanical interpretation is needed.

All statements in this section are algebraic on $\mathscr H_{\rm alg}=\operatorname{Sym}(V)$.  In particular, the lift and fusion maps are defined before Bargmann coordinates enter; Bargmann space will be used only afterward to display their coordinate action.

\subsection{Multiplication of symmetric tensors: the fusion map}
\label{subsec:maxwell_fusion_explained}

Consider first two homogeneous symmetric tensors.  If $v_1,\ldots,v_p,w_1,\ldots,w_q\in V$ and $\odot$ denotes symmetric tensor product, there is a canonical way to combine them:
\begin{equation}
(v_1\odot\cdots\odot v_p)
\;\mathsf m\;
(w_1\odot\cdots\odot w_q)
:=
v_1\odot\cdots\odot v_p\odot
w_1\odot\cdots\odot w_q.
\label{eq:maxwell_fusion_elementary_tensor}
\end{equation}
The result has degree $p+q$.  Extending linearly gives the ordinary multiplication map of the symmetric algebra,
\begin{equation}
\mathsf m:
\mathscr H_{\rm alg}\otimes_{\rm alg}\mathscr H_{\rm alg}
\longrightarrow
\mathscr H_{\rm alg}.
\label{eq:maxwell_fusion_intrinsic}
\end{equation}
We use the word \emph{fusion} only as a convenient name for this two-input, one-output operation.  Mathematically it is not an additional structure: it is the canonical multiplication already present in $\operatorname{Sym}(V)$.

The normalization seen in the number basis follows from this elementary multiplication.  In unnormalized creation monomials,
\begin{equation}
(\hat a^\dagger)^{\bm m}|0\rangle
\;\mathsf m\;
(\hat a^\dagger)^{\bm l}|0\rangle
=
(\hat a^\dagger)^{\bm m+\bm l}|0\rangle.
\label{eq:maxwell_fusion_unnormalized}
\end{equation}
Converting to the normalized states of Eq.~\eqref{eq:maxwell_number_states_intro} therefore gives
\begin{equation}
\mathsf m(|\bm m\rangle\otimes|\bm l\rangle)
=
\sqrt{\frac{(\bm m+\bm l)!}{\bm m!\bm l!}}
|\bm m+\bm l\rangle.
\label{eq:maxwell_fusion_number}
\end{equation}
The square-root factorial is thus forced by normalized Fock coordinates; it is not a collision-dependent coefficient.  In one dimension, for example,
\begin{equation}
|1\rangle\mathsf m|1\rangle=\sqrt2\,|2\rangle,
\qquad
|1\rangle\mathsf m|2\rangle=\sqrt3\,|3\rangle.
\label{eq:maxwell_fusion_low_example}
\end{equation}
These factors are simply what is required for multiplication of the corresponding normalized monomials $z^n/\sqrt{n!}$.

\subsection{Lifting a one-particle linear map}
\label{subsec:maxwell_lift_explained}

Let $A:V\to V$ be any linear map.  On the first Fock level $\mathscr H_1=V$ its action is already known.  There is then a canonical extension to a rank-two symmetric tensor:
\begin{equation}
v_1\odot v_2
\longmapsto
Av_1\odot Av_2,
\end{equation}
and similarly on degree $r$,
\begin{equation}
v_1\odot\cdots\odot v_r
\longmapsto
Av_1\odot\cdots\odot Av_r.
\label{eq:maxwell_lift_elementary_tensor}
\end{equation}
This defines
\begin{equation}
\Gamma(A)\big|_{\operatorname{Sym}^r(V)}
=A^{\otimes_s r},
\qquad r=0,1,2,\ldots,
\qquad
\Gamma(A)|0\rangle=|0\rangle.
\label{eq:maxwell_Gamma_sector}
\end{equation}
The notation $\Gamma(A)$ is standard for the bosonic lift, or second quantization, of $A$ \cite{Cook1951SecondQuantization,Segal1956TensorAlgebras}.  Here ``second quantization'' means only the canonical extension of a linear map on $V$ to all symmetric tensor powers.  Equivalently, $\Gamma(A)$ is the unique unital graded algebra map on $\operatorname{Sym}(V)$ whose restriction to degree one is $A$.

Two properties should be kept in view.  First, $\Gamma(A)$ does not change tensor rank:
\begin{equation}
\Gamma(A):\mathscr H_r\longrightarrow\mathscr H_r.
\label{eq:maxwell_Gamma_degree_preserving}
\end{equation}
Second, no basis has entered the definition; the construction is intrinsic to the symmetric tensor algebra.

\subsection{Lifting and fusion in Bargmann coordinates}
\label{subsec:maxwell_lift_fusion_bargmann_images}

The Bargmann realization turns the preceding tensor operations into familiar operations on generating functions \cite{Bargmann1961HilbertAnalytic}.  Fusion becomes ordinary multiplication.  On number states, Eq.~\eqref{eq:maxwell_fusion_number} gives
\begin{align}
\mathcal R_B\,\mathsf m(|\bm m\rangle\otimes|\bm l\rangle)
&=
\sqrt{\frac{(\bm m+\bm l)!}{\bm m!\bm l!}}
\frac{\bm z^{\bm m+\bm l}}{\sqrt{(\bm m+\bm l)!}}
\\
&=
\frac{\bm z^{\bm m}}{\sqrt{\bm m!}}
\frac{\bm z^{\bm l}}{\sqrt{\bm l!}}.
\end{align}
Hence, by bilinearity,
\begin{equation}
\mathcal R_B\,\mathsf m(|\phi\rangle\otimes|\psi\rangle)
=(\mathcal R_B|\phi\rangle)(\mathcal R_B|\psi\rangle).
\label{eq:maxwell_fusion_def}
\end{equation}

The lift becomes a linear substitution of the Bargmann variable.  The appearance of a transpose is already visible at rank one:
\begin{equation}
\bm z\cdot(Av)=(A^T\bm z)\cdot v.
\label{eq:maxwell_transpose_rank_one}
\end{equation}
Applying the same identity to every tensor index shows that contraction of $A^{\otimes_s r}\phi_r$ with $\bm z^{\otimes r}$ equals contraction of $\phi_r$ with $(A^T\bm z)^{\otimes r}$.  Therefore
\begin{equation}
\bigl(\mathcal R_B\Gamma(A)|\phi\rangle\bigr)(\bm z)
=(\mathcal R_B|\phi\rangle)(A^T\bm z).
\label{eq:maxwell_Gamma_def}
\end{equation}
We summarize Eqs.~\eqref{eq:maxwell_fusion_def} and \eqref{eq:maxwell_Gamma_def} as follows.

\begin{proposition}[Bargmann images of fusion and lifting]
\label{prop:maxwell_fusion_lift_Bargmann}
For $|\phi\rangle,|\psi\rangle\in\mathscr H_{\rm alg}$ and any $A:V\to V$,
\begin{align}
\mathcal R_B\,\mathsf m(|\phi\rangle\otimes|\psi\rangle)
&=(\mathcal R_B|\phi\rangle)(\mathcal R_B|\psi\rangle),
\\
\bigl(\mathcal R_B\Gamma(A)|\phi\rangle\bigr)(\bm z)
&=(\mathcal R_B|\phi\rangle)(A^T\bm z).
\end{align}
\end{proposition}

Thus the dictionary needed for Bobylev's gain term is already complete:
\begin{equation}
\begin{array}{ccl}
\Gamma(A) &\longleftrightarrow& F(\bm z)\mapsto F(A^T\bm z),\\[1mm]
\mathsf m &\longleftrightarrow& (F,G)\mapsto FG.
\end{array}
\label{eq:maxwell_lift_fusion_dictionary}
\end{equation}
The canonical two-input operation representing two substitutions followed by multiplication can now be stated without guesswork.

\subsection{Canonical lift--fusion theorem and uniqueness}
\label{subsec:maxwell_lift_fusion_theorem_statement}

\begin{theorem}[Canonical lift--fusion intertwining and uniqueness]
\label{thm:maxwell_canonical_lift_fusion}
Let $A,B:V\to V$ be arbitrary linear maps and define
\begin{equation}
|\mathcal V_{A,B}[\phi,\psi]\rangle
:=
\mathsf m\!\left(
\Gamma(A)|\phi\rangle
\otimes
\Gamma(B)|\psi\rangle
\right),
\qquad
|\phi\rangle,|\psi\rangle\in\mathscr H_{\rm alg}.
\label{eq:maxwell_general_lift_fusion_map}
\end{equation}
Then
\begin{equation}
\boxed{
\bigl(\mathcal R_B|\mathcal V_{A,B}[\phi,\psi]\rangle\bigr)(\bm z)
=
(\mathcal R_B|\phi\rangle)(A^T\bm z)\,
(\mathcal R_B|\psi\rangle)(B^T\bm z).
}
\label{eq:maxwell_general_lift_fusion_intertwining}
\end{equation}
Moreover, $\mathcal V_{A,B}$ is the unique bilinear map on $\mathscr H_{\rm alg}\times\mathscr H_{\rm alg}$ with this Bargmann-coordinate action, and
\begin{equation}
\mathcal V_{A,B}(\mathscr H_p,\mathscr H_q)
\subseteq\mathscr H_{p+q}.
\label{eq:maxwell_general_lift_fusion_grading}
\end{equation}
\end{theorem}

\begin{proof}
Apply the two entries of the dictionary \eqref{eq:maxwell_lift_fusion_dictionary} in their natural order.  Fusion first gives
\begin{equation}
\mathcal R_B|\mathcal V_{A,B}[\phi,\psi]\rangle
=
\bigl(\mathcal R_B\Gamma(A)|\phi\rangle\bigr)
\bigl(\mathcal R_B\Gamma(B)|\psi\rangle\bigr),
\end{equation}
and the lift identity then gives Eq.~\eqref{eq:maxwell_general_lift_fusion_intertwining}.

For uniqueness, suppose another bilinear map $\mathcal W_{A,B}$ has the same Bargmann image for every pair of inputs.  Then
\begin{equation}
\mathcal R_B\Big(
|\mathcal W_{A,B}[\phi,\psi]\rangle
-|\mathcal V_{A,B}[\phi,\psi]\rangle
\Big)=0.
\end{equation}
Because $\mathcal R_B$ is injective on the algebraic core, the two Fock outputs are equal for every $\phi,\psi$.  Hence no alternative Fock vertex and no adjustable normalization remain once the coordinate operation and normalized basis are fixed.

Finally, $\Gamma(A)$ preserves $\mathscr H_p$ and $\Gamma(B)$ preserves $\mathscr H_q$, while fusion maps $\mathscr H_p\otimes\mathscr H_q$ into $\mathscr H_{p+q}$.  This proves the grading statement.
\end{proof}

The uniqueness statement is important for the logic of the collision construction.  Equation~\eqref{eq:maxwell_general_lift_fusion_intertwining} is not a factorization chosen in order to reproduce Bobylev coefficients.  Rather, it says that once a coordinate operation of the form
\begin{equation}
(F,G)\longmapsto F(A^T\bm z)G(B^T\bm z)
\label{eq:maxwell_generic_substitution_product}
\end{equation}
is given, its Fock representative is already fixed.  Bobylev's role in the next section is therefore only to supply the collision-dependent maps $A$ and $B$, the angular kernel, and the loss term.

The grading is equally independent of Maxwell kinetics.  It follows before any collision geometry is inserted from the two elementary facts
\begin{equation}
\Gamma(A):\mathscr H_p\to\mathscr H_p,
\qquad
\mathsf m:\mathscr H_p\otimes\mathscr H_q\to\mathscr H_{p+q}.
\end{equation}
This observation will become the intrinsic Maxwell selection rule once Bobylev's maps are substituted.

\section{Fock-space Maxwell collision vertex}
\label{subsec:maxwell_bobylev}

We now return to the Maxwell-specific information derived in Sec.~\ref{subsec:maxwell_bobylev_derivation}.  We enter this specialization with the grading already established: for arbitrary linear maps $A$ and $B$, lift--fusion sends $\mathscr H_p\times\mathscr H_q$ into $\mathscr H_{p+q}$.  The particular Bobylev matrices therefore cannot create or spoil this degree rule; they only select collision-specific members of an already graded class.  It remains to verify that the angular average and loss term also preserve the same target level.  Bobylev's identity is thus used here to identify the two linear substitutions, angular kernel, and loss term that select the Maxwell vertex from the general class of Theorem~\ref{thm:maxwell_canonical_lift_fusion}.

The Fourier collision arguments satisfy
\begin{equation}
\bm k_++\bm k_-=\bm k,
\qquad
|\bm k_+|^2+|\bm k_-|^2=|\bm k|^2,
\label{eq:maxwell_k_invariants}
\end{equation}
and therefore the Maxwellian characteristic function of Eq.~\eqref{eq:maxwell_W_tilde} factorizes exactly:
\begin{equation}
\widetilde W_{\bm u,\theta}(\bm k_+)\,
\widetilde W_{\bm u,\theta}(\bm k_-)
=
\widetilde W_{\bm u,\theta}(\bm k).
\label{eq:maxwell_W_tilde_factorization}
\end{equation}
This is the second simplification, after Bobylev's Fourier factorization, that makes Bargmann coordinates natural for Maxwell molecules.

Write
\begin{equation}
\widetilde f=\widetilde W F,
\qquad
\widetilde g=\widetilde W G,
\qquad
\bm z=-i\sqrt\theta\,\bm k,
\qquad
\bm z_\pm=-i\sqrt\theta\,\bm k_\pm.
\end{equation}
Dividing Eq.~\eqref{eq:maxwell_bobylev} by the common factor $\widetilde W(\bm k)$ gives the Bargmann-coordinate collision map
\begin{equation}
\mathcal Q_B(F,G)(\bm z)
:=
\int_{S^{d-1}}
 b(\bm e_k\!\cdot\!\bm\sigma)
 \Big[
 F(\bm z_+)G(\bm z_-)-F(\bm z)G(\bm0)
 \Big]d\bm\sigma .
\label{eq:maxwell_QB_bilinear_def}
\end{equation}
Equivalently,
\begin{equation}
\widetilde W^{-1}\mathcal Q_F(\widetilde W F,\widetilde W G)
=\mathcal Q_B(F,G).
\label{eq:maxwell_fourier_to_bargmann_collision}
\end{equation}
For the physical state $F=G=\mathscr B_\chi$,
\begin{equation}
\partial_t\mathscr B_\chi
=\mathcal Q_B(\mathscr B_\chi,\mathscr B_\chi).
\label{eq:maxwell_bobylev_bargmann}
\end{equation}
The point is not to define a new collision model in Bargmann space.  Equations~\eqref{eq:maxwell_bobylev_operator_identity} and \eqref{eq:maxwell_fourier_to_bargmann_collision} state that $\mathcal Q_B$ is simply the Gaussian-stripped coordinate image of the standard velocity-space Maxwell operator.

\subsection{Abstract bilinear collision vertex}
\label{subsec:maxwell_abstract_vertex}

The preceding subsection has identified $\mathcal Q_B$ as the Bargmann-coordinate realization of the Maxwell collision map.  The general bilinear pull-back of Proposition~\ref{prop:maxwell_bilinear_pullback} therefore defines the abstract collision output by
\begin{equation}
|\widehat{\mathcal Q}[\phi,\psi]\rangle
:=
\mathcal R_B^{-1}\mathcal Q_B(\mathcal R_B|\phi\rangle,\mathcal R_B|\psi\rangle).
\label{eq:maxwell_abstract_Q}
\end{equation}
At this point no tensor-product collision notation is needed: $\widehat{\mathcal Q}$ is simply the abstract bilinear map corresponding to the familiar ordered collision map.

\begin{theorem}[Nonlinear intertwining of the Maxwell collision map]
\label{thm:maxwell_nonlinear_intertwining}
On the algebraic Fock core $\mathscr H_{\rm alg}$,
\begin{equation}
\boxed{
\mathcal R_B|\widehat{\mathcal Q}[\phi,\psi]\rangle
=
\mathcal Q_B(\mathcal R_B|\phi\rangle,\mathcal R_B|\psi\rangle),
\qquad
|\phi\rangle,|\psi\rangle\in\mathscr H_{\rm alg}.
}
\label{eq:maxwell_nonlinear_intertwining_theorem}
\end{equation}
Any extension to a completed Fock/Bargmann space requires the corresponding continuity and domain properties of the collision map; no such extension is needed for the finite-level results below.
\end{theorem}

\begin{proof}
This is the defining pull-back identity \eqref{eq:maxwell_abstract_Q} pushed forward by $\mathcal R_B$; uniqueness follows from invertibility of the realization.  The nontrivial structural statement needed for the gain construction is not hidden in this definition: Theorem~\ref{thm:maxwell_canonical_lift_fusion} has already proved, independently of Bobylev, that every product of two linearly substituted Bargmann symbols has the unique Fock representative obtained by lifting the two linear maps and then applying intrinsic fusion.  The Maxwell collision geometry will be inserted into that theorem in Sec.~\ref{subsec:maxwell_gamma_m}.
\end{proof}

Let
\begin{equation}
 e_{\bm n}(\bm z)
 :=\mathcal R_B|\bm n\rangle
 =\frac{\bm z^{\bm n}}{\sqrt{\bm n!}}
\label{eq:maxwell_Bargmann_basis_short}
\end{equation}
denote the normalized Bargmann monomials.  Directly from Eq.~\eqref{eq:maxwell_QB_bilinear_def}, a number-state pair satisfies
\begin{equation}
\begin{aligned}
\mathcal R_B|\widehat{\mathcal Q}[\bm m,\bm l]\rangle(\bm z)
=\int_{S^{d-1}}
 b(\bm e_k\!\cdot\!\bm\sigma)
 \Big[
 &e_{\bm m}(\bm z_+)e_{\bm l}(\bm z_-)
 -\delta_{\bm l,\bm0}\,e_{\bm m}(\bm z)
 \Big]d\bm\sigma .
\end{aligned}
\label{eq:maxwell_basis_pair_Bobylev_action}
\end{equation}
The gain part is thus a product of two independently transformed Bargmann monomials.  Section~\ref{subsec:maxwell_gamma_m} will identify this product, before angular averaging, with the canonical lift--fusion map of Theorem~\ref{thm:maxwell_canonical_lift_fusion}; no basis-dependent normalization is introduced at that stage.

Every number-basis matrix element of the abstract vertex is now an expansion coefficient of the explicitly known outgoing Bargmann function in Eq.~\eqref{eq:maxwell_basis_pair_Bobylev_action}.  Define its gain part by
\begin{equation}
\mathcal G_{\bm m\bm l}(\bm z)
:=
\int_{S^{d-1}}
 b(\bm e_k\!\cdot\!\bm\sigma)
 e_{\bm m}(\bm z_+)e_{\bm l}(\bm z_-)
 \,d\bm\sigma .
\label{eq:maxwell_gain_polynomial_def}
\end{equation}
The gain coefficients are the unique Bargmann coefficients of this function,
\begin{equation}
\mathcal G_{\bm m\bm l}(\bm z)
=
\sum_{\bm n}
\mathcal C^{+}_{\bm n;\bm m\bm l}\,
 e_{\bm n}(\bm z),
\qquad
\mathcal C^{+}_{\bm n;\bm m\bm l}
=
\frac{1}{\sqrt{\bm n!}}
\left.
\partial_{\bm z}^{\bm n}\mathcal G_{\bm m\bm l}(\bm z)
\right|_{\bm z=\bm0}.
\label{eq:maxwell_gain_coefficients_explicit}
\end{equation}
The derivative formula is simply coefficient extraction from the normalized Bargmann basis.  The collision-normal factorization below provides the corresponding finite-level computational route.

With
\begin{equation}
 \beta_0:=\int_{S^{d-1}}b(\bm e\!\cdot\!\bm\sigma)\,d\bm\sigma,
 \qquad |\bm e|=1,
\label{eq:maxwell_beta0_def}
\end{equation}
rotational invariance makes $\beta_0$ independent of the chosen unit vector $\bm e$.  Equation~\eqref{eq:maxwell_basis_pair_Bobylev_action} then gives
\begin{equation}
\mathcal Q_{\bm n;\bm m\bm l}
=
\mathcal C^{+}_{\bm n;\bm m\bm l}
-
\beta_0\,\delta_{\bm n,\bm m}\delta_{\bm l,\bm0}.
\label{eq:maxwell_Q_matrix_from_Bobylev}
\end{equation}
Hence
\begin{equation}
\widehat{\mathcal Q}(|\bm m\rangle,|\bm l\rangle)
=
\sum_{\bm n}
\mathcal Q_{\bm n;\bm m\bm l}|\bm n\rangle .
\label{eq:maxwell_basis_pair_output}
\end{equation}
For a general state $|\chi\rangle=\sum_{\bm n}\chi_{\bm n}|\bm n\rangle$,
\begin{equation}
\widehat{\mathcal Q}(|\chi\rangle,|\chi\rangle)
=
\sum_{\bm n,\bm m,\bm l}
\mathcal Q_{\bm n;\bm m\bm l}
\chi_{\bm m}\chi_{\bm l}|\bm n\rangle .
\label{eq:maxwell_Q_components}
\end{equation}
Only now, when a linear two-leg representation is useful, do we introduce the tensor-product representative of the same bilinear map,
\begin{equation}
\widehat{\mathbb Q}:
\mathscr H_{\rm alg}\otimes_{\rm alg}\mathscr H_{\rm alg}
\longrightarrow\mathscr H_{\rm alg},
\qquad
\widehat{\mathbb Q}(|\phi\rangle\otimes|\psi\rangle)
:=|\widehat{\mathcal Q}[\phi,\psi]\rangle .
\label{eq:maxwell_Q_tensor_from_bilinear}
\end{equation}
In the number basis,
\begin{equation}
\widehat{\mathbb Q}
=
\sum_{\bm n,\bm m,\bm l}
\mathcal Q_{\bm n;\bm m\bm l}
|\bm n\rangle\langle\bm m|\otimes\langle\bm l|.
\label{eq:maxwell_rank3_vertex}
\end{equation}
Thus $\widehat{\mathbb Q}$ contains no new collision physics; it is the linear tensor-product representative of $\widehat{\mathcal Q}$, introduced only where the two incoming legs need to be handled simultaneously.

\subsection{Bobylev specialization of the lift--fusion theorem}
\label{subsec:maxwell_gamma_m}

Theorem~\ref{thm:maxwell_canonical_lift_fusion} has already determined, independently of the collision model, the unique Fock representative of two linear substitutions followed by multiplication.  The task here is therefore not to guess or fit a Fock vertex.  Bobylev supplies the $\bm k$-dependence of the two factors; after rewriting that dependence as two linear maps, the lift--fusion theorem fixes the corresponding Fock operation automatically.  The physical angular kernel is incorporated only afterward.

\paragraph{From Bobylev's scattering vector to a collision normal.}
The Bobylev arguments
\begin{equation}
\bm k_\pm=\frac12\left(\bm k\pm|\bm k|\bm\sigma\right)
\end{equation}
form an orthogonal decomposition of $\bm k$:
\begin{equation}
\bm k_++\bm k_-=\bm k,
\qquad
\bm k_+\cdot\bm k_-=0.
\end{equation}
This suggests replacing the scattering vector $\bm\sigma$ by the normal to the corresponding orthogonal decomposition.  For $\bm\sigma\neq\bm e_k$ define
\begin{equation}
\bm\omega
:=
\frac{\bm e_k-\bm\sigma}{|\bm e_k-\bm\sigma|},
\qquad
\bm e_k\cdot\bm\omega\ge0,
\label{eq:maxwell_omega_def}
\end{equation}
so that $\bm\omega$ lies on the hemisphere
\begin{equation}
H_{\bm e_k}
:=
\left\{\bm\omega\in S^{d-1}:\bm e_k\cdot\bm\omega\ge0\right\}.
\end{equation}
The inverse relation is the reflection formula
\begin{equation}
\bm\sigma
=
\bm e_k
-2(\bm e_k\cdot\bm\omega)\bm\omega .
\label{eq:maxwell_sigma_omega_relation}
\end{equation}
Introduce the orthogonal projectors
\begin{equation}
P_{\bm\omega}=\bm\omega\bm\omega^T,
\qquad
A_{\bm\omega}=I-P_{\bm\omega},
\qquad
B_{\bm\omega}=P_{\bm\omega}.
\label{eq:maxwell_AB_projectors}
\end{equation}
Then the two collision arguments become \emph{linear} maps of the incoming Fourier vector,
\begin{equation}
\bm k_+=A_{\bm\omega}\bm k,
\qquad
\bm k_-=B_{\bm\omega}\bm k,
\label{eq:maxwell_collision_normal_k}
\end{equation}
and consequently
\begin{equation}
\bm z_+=A_{\bm\omega}\bm z,
\qquad
\bm z_-=B_{\bm\omega}\bm z .
\label{eq:maxwell_collision_normal_z}
\end{equation}
This linearization of the two Bargmann arguments is the sole reason for introducing the collision-normal parametrization.

The change of angular variable also transforms the Bobylev measure.  We denote its push-forward by $d\nu_{\bm e_k}(\bm\omega)$, defined through
\begin{equation}
\begin{aligned}
&\int_{S^{d-1}}
 b(\bm e_k\!\cdot\!\bm\sigma)
 F(\bm k_+,\bm k_-)
 \,d\bm\sigma
=
\int_{H_{\bm e_k}}
F(A_{\bm\omega}\bm k,B_{\bm\omega}\bm k)
\,d\nu_{\bm e_k}(\bm\omega) .
\end{aligned}
\label{eq:maxwell_pushforward_measure_def}
\end{equation}
With the standard surface measure this push-forward is explicitly
\begin{equation}
 d\nu_{\bm e_k}(\bm\omega)
 =
 2^{d-1}
 b\!\left(1-2(\bm e_k\cdot\bm\omega)^2\right)
 (\bm e_k\cdot\bm\omega)^{d-2}
 \,d\bm\omega,
 \qquad \bm\omega\in H_{\bm e_k}.
\label{eq:maxwell_pushforward_measure_explicit}
\end{equation}
The dependence on $\bm e_k$ is important: for a general angular Maxwell kernel the angular average is a separate part of the collision map and should not be confused with a fixed scalar measure on Fock space.

\paragraph{Specializing the canonical lifts.}
Equations~\eqref{eq:maxwell_collision_normal_z} have exactly the substitution form required by Theorem~\ref{thm:maxwell_canonical_lift_fusion}.  Hence the two incoming Bargmann hierarchies are represented by the already-defined canonical lifts $\Gamma(A_{\bm\omega})$ and $\Gamma(B_{\bm\omega})$.  No new lifting rule is introduced here.  For later calculations, the level matrices of these lifts can be written explicitly.  For a number state
$|\bm m\rangle$ with $|\bm m|=p$, Eq.~\eqref{eq:maxwell_Gamma_def} gives
\begin{equation}
\mathcal R_B\Gamma(A)|\bm m\rangle
=
\frac{(A^T\bm z)^{\bm m}}{\sqrt{\bm m!}},
\qquad
(A^T\bm z)^{\bm m}
:=
\prod_{\alpha=1}^{d}
\left(
\sum_{\beta=1}^{d}A_{\beta\alpha}z_\beta
\right)^{m_\alpha}.
\label{eq:maxwell_Gamma_number_symbol}
\end{equation}
Expanding again in the normalized Bargmann basis,
\begin{equation}
\frac{(A^T\bm z)^{\bm m}}{\sqrt{\bm m!}}
=
\sum_{|\bm r|=p}
\Gamma^{(p)}_{\bm r\bm m}(A)
\frac{\bm z^{\bm r}}{\sqrt{\bm r!}},
\end{equation}
we obtain the coefficient-extraction formula
\begin{equation}
\Gamma^{(p)}_{\bm r\bm m}(A)
=
\sqrt{\frac{\bm r!}{\bm m!}}
\,[\bm z^{\bm r}]\,(A^T\bm z)^{\bm m},
\qquad
|\bm r|=|\bm m|=p,
\label{eq:maxwell_Gamma_coeff_extract}
\end{equation}
%
%
where $[\bm z^{\bm r}]F(\bm z)$ denotes the coefficient of
$\bm z^{\bm r}=\prod_{\alpha=1}^d z_\alpha^{r_\alpha}$ in the polynomial
$F(\bm z)$. Equivalently,
\begin{equation}
[\bm z^{\bm r}]F(\bm z)
=
\frac{1}{\bm r!}
\left.
\partial_{\bm z}^{\bm r}F(\bm z)
\right|_{\bm z=0},
\qquad
\partial_{\bm z}^{\bm r}
=\prod_{\alpha=1}^d \partial_{z_\alpha}^{r_\alpha}.
\end{equation}
Hence Eq.~\eqref{eq:maxwell_Gamma_coeff_extract} may also be written as
\begin{equation}
\Gamma^{(p)}_{\bm r\bm m}(A)
=
\frac{1}{\sqrt{\bm r!\bm m!}}
\left.
\partial_{\bm z}^{\bm r}
(A^T\bm z)^{\bm m}
\right|_{\bm z=0}.
\end{equation}

This may be expanded into a completely finite
combinatorial sum.  Let $K=(k_{\alpha\beta})$ be a $d\times d$ matrix of
nonnegative integers with row sums $\bm m$ and column sums $\bm r$,
\begin{equation}
\sum_{\beta}k_{\alpha\beta}=m_\alpha,
\qquad
\sum_{\alpha}k_{\alpha\beta}=r_\beta,
\end{equation}
and denote the set of such matrices by
$\mathcal K(\bm r,\bm m)$.  Then the multinomial theorem gives
\begin{equation}
\Gamma^{(p)}_{\bm r\bm m}(A)
=
\sqrt{\bm r!\,\bm m!}
\sum_{K\in\mathcal K(\bm r,\bm m)}
\prod_{\alpha,\beta=1}^{d}
\frac{A_{\beta\alpha}^{\,k_{\alpha\beta}}}
{k_{\alpha\beta}!}.
\label{eq:maxwell_Gamma_coeff_finite_sum}
\end{equation}
Thus each lifted level is an explicitly computable finite matrix.  In
particular, at first level,
\begin{equation}
\Gamma^{(1)}_{\bm e_\alpha,\bm e_\beta}(A)
=A_{\alpha\beta},
\label{eq:maxwell_Gamma_first_sector}
\end{equation}
which is the matrix of $A$ on the one-particle coefficient vector, in
agreement with Eq.~\eqref{eq:maxwell_Gamma_sector}.

For the collision projectors of Eq.~\eqref{eq:maxwell_AB_projectors},
\begin{equation}
A_{\bm\omega}^T=A_{\bm\omega},
\qquad
B_{\bm\omega}^T=B_{\bm\omega}.
\label{eq:maxwell_AB_symmetric}
\end{equation}
Thus the transposes required by the general Bargmann-coordinate lift are
numerically immaterial after specializing to the collision maps.  We
nevertheless retain the transposition symbol in formulas inherited from the
general lift, so that the standard convention
$\Gamma(A)|_{\operatorname{Sym}^r(V)}=A^{\otimes_s r}$ remains visible
throughout.

The rank-one leg $B_{\bm\omega}=\bm\omega\bm\omega^T$ simplifies further.
Since
\begin{equation}
B_{\bm\omega}^T\bm z
=
B_{\bm\omega}\bm z
=
\bm\omega(\bm\omega\cdot\bm z),
\end{equation}
Eq.~\eqref{eq:maxwell_Gamma_coeff_extract} gives, for
$|\bm l|=|\bm s|=q$,
\begin{equation}
\Gamma^{(q)}_{\bm s\bm l}(B_{\bm\omega})
=
\frac{q!}{\sqrt{\bm s!\,\bm l!}}
\,\bm\omega^{\bm s+\bm l}.
\label{eq:maxwell_Gamma_B_rankone}
\end{equation}
For the complementary projector
$A_{\bm\omega}=I-\bm\omega\bm\omega^T$, one may use directly
\begin{equation}
\Gamma^{(p)}_{\bm r\bm m}(A_{\bm\omega})
=
\sqrt{\frac{\bm r!}{\bm m!}}
[\bm z^{\bm r}]
(A_{\bm\omega}^T\bm z)^{\bm m}
=
\sqrt{\frac{\bm r!}{\bm m!}}
[\bm z^{\bm r}]
\prod_{\alpha=1}^{d}
\left[
 z_\alpha-\omega_\alpha(\bm\omega\cdot\bm z)
\right]^{m_\alpha},
\label{eq:maxwell_Gamma_A_projector}
\end{equation}
or equivalently the finite sum
\eqref{eq:maxwell_Gamma_coeff_finite_sum} with
$A_{\alpha\beta}=\delta_{\alpha\beta}-\omega_\alpha\omega_\beta$.
Equations~\eqref{eq:maxwell_Gamma_B_rankone} and
\eqref{eq:maxwell_Gamma_A_projector} are the explicit level coefficients
needed in the collision-normal construction; no matrix exponential is
required for their evaluation.

For reference, in ladder-operator notation the same map has the normally ordered representation
\begin{equation}
\Gamma(A)
=
:\!\exp\!\left[
\hat a_\alpha^\dagger(A-I)_{\alpha\beta}\hat a_\beta
\right]\!:
\label{eq:maxwell_Gamma_normal}
\end{equation}
which remains valid for the noninvertible projectors $A_{\bm\omega}$ and $B_{\bm\omega}$ occurring in Eq.~\eqref{eq:maxwell_AB_projectors}.

\paragraph{Fusion of the transformed incoming legs.}
The second step is the intrinsic symmetric-algebra multiplication $\mathsf m$ defined in Eq.~\eqref{eq:maxwell_fusion_intrinsic}.  Proposition~\ref{prop:maxwell_fusion_lift_Bargmann} proves that this canonical multiplication becomes ordinary multiplication of Bargmann symbols, while Eq.~\eqref{eq:maxwell_fusion_number} fixes its normalized-basis prefactor.  Thus there is no collision-dependent choice in the fusion rule.  For the collision problem the relevant statement is the fusion of the \emph{transformed} basis legs.  If $|\bm m|=p$ and $|\bm l|=q$, write the level matrices of the lifts as in Eqs.~\eqref{eq:maxwell_Gamma_coeff_extract}--\eqref{eq:maxwell_Gamma_coeff_finite_sum}, namely
\begin{align}
\Gamma(A_{\bm\omega})|\bm m\rangle
&=
\sum_{|\bm r|=p}
\Gamma^{(p)}_{\bm r\bm m}(A_{\bm\omega})|\bm r\rangle,
\\
\Gamma(B_{\bm\omega})|\bm l\rangle
&=
\sum_{|\bm s|=q}
\Gamma^{(q)}_{\bm s\bm l}(B_{\bm\omega})|\bm s\rangle.
\end{align}
By Theorem~\ref{thm:maxwell_canonical_lift_fusion}, the unique fixed-normal gain state is
\begin{equation}
|\mathcal V_{\bm\omega}^{+}[\phi,\psi]\rangle
:=
|\mathcal V_{A_{\bm\omega},B_{\bm\omega}}[\phi,\psi]\rangle
=
\mathsf m\!\left[
\Gamma(A_{\bm\omega})|\phi\rangle
\otimes
\Gamma(B_{\bm\omega})|\psi\rangle
\right].
\label{eq:maxwell_fixed_normal_vertex}
\end{equation}
For an incoming basis pair this gives
\begin{equation}
\begin{aligned}
|\mathcal V_{\bm\omega}^{+}[\bm m,\bm l]\rangle
={}&
\sum_{\substack{|\bm r|=p\\|\bm s|=q}}
\Gamma^{(p)}_{\bm r\bm m}(A_{\bm\omega})
\Gamma^{(q)}_{\bm s\bm l}(B_{\bm\omega})
\sqrt{\frac{(\bm r+\bm s)!}{\bm r!\bm s!}}
|\bm r+\bm s\rangle .
\end{aligned}
\label{eq:maxwell_transformed_basis_fusion}
\end{equation}
Equation~\eqref{eq:maxwell_transformed_basis_fusion} also provides the practical finite-level route to the gain coefficients of Eq.~\eqref{eq:maxwell_gain_coefficients_explicit}: for fixed incoming degrees $p$ and $q$, evaluate the lifted matrices $\Gamma^{(p)}(A_{\bm\omega})$ and $\Gamma^{(q)}(B_{\bm\omega})$, fuse the resulting number states, perform the physical angular contraction, and finally read off the coefficient of $|\bm n\rangle$.  Because the pushed-forward angular measure retains its $\bm e_k$ dependence, this angular contraction is to be completed before the global Bargmann coefficient is extracted; Eq.~\eqref{eq:maxwell_gain_coefficients_explicit} is the unambiguous coefficient formula.

This formula is the basis-level content of the factorization.  The two bosonic lifts generally mix Cartesian number states within their respective incoming levels, and fusion combines every resulting pair into a single outgoing Fock vector.  In particular, the fixed-normal output is generally a superposition of number states rather than the single ket $|\bm m+\bm l\rangle$.

Equivalently, the general intertwining theorem gives directly
\begin{equation}
\bigl(\mathcal R_B|\mathcal V_{\bm\omega}^{+}[\chi,\phi]\rangle\bigr)(\bm z)
=
\mathscr B_\chi(A_{\bm\omega}^T\bm z)
\mathscr B_\phi(B_{\bm\omega}^T\bm z).
\label{eq:maxwell_fixed_normal_symbol}
\end{equation}
This identity is the decisive point: the fixed-normal Fock state is not chosen so as to reproduce the Bobylev product; it is the unique state already prescribed by Theorem~\ref{thm:maxwell_canonical_lift_fusion} once Bobylev has supplied $A_{\bm\omega}$ and $B_{\bm\omega}$.

\paragraph{From the fixed-normal vertex to the collision state.}
For two arbitrary incoming Fock vectors $|\phi\rangle$ and $|\psi\rangle$, denote this unique fixed-normal state by
\begin{equation}
|g_{\bm\omega}[\phi,\psi]\rangle
:=
|\mathcal V_{\bm\omega}^{+}[\phi,\psi]\rangle .
\label{eq:maxwell_fixed_normal_gain_state}
\end{equation}
We denote the full $\bm\omega$-indexed family in Eq.~\eqref{eq:maxwell_fixed_normal_gain_state} by $g_{\phi,\psi}$.  Its Bargmann symbol is
\begin{equation}
\mathscr B_{g_{\bm\omega}[\phi,\psi]}(\bm z)
=
\mathscr B_\phi(A_{\bm\omega}^T\bm z)
\mathscr B_\psi(B_{\bm\omega}^T\bm z).
\label{eq:maxwell_fixed_normal_gain_symbol}
\end{equation}
For the physical quadratic collision one sets $|\phi\rangle=|\psi\rangle=|\chi\rangle$.

It remains to contract this state-valued family with the physical angular kernel.  Define $\mathfrak A_b$ by
\begin{equation}
\bigl(\mathcal R_B\mathfrak A_b[g_{\phi,\psi}]\bigr)(\bm z)
:=
\int_{H_{\bm e_k}}
\mathscr B_{g_{\bm\omega}[\phi,\psi]}(\bm z)
\,d\nu_{\bm e_k}(\bm\omega),
\qquad
\bm z=-i\sqrt\theta\,\bm k .
\label{eq:maxwell_angular_average_map}
\end{equation}
The gain part is therefore simply
\begin{equation}
\widehat{\mathcal Q}^{+}(|\phi\rangle,|\psi\rangle)
=
\mathfrak A_b[g_{\phi,\psi}],
\label{eq:maxwell_Q_gain_operator_recipe}
\end{equation}
and the complete ordered bilinear Maxwell collision map is
\begin{equation}
\widehat{\mathcal Q}(|\phi\rangle,|\psi\rangle)
=
\mathfrak A_b[g_{\phi,\psi}]
-
\beta_0\,\langle0|\psi\rangle\,|\phi\rangle .
\label{eq:maxwell_Q_compact_bilinear}
\end{equation}
On the physical diagonal this becomes
\begin{equation}
\widehat{\mathcal Q}(|\chi\rangle,|\chi\rangle)
=
\mathfrak A_b[g_{\chi,\chi}]
-
\beta_0\chi_{\bm0}|\chi\rangle .
\label{eq:maxwell_Q_compact}
\end{equation}
Equations~\eqref{eq:maxwell_fixed_normal_gain_state}--\eqref{eq:maxwell_Q_compact} give the constructive recipe: the $\bm\omega$-dependence is carried by the explicitly named family $|g_{\bm\omega}[\phi,\psi]\rangle$, and $\mathfrak A_b$ performs only its final angular contraction.

\subsection{Worked example: the \texorpdfstring{$1+1\to2$}{1+1 to 2} collision block}
\label{subsec:maxwell_worked_11_2}

Before passing to the general grading theorem, it is useful to see the lifting--fusion--averaging mechanism on the lowest block for which both incoming legs are nontrivial.  No specialization of the spatial dimension is needed for this calculation, so we keep general $d$ and take two arbitrary first-level states
\begin{equation}
|p\rangle
=
p_\alpha\hat a_\alpha^\dagger|0\rangle,
\qquad
|q\rangle
=
q_\alpha\hat a_\alpha^\dagger|0\rangle.
\label{eq:maxwell_worked_input_states}
\end{equation}
By the general Bargmann realization
\eqref{eq:maxwell_bargmann_symbol}, their symbols are
\begin{equation}
\mathscr B_p(\bm z)=p_\alpha z_\alpha,
\qquad
\mathscr B_q(\bm z)=q_\alpha z_\alpha.
\label{eq:maxwell_worked_input_symbols}
\end{equation}
For compactness we denote the restriction of the collision map to this pair of first-level sectors by $\widehat{\mathcal Q}_{1,1}$; the general block notation is defined in the next subsection.

At fixed collision normal $\bm\omega$, specialize both incoming level degrees to one in the general first-level formula
\eqref{eq:maxwell_Gamma_first_sector}.  Because the collision maps
$A_{\bm\omega}$ and $B_{\bm\omega}$ are symmetric projectors,
$A_{\bm\omega}^T=A_{\bm\omega}$ and
$B_{\bm\omega}^T=B_{\bm\omega}$, so Eq.~\eqref{eq:maxwell_Gamma_sector}
gives
\begin{equation}
\Gamma(A_{\bm\omega})|p\rangle
=
(A_{\bm\omega}p)_\alpha
\hat a_\alpha^\dagger|0\rangle,
\qquad
\Gamma(B_{\bm\omega})|q\rangle
=
(B_{\bm\omega}q)_\beta
\hat a_\beta^\dagger|0\rangle .
\label{eq:maxwell_worked_lifts}
\end{equation}
This is the first-level specialization of the transformed-leg expansions
appearing immediately above Eq.~\eqref{eq:maxwell_transformed_basis_fusion}.

Applying the general fixed-normal gain-state definition
\eqref{eq:maxwell_fixed_normal_gain_state}, and then the number-state fusion
rule \eqref{eq:maxwell_fusion_number}, gives
\begin{equation}
|g_{\bm\omega}[p,q]\rangle
=
(A_{\bm\omega}p)_\alpha
(B_{\bm\omega}q)_\beta
\hat a_\alpha^\dagger\hat a_\beta^\dagger|0\rangle .
\label{eq:maxwell_worked_fusion}
\end{equation}
No additional global factor of $\sqrt2$ is required in
Eq.~\eqref{eq:maxwell_worked_fusion}: the occupation-number factor in
Eq.~\eqref{eq:maxwell_fusion_number} is already contained in the
creation-operator product.  In particular,
$(\hat a_\alpha^\dagger)^2|0\rangle=\sqrt2\,|2\bm e_\alpha\rangle$,
whereas
$\hat a_\alpha^\dagger\hat a_\beta^\dagger|0\rangle
=|\bm e_\alpha+\bm e_\beta\rangle$ for $\alpha\neq\beta$.

Finally, the general symbol identity
\eqref{eq:maxwell_fixed_normal_gain_symbol} gives
\begin{equation}
\mathscr B_{g_{\bm\omega}[p,q]}(\bm z)
=
[p\cdot A_{\bm\omega}^T\bm z]\,
[q\cdot B_{\bm\omega}^T\bm z].
\label{eq:maxwell_worked_fusion_symbol}
\end{equation}
Thus Eqs.~\eqref{eq:maxwell_worked_lifts}--\eqref{eq:maxwell_worked_fusion_symbol}
are respectively the $1+1$ specializations of the general lift, fusion, and
fixed-normal symbol formulas of the preceding subsection.

For the final angular contraction it is convenient to return to the equivalent scattering-vector parametrization.  Write
\begin{equation}
\bm z=\rho\bm e,
\qquad
|\bm e|=1,
\qquad
\eta=\bm e\cdot\bm\sigma ,
\end{equation}
so that
\begin{equation}
\bm z_\pm=\frac{\rho}{2}(\bm e\pm\bm\sigma).
\end{equation}
Applying the angular contraction
\eqref{eq:maxwell_angular_average_map}, equivalently the gain recipe
\eqref{eq:maxwell_Q_gain_operator_recipe}, and reverting from the
collision-normal variable $\bm\omega$ to the scattering vector $\bm\sigma$
through Eq.~\eqref{eq:maxwell_pushforward_measure_def}, gives the same
Bobylev gain integral as Eq.~\eqref{eq:maxwell_basis_pair_Bobylev_action},
now evaluated on the first-level superpositions \eqref{eq:maxwell_worked_input_states}.
The gain symbol is
\begin{equation}
\begin{aligned}
\mathcal R_B\widehat{\mathcal Q}^{+}(|p\rangle,|q\rangle)
=
\frac{\rho^2}{4}
\int_{S^{d-1}}b(\eta)
[p\cdot(\bm e+\bm\sigma)]
[q\cdot(\bm e-\bm\sigma)]
\,d\bm\sigma .
\end{aligned}
\label{eq:maxwell_worked_gain_integral}
\end{equation}
Define
\begin{equation}
\Omega_2
:=
\int_{S^{d-1}}b(\eta)(1-\eta^2)\,d\bm\sigma .
\label{eq:maxwell_Omega2}
\end{equation}
Rotational symmetry about $\bm e$ implies that the second angular moment has the form
\begin{equation}
\int_{S^{d-1}}
b(\eta)\sigma_\alpha\sigma_\beta\,d\bm\sigma
=
A\delta_{\alpha\beta}+B e_\alpha e_\beta .
\end{equation}
Taking the trace and contracting with $e_\alpha e_\beta$ gives
\begin{equation}
dA+B=\beta_0,
\qquad
A+B=\beta_0-\Omega_2,
\end{equation}
and therefore
\begin{equation}
\int_{S^{d-1}}
b(\eta)\sigma_\alpha\sigma_\beta\,d\bm\sigma
=
\frac{\Omega_2}{d-1}\delta_{\alpha\beta}
+
\left(
\beta_0-\frac{d}{d-1}\Omega_2
\right)e_\alpha e_\beta .
\label{eq:maxwell_worked_second_moment}
\end{equation}
The two first-moment cross terms in Eq.~\eqref{eq:maxwell_worked_gain_integral} cancel after angular integration.  Substituting Eq.~\eqref{eq:maxwell_worked_second_moment} and using $\bm z=\rho\bm e$ yields
\begin{equation}
\mathcal R_B\widehat{\mathcal Q}^{+}(|p\rangle,|q\rangle)
=
\frac{d\Omega_2}{4(d-1)}
\left[
(p\cdot\bm z)(q\cdot\bm z)
-
\frac1d(p\cdot q)\bm z^2
\right].
\label{eq:maxwell_worked_gain_result}
\end{equation}
Equivalently, with the $d$-dimensional STF product
\begin{equation}
p_{\langle\alpha}q_{\beta\rangle}
:=
\frac12(p_\alpha q_\beta+p_\beta q_\alpha)
-\frac1d\delta_{\alpha\beta}(p_\gamma q_\gamma),
\end{equation}
the Fock output is
\begin{equation}
\widehat{\mathcal Q}_{1,1}(|p\rangle,|q\rangle)
=
\frac{d\Omega_2}{4(d-1)}
p_{\langle\alpha}q_{\beta\rangle}
\hat a_\alpha^\dagger\hat a_\beta^\dagger|0\rangle .
\label{eq:maxwell_worked_11_to_2}
\end{equation}
By the general ordered bilinear formula
\eqref{eq:maxwell_Q_compact_bilinear}, there is no loss contribution because
the second incoming first-level state has zero vacuum component,
$\langle0|q\rangle=0$.  Hence the gain result
\eqref{eq:maxwell_worked_gain_result} is already the complete $1+1$ block,
and the output lands in the symmetric-traceless quadratic sector,
\begin{equation}
\mathscr H_1\times\mathscr H_1
\longrightarrow
\mathscr H_{2,\mathrm{tl}}.
\end{equation}

This calculation displays all three operations in their simplest nontrivial form: the lifts transform each incoming vector, fusion combines the two legs, and the physical angular average projects the resulting quadratic tensor onto its rotationally allowed component.  It also makes the degree addition $1+1=2$ visible before the full Maxwell grading is assembled in the next section.  Later, moment adaptation will set the physical first-level state to zero, so the $1+1$ channel disappears from the compatible low-order hierarchy; its purpose here is to expose the mechanics of the vertex itself.

\section{Intrinsic grading of the Maxwell collision vertex}
\label{subsec:maxwell_grading}

The Fock levels $\mathscr H_r$ were introduced in Eq.~\eqref{eq:maxwell_symmetric_fock_def}.  In the occupation-number basis they are equivalently characterized by the total occupation number
\begin{equation}
 |\bm n|:=\sum_{\alpha=1}^{d}n_\alpha,
\qquad
\mathscr H_r
=
\operatorname{span}\{|\bm n\rangle:|\bm n|=r\}.
\label{eq:maxwell_total_fock_grading}
\end{equation}
Let
\begin{equation}
 P_r
 :=
 \sum_{|\bm n|=r}|\bm n\rangle\langle\bm n|
\label{eq:maxwell_sector_projector}
\end{equation}
be the projector onto $\mathscr H_r$.  For any state $|\chi\rangle$ we define
\begin{equation}
|\chi^{[r]}\rangle:=P_r|\chi\rangle,
\qquad
|\chi\rangle=\sum_{r=0}^{\infty}|\chi^{[r]}\rangle .
\label{eq:maxwell_state_grading_def}
\end{equation}
The bracketed superscript $[r]$ therefore denotes the component on Fock level $r$; it is not a perturbation order.

The factorization of the preceding subsection now gives the grading almost directly.  Equation~\eqref{eq:maxwell_Gamma_sector} shows that
\begin{equation}
 \Gamma(A):\mathscr H_p\longrightarrow\mathscr H_p,
\end{equation}
whereas Eq.~\eqref{eq:maxwell_fusion_number} gives
\begin{equation}
 \mathsf m:\mathscr H_p\otimes\mathscr H_q
 \longrightarrow\mathscr H_{p+q}.
\end{equation}
The fixed-normal basis formula \eqref{eq:maxwell_transformed_basis_fusion} makes this explicit: every outgoing ket $|\bm r+\bm s\rangle$ satisfies
\begin{equation}
 |\bm r+\bm s|=p+q.
\end{equation}
The angular average changes only the coefficients inside that Fock level.  The loss term also respects the same rule: it is nonzero on an ordered pair in $\mathscr H_p\times\mathscr H_q$ only when $q=0$, in which case its output belongs to $\mathscr H_p=\mathscr H_{p+q}$.  Hence
\begin{equation}
\boxed{
\widehat{\mathcal Q}(\mathscr H_p,\mathscr H_q)
\subseteq
\mathscr H_{p+q}.
}
\label{eq:maxwell_exact_grading}
\end{equation}

Equation~\eqref{eq:maxwell_exact_grading} answers the question posed in the
Introduction: the degree-addition rule is a property of the abstract
bilinear vertex itself.  Coordinate realizations can make the same rule
look quite different.  A particularly explicit comparison is provided by
the Laguerre--spherical spectral calculation of
Glangetas, Li, and Xu \cite{GlangetasLiXu2016Triangular}.  Relabeling
their radial index $n$ as $j$ to match the notation used here, their
nonlinear perturbation operator in the harmonic-oscillator basis
$\varphi_{j,\ell,m}$, whose oscillator degree is $2j+\ell$, has the
schematic form

\begin{equation}
\mathcal G_{\rm GLX}\!\left(
\varphi_{j,\ell,m},
\varphi_{\tilde j,\tilde\ell,\tilde m}
\right)
=
\sum_k
\mu_k\,
\varphi_{
j+\tilde j+k,\,
\ell+\tilde\ell-2k,\,
m+\tilde m},
\label{eq:maxwell_GLX_triangular_coordinate_rule}
\end{equation}
with the allowed range of $k$ fixed by the angular-momentum coupling.  Every
term on the right satisfies
\begin{align}
2(j+\tilde j+k)+\ell+\tilde\ell-2k
&=
(2j+\ell)+(2\tilde j+\tilde\ell).
\label{eq:maxwell_GLX_degree_addition}
\end{align}
Thus their triangular selection rule is precisely the radial--angular
coordinate image of total-degree addition.  The shift by $k$ redistributes
pairs between radial and angular labels while leaving the total Fock degree
unchanged.  Earlier Fourier and moment formulations reveal the same
structure in different forms: in Bobylev variables it is already latent in
the homogeneity of the two collision arguments, while in ordinary and
Sonine moment coordinates it appears as successive solvability
\cite{Bobylev1975Fourier,Bobylev1984ExactRelaxation,Bobylev1988MaxwellReview,Ernst1979MaxwellMoments}.
In the abstract Fock representation no radial--angular recoupling is needed
to establish the selection rule: degree preservation by the lifts and
degree addition by fusion prove it before any coordinate realization is
chosen.

Equivalently, the matrix elements introduced in Eq.~\eqref{eq:maxwell_Q_matrix_from_Bobylev} obey
\begin{equation}
\mathcal Q_{\bm n;\bm m\bm l}=0
\qquad\text{unless}\qquad
|\bm n|=|\bm m|+|\bm l|.
\label{eq:maxwell_selection_rule}
\end{equation}

It is useful to name the restriction of the bilinear collision map to a pair of incoming levels,
\begin{equation}
\widehat{\mathcal Q}_{p,q}
:=
\widehat{\mathcal Q}\big|_{\mathscr H_p\times\mathscr H_q}:
\mathscr H_p\times\mathscr H_q
\longrightarrow
\mathscr H_{p+q}.
\label{eq:maxwell_sector_collision_block}
\end{equation}
This restriction has an explicit operational form.  For $|\phi_p\rangle\in\mathscr H_p$ and $|\psi_q\rangle\in\mathscr H_q$, define
\begin{equation}
|g^{(p,q)}_{\bm\omega}[\phi_p,\psi_q]\rangle
:=
\mathsf m\!\left(
\Gamma(A_{\bm\omega})|\phi_p\rangle
\otimes
\Gamma(B_{\bm\omega})|\psi_q\rangle
\right).
\label{eq:maxwell_sector_gain_family}
\end{equation}
Denote this $\bm\omega$-indexed family by $g^{(p,q)}_{\phi_p,\psi_q}$.  Then
\begin{equation}
\widehat{\mathcal Q}_{p,q}(|\phi_p\rangle,|\psi_q\rangle)
=
\mathfrak A_b[g^{(p,q)}_{\phi_p,\psi_q}]
-
\beta_0\,\delta_{q0}\,\langle0|\psi_q\rangle\,|\phi_p\rangle .
\label{eq:maxwell_sector_collision_block_explicit}
\end{equation}
The Kronecker factor makes the ordered loss term explicit: a homogeneous state of positive degree has zero Bargmann value at the origin, so the loss survives only on the $q=0$ incoming leg.  In the number basis, for $|\bm m|=p$ and $|\bm l|=q$,
\begin{equation}
\widehat{\mathcal Q}_{p,q}(|\bm m\rangle,|\bm l\rangle)
=
\sum_{|\bm n|=p+q}
\mathcal Q_{\bm n;\bm m\bm l}|\bm n\rangle,
\label{eq:maxwell_sector_collision_block_basis}
\end{equation}
with $\mathcal Q_{\bm n;\bm m\bm l}$ computed from Eqs.~\eqref{eq:maxwell_gain_coefficients_explicit} and \eqref{eq:maxwell_Q_matrix_from_Bobylev}.  Thus $\widehat{\mathcal Q}_{p,q}$ is fully specified rather than merely named as a formal graded contribution.

Define the outgoing collision component of total degree $r$ by
\begin{equation}
|\mathcal Q^{[r]}[\chi]\rangle
:=
P_r\widehat{\mathcal Q}(|\chi\rangle,|\chi\rangle).
\label{eq:maxwell_collision_output_level_def}
\end{equation}
Bilinearity together with Eq.~\eqref{eq:maxwell_exact_grading} gives the exact hierarchy
\begin{equation}
|\mathcal Q^{[r]}[\chi]\rangle
=
\sum_{p+q=r}
\widehat{\mathcal Q}_{p,q}
\big(|\chi^{[p]}\rangle,|\chi^{[q]}\rangle\big).
\label{eq:maxwell_graded_hierarchy}
\end{equation}
For the first levels this reads, before imposing any moment matching or conservation law,
\begin{align}
|\mathcal Q^{[0]}[\chi]\rangle
&=\widehat{\mathcal Q}_{0,0}(|\chi^{[0]}\rangle,|\chi^{[0]}\rangle),
\\
|\mathcal Q^{[1]}[\chi]\rangle
&=\widehat{\mathcal Q}_{0,1}(|\chi^{[0]}\rangle,|\chi^{[1]}\rangle)
 +\widehat{\mathcal Q}_{1,0}(|\chi^{[1]}\rangle,|\chi^{[0]}\rangle),
\\
|\mathcal Q^{[2]}[\chi]\rangle
&=\widehat{\mathcal Q}_{0,2}(|\chi^{[0]}\rangle,|\chi^{[2]}\rangle)
 +\widehat{\mathcal Q}_{1,1}(|\chi^{[1]}\rangle,|\chi^{[1]}\rangle)
 +\widehat{\mathcal Q}_{2,0}(|\chi^{[2]}\rangle,|\chi^{[0]}\rangle),
\\
|\mathcal Q^{[3]}[\chi]\rangle
&=\widehat{\mathcal Q}_{0,3}(|\chi^{[0]}\rangle,|\chi^{[3]}\rangle)
 +\widehat{\mathcal Q}_{1,2}(|\chi^{[1]}\rangle,|\chi^{[2]}\rangle)
 +\widehat{\mathcal Q}_{2,1}(|\chi^{[2]}\rangle,|\chi^{[1]}\rangle)
 +\widehat{\mathcal Q}_{3,0}(|\chi^{[3]}\rangle,|\chi^{[0]}\rangle),
\\
|\mathcal Q^{[4]}[\chi]\rangle
&=\widehat{\mathcal Q}_{0,4}(|\chi^{[0]}\rangle,|\chi^{[4]}\rangle)
 +\widehat{\mathcal Q}_{1,3}(|\chi^{[1]}\rangle,|\chi^{[3]}\rangle)
 +\widehat{\mathcal Q}_{2,2}(|\chi^{[2]}\rangle,|\chi^{[2]}\rangle)
\nonumber\\
&\quad
 +\widehat{\mathcal Q}_{3,1}(|\chi^{[3]}\rangle,|\chi^{[1]}\rangle)
 +\widehat{\mathcal Q}_{4,0}(|\chi^{[4]}\rangle,|\chi^{[0]}\rangle).
\label{eq:maxwell_first_levels_unreduced}
\end{align}
At this point these are purely algebraic statements about the collision vertex.  No compatibility condition has been imposed and no conserved component has been removed.

\section{One-vacuum spectrum, invariants, and compatible hierarchy}
\label{subsec:maxwell_linearized_spectrum}

\subsection{Linearized Fock vertex and the complete Maxwell collision spectrum}

The factorized collision vertex contains considerably more spectral
information than the two low-order relaxation rates required later by the
Navier--Stokes--Fourier calculation.  Its one-vacuum restriction yields
the complete spectrum of the Maxwell-molecule collision operator
linearized about the equilibrium state $n|0\rangle$.  The classical
spectrum originates in the work of Wang Chang and Uhlenbeck
\cite{WangChangUhlenbeck1952,WangChangUhlenbeck1970}; later explicit and
modern formulations include Refs.~\cite{AltermanFrankowskiPekeris1962,LernerMorimotoPravdaStarovXu2013}.

We emphasize that the linearization considered here is not an additional
kinetic model.  It is the exact first variation of the bilinear Fock
vertex about the equilibrium state $n|0\rangle$.  For an arbitrary Fock
perturbation $|\phi\rangle$, define
\begin{equation}
\widehat{\mathcal L}|\phi\rangle
:=
\widehat{\mathcal Q}(n|0\rangle,|\phi\rangle)
+
\widehat{\mathcal Q}(|\phi\rangle,n|0\rangle).
\label{eq:maxwell_full_linearized_Fock_operator}
\end{equation}
Thus, if
\begin{equation}
|\chi\rangle=n|0\rangle+\varepsilon|\phi\rangle,
\end{equation}
then
\begin{equation}
\widehat{\mathcal Q}(|\chi\rangle,|\chi\rangle)
=
\varepsilon\widehat{\mathcal L}|\phi\rangle
+
O(\varepsilon^2).
\end{equation}

The exact grading of the nonlinear vertex immediately implies
\begin{equation}
\widehat{\mathcal L}\mathscr H_r
\subseteq
\mathscr H_r,
\label{eq:maxwell_linearized_preserves_grade}
\end{equation}
so that
\begin{equation}
\widehat{\mathcal L}
=
\bigoplus_{r=0}^{\infty}\widehat{\mathcal L}_r,
\qquad
\widehat{\mathcal L}_r
:=
P_r\widehat{\mathcal L}P_r .
\end{equation}
For $r>0$ this agrees with the one-vacuum block
\begin{equation}
\widehat{\mathcal L}_r|\phi^{[r]}\rangle
=
\widehat{\mathcal Q}(n|0\rangle,|\phi^{[r]}\rangle)
+
\widehat{\mathcal Q}(|\phi^{[r]}\rangle,n|0\rangle).
\label{eq:maxwell_exact_one_vacuum_block}
\end{equation}

Let
\begin{equation}
\Phi(\bm z):=(\mathcal R_B|\phi\rangle)(\bm z)
\end{equation}
be the Bargmann symbol of the perturbation.  Since $\mathcal R_B|0\rangle=1$, nonlinear intertwining applied separately to the two ordered one-vacuum terms gives
\begin{align}
\mathcal R_B\widehat{\mathcal Q}(n|0\rangle,|\phi\rangle)
&=
\mathcal Q_B(n,\Phi)
=
n\int_{S^{d-1}}b(\eta)
\bigl[\Phi(\bm z_-)-\Phi(\bm0)\bigr] \,d\bm\sigma,
\label{eq:maxwell_linearized_left_vacuum_Bargmann}
\\
\mathcal R_B\widehat{\mathcal Q}(|\phi\rangle,n|0\rangle)
&=
\mathcal Q_B(\Phi,n)
=
n\int_{S^{d-1}}b(\eta)
\bigl[\Phi(\bm z_+)-\Phi(\bm z)\bigr] \,d\bm\sigma.
\label{eq:maxwell_linearized_right_vacuum_Bargmann}
\end{align}
Here the first gain term places the perturbation on the second incoming leg, while the second places it on the first; the two ordered loss terms give vacuum evaluation and the identity contribution, respectively.  Adding Eqs.~\eqref{eq:maxwell_linearized_left_vacuum_Bargmann} and \eqref{eq:maxwell_linearized_right_vacuum_Bargmann} therefore yields
\begin{equation}
\begin{aligned}
(\mathcal R_B\widehat{\mathcal L}|\phi\rangle)(\bm z)
=
n\int_{S^{d-1}}b(\eta)
\Big[
&
\Phi(\bm z_+)
+
\Phi(\bm z_-)
-\Phi(\bm z)
-\Phi(\bm0)
\Big] \,d\bm\sigma ,
\end{aligned}
\label{eq:maxwell_linearized_Bargmann_action}
\end{equation}
where
\begin{equation}
\bm z_\pm
=
\frac12
\left(
\bm z\pm|\bm z|\bm\sigma
\right),
\qquad
\eta
=
\frac{\bm z}{|\bm z|}\cdot\bm\sigma .
\end{equation}
Thus the familiar four-term linearized Maxwell expression is the Bargmann image of the first variation of the abstract bilinear Fock vertex.  Its two gain terms arise from the two possible placements of the perturbation on the incoming legs, and its two loss terms arise from the corresponding ordered loss contributions.

\paragraph{Irreducible decomposition of a Fock level.}

In three dimensions the homogeneous Fock level
\begin{equation}
\mathscr H_r
\simeq
\operatorname{Sym}^r(\mathbb R^3)
\end{equation}
has the multiplicity-free rotational decomposition
\begin{equation}
\mathscr H_r
=
\bigoplus_{\substack{j,\ell\geq0\\2j+\ell=r}}
\mathscr K_{j\ell},
\label{eq:maxwell_Fock_SO3_decomposition}
\end{equation}
where $\mathscr K_{j\ell}$ carries the irreducible angular-momentum
representation of rank $\ell$.

To make this decomposition explicit, let
$\mathcal Y_{\ell m}(\bm z)$,
$-\ell\leq m\leq\ell$, be homogeneous solid harmonics,
\begin{equation}
\mathcal Y_{\ell m}(\lambda\bm z)
=
\lambda^\ell\mathcal Y_{\ell m}(\bm z),
\qquad
\Delta_{\bm z}\mathcal Y_{\ell m}=0.
\end{equation}
For
\begin{equation}
r=2j+\ell
\end{equation}
define the homogeneous Bargmann polynomials
\begin{equation}
\Phi_{j\ell m}(\bm z)
=
(\bm z^2)^j
\mathcal Y_{\ell m}(\bm z).
\label{eq:maxwell_Fock_spherical_Bargmann_basis}
\end{equation}
Their abstract Fock representatives may be written, up to normalization,
as
\begin{equation}
|j,\ell,m\rangle
\propto
\left[
(\hat a_\alpha^\dagger)^2
\right]^j
\mathcal Y_{\ell m}(\hat{\bm a}^\dagger)|0\rangle .
\label{eq:maxwell_Fock_spherical_states}
\end{equation}
The integer $j$ therefore counts trace, or radial-pair, excitations,
whereas $\ell$ is the rank of the irreducible symmetric trace-free (STF) tensor sector, that is, the sector of totally symmetric tensors whose contraction over any pair of indices vanishes.  These labels belong to the intrinsic Fock-level decomposition; no coordinate adaptation of the Fock space is involved.  Their relation to Cartesian Hermite products and to the conventional Laguerre--spherical coordinate representation is derived explicitly in Appendix~\ref{app:maxwell_Hermite_Laguerre}.

Because the collision kernel is rotationally invariant,
$\widehat{\mathcal L}$ commutes with the natural action of $SO(3)$.
Each irreducible $\ell$-representation occurs only once in a fixed
Fock level $r$.  Hence
$\widehat{\mathcal L}$ acts as a scalar on every
$\mathscr K_{j\ell}$:
\begin{equation}
\widehat{\mathcal L}|j,\ell,m\rangle
=
-\lambda_{j\ell}|j,\ell,m\rangle,
\qquad
-\ell\leq m\leq\ell .
\label{eq:maxwell_full_spectrum_eigenvalue_def}
\end{equation}
The eigenvalue is independent of $m$, and the corresponding degeneracy is
$2\ell+1$.

\begin{proposition}[Complete linearized Maxwell spectrum in Fock space]
\label{prop:maxwell_complete_Fock_spectrum}
For the angular-cutoff Maxwell kernel used here, the eigenvalue of the
linearized Fock collision generator on $\mathscr K_{j\ell}$ is
\begin{equation}
\boxed{
\begin{aligned}
\lambda_{j\ell}
=
n\int_{S^2}b(\eta)
\Big[
&
1+\delta_{j0}\delta_{\ell0}
-c(\eta)^{\,2j+\ell}
P_\ell\!\bigl(c(\eta)\bigr)
-s(\eta)^{\,2j+\ell}
P_\ell\!\bigl(s(\eta)\bigr)
\Big]\,d\bm\sigma ,
\end{aligned}
}
\label{eq:maxwell_complete_Fock_spectrum}
\end{equation}
where
\begin{equation}
c(\eta):=\sqrt{\frac{1+\eta}{2}},
\qquad
s(\eta):=\sqrt{\frac{1-\eta}{2}},
\label{eq:maxwell_half_angle_cs}
\end{equation}
and $P_\ell$ is the Legendre polynomial of degree $\ell$.
\end{proposition}

\begin{proof}
Since the restriction of
$\widehat{\mathcal L}$ to $\mathscr K_{j\ell}$ is already known to be
scalar, it is sufficient to evaluate it on one nonzero member of the
irreducible sector.

Choose a unit vector $\bm e$ and the zonal solid harmonic
\begin{equation}
\mathcal Y_\ell^{(\bm e)}(\bm z)
=
|\bm z|^\ell
P_\ell
\left(
\frac{\bm e\cdot\bm z}{|\bm z|}
\right),
\end{equation}
and define
\begin{equation}
\Phi_{j\ell}^{(\bm e)}(\bm z)
=
(\bm z^2)^j
\mathcal Y_\ell^{(\bm e)}(\bm z).
\end{equation}
Evaluate the result at
\begin{equation}
\bm z=\rho\bm e,
\qquad
\rho>0.
\end{equation}
Then
\begin{equation}
\Phi_{j\ell}^{(\bm e)}(\rho\bm e)
=
\rho^{2j+\ell}.
\end{equation}

For
\begin{equation}
\eta=\bm e\cdot\bm\sigma
\end{equation}
the two Bobylev arguments are
\begin{equation}
\bm z_\pm
=
\frac{\rho}{2}
(\bm e\pm\bm\sigma).
\end{equation}
Their lengths satisfy
\begin{equation}
|\bm z_+|
=
\rho c(\eta),
\qquad
|\bm z_-|
=
\rho s(\eta),
\end{equation}
while their directions obey
\begin{equation}
\bm e\cdot\frac{\bm z_+}{|\bm z_+|}
=
c(\eta),
\qquad
\bm e\cdot\frac{\bm z_-}{|\bm z_-|}
=
s(\eta).
\end{equation}
Consequently,
\begin{equation}
\frac{
\Phi_{j\ell}^{(\bm e)}(\bm z_+)
}{
\Phi_{j\ell}^{(\bm e)}(\rho\bm e)
}
=
c(\eta)^{\,2j+\ell}
P_\ell\!\bigl(c(\eta)\bigr),
\end{equation}
and
\begin{equation}
\frac{
\Phi_{j\ell}^{(\bm e)}(\bm z_-)
}{
\Phi_{j\ell}^{(\bm e)}(\rho\bm e)
}
=
s(\eta)^{\,2j+\ell}
P_\ell\!\bigl(s(\eta)\bigr).
\end{equation}

For positive total degree $2j+\ell>0$,
$\Phi_{j\ell}^{(\bm e)}(\bm0)=0$.
For the vacuum sector $(j,\ell)=(0,0)$,
$\Phi_{000}=1$ and the additional loss term
$-\Phi(\bm0)$ contributes a second unit.  Substitution into
Eq.~\eqref{eq:maxwell_linearized_Bargmann_action} therefore gives
Eq.~\eqref{eq:maxwell_complete_Fock_spectrum}.
\end{proof}

The spectral labels have a direct Fock interpretation:
\begin{equation}
r=2j+\ell,
\qquad
j=\text{number of trace pairs},
\qquad
\ell=\text{irreducible angular rank}.
\label{eq:maxwell_spectrum_Fock_labels}
\end{equation}
Thus the conventional radial and angular spectral quantum numbers are
already encoded in the trace decomposition of a homogeneous Fock level.

\paragraph{Relation to the standard Maxwell spectrum.}

The irreducible labels $(j,\ell,m)$ were introduced intrinsically in Fock
space through Eqs.~\eqref{eq:maxwell_Fock_SO3_decomposition}--\eqref{eq:maxwell_Fock_spherical_states}.  Their appearance in the conventional
Maxwell spectrum is a statement about the Hermite coordinate realization,
not an additional adaptation of the abstract Fock space.  Appendix~\ref{app:maxwell_Hermite_Laguerre} derives this coordinate transformation explicitly, starting from the Cartesian Hermite products of Eq.~\eqref{eq:maxwell_H_basis} and the polynomial ladder operators of Eq.~\eqref{eq:maxwell_ladder_P}.  The result is, up to normalization,
\begin{equation}
\mathcal R_H|j,\ell,m\rangle
\propto
\sqrt W\,
\left(\frac{|\bm\xi|}{\sqrt2}\right)^\ell
L_j^{(\ell+1/2)}
\left(\frac{|\bm\xi|^2}{2}\right)
Y_\ell^m
\left(
\frac{\bm\xi}{|\bm\xi|}
\right),
\label{eq:maxwell_Fock_to_Laguerre_spherical}
\end{equation}
where $L_j^{(\ell+1/2)}$ is the generalized Laguerre polynomial.  Thus the
familiar Laguerre--spherical functions are the radial--angular Hermite
coordinate images of the same abstract Fock states.  In this sense the
intrinsic trace/STF decomposition
\eqref{eq:maxwell_Fock_SO3_decomposition} is the representation-independent
origin of the usual eigenfunction labels of the linearized Maxwell
collision operator.

To compare Eq.~\eqref{eq:maxwell_complete_Fock_spectrum} with the conventional half-angle form, first make the angular symmetrization explicit.  Since the expression in square brackets in Eq.~\eqref{eq:maxwell_complete_Fock_spectrum} is even under $\eta\mapsto-\eta$, define, for $0\leq\eta\leq1$,
\begin{equation}
b_+(\eta):=b(\eta)+b(-\eta).
\label{eq:maxwell_symmetrized_angular_kernel}
\end{equation}
Then set
\begin{equation}
\eta=\cos2\vartheta,
\qquad
c=\cos\vartheta,
\qquad
s=\sin\vartheta,
\qquad
0\leq\vartheta\leq\frac{\pi}{4}.
\end{equation}
Using $d\bm\sigma=2\pi\,d\eta$ for a zonal integral on $S^2$ and $d\eta=-2\sin(2\vartheta)\,d\vartheta$, Eq.~\eqref{eq:maxwell_complete_Fock_spectrum} becomes
\begin{equation}
\begin{aligned}
\lambda_{j\ell}
={}&
4\pi n
\int_0^{\pi/4}
b_+(\cos2\vartheta)\sin(2\vartheta)
\Big[
1+\delta_{j0}\delta_{\ell0}
-
P_\ell(\cos\vartheta)
(\cos\vartheta)^{2j+\ell}
-
P_\ell(\sin\vartheta)
(\sin\vartheta)^{2j+\ell}
\Big]
\,d\vartheta .
\end{aligned}
\label{eq:maxwell_standard_spectrum_from_Fock}
\end{equation}
This is the standard Maxwell-molecule spectrum, here obtained directly
from the factorized Fock collision vertex.  The explicit definition
\eqref{eq:maxwell_symmetrized_angular_kernel} fixes the overall angular convention; equivalently, one may use the averaged kernel $b_{\rm sym}=b_+/2$, in which case the prefactor in Eq.~\eqref{eq:maxwell_standard_spectrum_from_Fock} is $8\pi n$.  The factor $n$ reflects the local equilibrium density used in the present normalization.

\paragraph{Collision invariants as zero modes.}

The null space follows directly from
Eq.~\eqref{eq:maxwell_complete_Fock_spectrum}.  For density,
\begin{equation}
(j,\ell)=(0,0),
\end{equation}
and the integrand vanishes identically:
\begin{equation}
1+1-P_0(c)-P_0(s)=0.
\end{equation}
For momentum,
\begin{equation}
(j,\ell)=(0,1),
\end{equation}
and, since $P_1(x)=x$,
\begin{equation}
1-cP_1(c)-sP_1(s)
=
1-c^2-s^2
=
0.
\end{equation}
Finally, the scalar energy mode corresponds to
\begin{equation}
(j,\ell)=(1,0),
\end{equation}
for which again
\begin{equation}
1-c^2-s^2=0.
\end{equation}
Hence
\begin{equation}
\lambda_{00}
=
\lambda_{01}
=
\lambda_{10}
=
0.
\label{eq:maxwell_spectral_zero_modes}
\end{equation}
Their degeneracies $1$, $3$, and $1$ give the scalar, vector, and
scalar-quadratic zero modes.  These are precisely the modes that will be
identified below, independently through moment intertwining, with mass,
momentum, and energy conservation.

\paragraph{Hydrodynamic eigenvalues as spectral entries.}

The two collision rates required by the first Chapman--Enskog
approximation are now immediate special cases of the complete spectrum.

We retain the angular moment $\Omega_2$ defined in Eq.~\eqref{eq:maxwell_Omega2}; $\beta_0$ was defined in Eq.~\eqref{eq:maxwell_beta0_def}.

The traceless stress sector is
\begin{equation}
\mathscr H_{2,\mathrm{tl}}
=
\mathscr K_{02}.
\end{equation}
Using
\begin{equation}
P_2(x)=\frac12(3x^2-1)
\end{equation}
and
\begin{equation}
c^2P_2(c)+s^2P_2(s)
=
\frac14(1+3\eta^2),
\end{equation}
Eq.~\eqref{eq:maxwell_complete_Fock_spectrum} gives
\begin{equation}
\widehat{\mathcal L}_2
P_{2,\mathrm{tl}}
=
-\lambda_\pi P_{2,\mathrm{tl}},
\qquad
\lambda_\pi
=
\lambda_{02}
=
\frac34n\Omega_2.
\label{eq:maxwell_lambda_pi}
\end{equation}

The heat-flux vector sector is
\begin{equation}
\mathscr H_{3,\mathrm{vec}}
=
\mathscr K_{11}.
\end{equation}
Since $P_1(x)=x$,
\begin{equation}
c^3P_1(c)+s^3P_1(s)
=
c^4+s^4
=
\frac12(1+\eta^2),
\end{equation}
and therefore
\begin{equation}
\widehat{\mathcal L}_3
P_{3,\mathrm{vec}}
=
-\lambda_q P_{3,\mathrm{vec}},
\qquad
\lambda_q
=
\lambda_{11}
=
\frac12n\Omega_2.
\label{eq:maxwell_lambda_q}
\end{equation}
Consequently,
\begin{equation}
\frac{\lambda_q}{\lambda_\pi}
=
\frac{\lambda_{11}}{\lambda_{02}}
=
\frac23.
\label{eq:maxwell_eigenvalue_ratio}
\end{equation}

The hydrodynamic relaxation rates are thus not isolated low-order
calculations.  They are two members of the complete spectral family
generated by the same factorized Fock vertex:
\begin{equation}
\text{factorized nonlinear vertex}
\;\longrightarrow\;
\text{exact one-vacuum operator}
\;\longrightarrow\;
\{\lambda_{j\ell}\}
\;\longrightarrow\;
\lambda_\pi,\lambda_q.
\label{eq:maxwell_spectrum_to_hydrodynamics}
\end{equation}
The remaining irreducible component of the cubic level is
$\mathscr K_{03}$; it possesses its own eigenvalue
$\lambda_{03}$ but is not selected by the first-order propagation
source.


\subsection{Moment intertwining, adaptation, and collision invariants}
\label{subsec:maxwell_moment_compatibility}
\label{subsec:maxwell_collision_invariants}

For a polynomial observable $m(\bm v)$, the reconstruction
$f=W\mathcal R_P|\chi\rangle$ pulls the physical moment functional back to
\begin{equation}
 \mathcal M_m[f]
 :=\int_{\mathbb R^d}m(\bm v)f(\bm v)\,d^dv
 =\langle0|m(\widehat{\bm v})|\chi\rangle.
\label{eq:maxwell_moment_intertwining}
\end{equation}
Here velocity multiplication is represented by
\begin{equation}
\widehat v_\alpha
=u_\alpha I+\sqrt\theta(\hat a_\alpha+\hat a_\alpha^\dagger).
\label{eq:maxwell_vhat}
\end{equation}
Thus the density, momentum, and energy moments determine a moment-adapted local frame by
\begin{equation}
 n=n_f,
 \qquad
 \bm u=\bm u_f,
 \qquad
 \theta=\theta_f,
\label{eq:maxwell_moment_adapted_frame_def}
\end{equation}
which is equivalent to
\begin{equation}
\chi_{\bm0}=n,
\qquad
\chi_{\bm e_\alpha}=0,
\qquad
\sum_{\alpha=1}^{d}\chi_{2\bm e_\alpha}=0,
\label{eq:maxwell_compatibility_chi}
\end{equation}
or, intrinsically,
\begin{equation}
 |\chi^{[0]}\rangle=n|0\rangle,
 \qquad
 |\chi^{[1]}\rangle=0,
 \qquad
 \langle0|\hat a_\alpha\hat a_\alpha|\chi^{[2]}\rangle=0.
\label{eq:maxwell_compatibility_graded}
\end{equation}
These are compatibility identities produced by moment matching, not restrictions imposed on the kinetic state.

The same intertwining transfers the standard elastic collision invariants directly to Fock space.  Since
\begin{equation}
 \int m(\bm v)Q(f,f)\,d^dv=0,
 \qquad m=1,v_\alpha,\bm v^2,
\label{eq:maxwell_physical_collision_invariants}
\end{equation}
and
\begin{equation}
 \int m(\bm v)Q(f,f)\,d^dv
 =\langle0|m(\widehat{\bm v})
 \widehat{\mathcal Q}(|\chi\rangle,|\chi\rangle)\rangle,
\label{eq:maxwell_collision_rate_intertwining}
\end{equation}
the centered moment-adapted forms are simply
\begin{equation}
\begin{aligned}
\langle0|\widehat{\mathcal Q}(|\chi\rangle,|\chi\rangle)\rangle&=0,\\
\langle0|\hat a_\alpha\widehat{\mathcal Q}(|\chi\rangle,|\chi\rangle)\rangle&=0,\\
\langle0|\hat a_\alpha\hat a_\alpha
\widehat{\mathcal Q}(|\chi\rangle,|\chi\rangle)\rangle&=0.
\end{aligned}
\label{eq:maxwell_collision_invariants_summary}
\end{equation}
Hence the collision dynamics preserves the vacuum, vector, and scalar-trace compatibility directions.  This is all that will be needed below from the conservation laws.

\subsection{Compatible low-order collision hierarchy}
\label{subsec:maxwell_compatible_low_order}

We can now combine three results that have been established independently: the exact grading of the bilinear collision vertex, the moment-adapted compatibility identities of the incoming state, and the collision invariants of the outgoing state.  Compatibility gives $|\chi^{[1]}\rangle=0$, so the unreduced hierarchy \eqref{eq:maxwell_first_levels_unreduced} becomes
\begin{align}
|\mathcal Q^{[2]}[\chi]\rangle
&=
\widehat{\mathcal Q}_{0,2}(|\chi^{[0]}\rangle,|\chi^{[2]}\rangle)
+
\widehat{\mathcal Q}_{2,0}(|\chi^{[2]}\rangle,|\chi^{[0]}\rangle),
\label{eq:maxwell_Q2_compatible_structural}
\\
|\mathcal Q^{[3]}[\chi]\rangle
&=
\widehat{\mathcal Q}_{0,3}(|\chi^{[0]}\rangle,|\chi^{[3]}\rangle)
+
\widehat{\mathcal Q}_{3,0}(|\chi^{[3]}\rangle,|\chi^{[0]}\rangle),
\label{eq:maxwell_Q3_compatible_structural}
\end{align}
while
\begin{align}
|\mathcal Q^{[4]}[\chi]\rangle
={}&
\widehat{\mathcal Q}_{0,4}(|\chi^{[0]}\rangle,|\chi^{[4]}\rangle)
+
\widehat{\mathcal Q}_{2,2}(|\chi^{[2]}\rangle,|\chi^{[2]}\rangle)
+
\widehat{\mathcal Q}_{4,0}(|\chi^{[4]}\rangle,|\chi^{[0]}\rangle).
\label{eq:maxwell_Q4_compatible_structural}
\end{align}
Because $|\chi^{[0]}\rangle=n|0\rangle$, the first two formulas are linear in the nonequilibrium amplitudes $|\chi^{[2]}\rangle$ and $|\chi^{[3]}\rangle$.  No separate linearization of the full nonlinear collision integral has been used.  At level four, by contrast, the term
\begin{equation}
\widehat{\mathcal Q}_{2,2}:
\mathscr H_2\times\mathscr H_2
\longrightarrow
\mathscr H_4
\label{eq:maxwell_first_nonlinear_vertex}
\end{equation}
is the first genuinely nonlinear coupling between compatible nonequilibrium sectors.  The collision invariants additionally give
\begin{equation}
 |\mathcal Q^{[0]}[\chi]\rangle=0,
 \qquad
 |\mathcal Q^{[1]}[\chi]\rangle=0,
 \qquad
 \operatorname{tr}|\mathcal Q^{[2]}[\chi]\rangle=0,
\end{equation}
showing that the collision dynamics is tangent to the moment-adapted manifold in its density, momentum, and scalar-energy directions.

The exact linearity of the Maxwell-molecule pressure-tensor and energy-flux production laws is classical, with its origin in Maxwell's moment method and its systematic development by Ikenberry and Truesdell \cite{Maxwell1867DynamicalTheory,IkenberryTruesdell1956I,Truesdell1956II,TruesdellMuncaster1980Maxwell}.  The present grading gives a compact structural explanation of that fact: after compatibility removes $\mathscr H_1$, the second and third collision outputs contain one vacuum leg and one nonequilibrium leg only.  The first genuinely nonlinear compatible nonequilibrium interaction is therefore the $2+2\to4$ channel.

\subsection{First nonlinear compatible block in physical variables}
\label{subsec:maxwell_first_nonlinear_physical}

The channel in Eq.~\eqref{eq:maxwell_first_nonlinear_vertex} has a direct physical interpretation.  We give it explicitly in three dimensions before passing to the general triangular hierarchy.  The moment intertwining \eqref{eq:maxwell_moment_intertwining} identifies the traceless second-degree Hermite coefficient with the physical nonequilibrium stress,
\begin{equation}
\pi_{\alpha\beta}
:=m\int c_{\langle\alpha}c_{\beta\rangle}f\,d\bm v
=m\theta\,
\left\langle0\left|\hat a_{\langle\alpha}\hat a_{\beta\rangle}\right|\chi\right\rangle .
\label{eq:maxwell_stress_Fock_pairing}
\end{equation}
Here and below angular brackets denote symmetric trace-free projection.  Compatibility removes the scalar trace at level two, so the complete quadratic state is therefore
\begin{equation}
 |\chi^{[2]}\rangle
 =\frac{1}{2m\theta}\,
 \pi_{\alpha\beta}\hat a_\alpha^\dagger\hat a_\beta^\dagger|0\rangle,
 \qquad
 \pi_{\alpha\alpha}=0,
\label{eq:maxwell_stress_quadratic_state}
\end{equation}
and its Bargmann symbol is
\begin{equation}
 \Phi_2(\bm z)
 =\frac{1}{2m\theta}\pi_{\alpha\beta}z_\alpha z_\beta.
\label{eq:maxwell_stress_Bargmann_symbol}
\end{equation}
Thus $\pi_{\alpha\beta}/(m\theta)$ is the irreducible Hermite/Fock coefficient of the compatible level-$2$ state and the tensor coefficient of its Bargmann polynomial, while $\pi_{\alpha\beta}$ is the corresponding physical normalization.  For the contractions below, write
\begin{equation}
 (\pi^2)_{\alpha\beta}:=\pi_{\alpha\gamma}\pi_{\gamma\beta},
 \qquad
 \operatorname{tr}(\pi^2)=\pi_{\alpha\beta}\pi_{\alpha\beta}.
\label{eq:maxwell_pi_square_def}
\end{equation}

Since both incoming degrees are positive, the ordered loss term vanishes and
\begin{equation}
(\mathcal R_B\widehat{\mathcal Q}_{2,2})(\bm z)
=\int_{S^2}b(\eta)\,\Phi_2(\bm z_+)\Phi_2(\bm z_-)\,d\bm\sigma.
\label{eq:maxwell_Q22_Bargmann_start}
\end{equation}
Introduce
\begin{equation}
 \Omega_4:=\int_{S^2}b(\eta)(1-\eta^2)^2\,d\bm\sigma.
\label{eq:maxwell_Omega4}
\end{equation}
Evaluation of the axisymmetric fourth angular moment gives
\begin{equation}
\begin{aligned}
(\mathcal R_B\widehat{\mathcal Q}_{2,2})(\bm z)
={}&\frac{\Omega_4}{256m^2\theta^2}
 \operatorname{tr}(\pi^2)(\bm z^2)^2
-\frac{5\Omega_4}{128m^2\theta^2}\bm z^2\,
 (\pi^2)_{\alpha\beta}z_\alpha z_\beta\\
&+\frac{35\Omega_4-16\Omega_2}{512m^2\theta^2}
 (\pi_{\alpha\beta}z_\alpha z_\beta)^2.
\end{aligned}
\label{eq:maxwell_Q22_explicit}
\end{equation}
Rotational covariance is manifest: the three terms are precisely the scalar, rank-two, and rank-four pieces allowed in
$\operatorname{Sym}^2(\mathscr K_{02})$.

The corresponding level-$4$ coordinates are most naturally introduced as dimensionally scaled irreducible Hermite coefficients of $|\chi^{[4]}\rangle$.  Using the same moment intertwining as in Eq.~\eqref{eq:maxwell_stress_Fock_pairing}, they may be written simultaneously as Fock pairings and as their physical velocity-space representatives:
\begin{align}
M_{4|0}
&:=m\theta^2
\left\langle0\left|\hat a_\alpha\hat a_\alpha
\hat a_\beta\hat a_\beta\right|\chi\right\rangle
=m\int
\left(c^4-10\theta c^2+15\theta^2\right)f\,d\bm v,
\label{eq:maxwell_M40_def}\\
M_{2|\alpha\beta}
&:=m\theta^2
\left\langle0\left|\hat a_\gamma\hat a_\gamma
\hat a_{\langle\alpha}\hat a_{\beta\rangle}\right|\chi\right\rangle
=m\int
(c^2-7\theta)c_{\langle\alpha}c_{\beta\rangle}f\,d\bm v,
\label{eq:maxwell_M2ab_def}\\
M_{0|\alpha\beta\gamma\delta}
&:=m\theta^2
\left\langle0\left|
\hat a_{\langle\alpha}\hat a_\beta\hat a_\gamma\hat a_{\delta\rangle}
\right|\chi\right\rangle
=m\int
c_{\langle\alpha}c_\beta c_\gamma c_{\delta\rangle}f\,d\bm v.
\label{eq:maxwell_M0abcd_def}
\end{align}
They are respectively the $\mathscr K_{20}$, $\mathscr K_{12}$, and $\mathscr K_{04}$ components in
\begin{equation}
 \mathscr H_4
 =\mathscr K_{20}\oplus\mathscr K_{12}\oplus\mathscr K_{04}.
\label{eq:maxwell_H4_irreducible_split}
\end{equation}
The passage from Eq.~\eqref{eq:maxwell_Q22_explicit} to physical moment production is therefore simply projection of the quartic Bargmann/Fock coefficient onto these three Hermite bras.

For homogeneous collision dynamics, $n$, $\bm u$, and $\theta$ are fixed by the collision invariants.  Combining the one-vacuum relaxation with the $2+2\to4$ source gives the first closed nonlinear triangular subsystem,
\begin{align}
\partial_t\pi_{\alpha\beta}
&=-\lambda_{02}\pi_{\alpha\beta},
\label{eq:maxwell_stress_homogeneous_evolution}\\
\partial_t M_{4|0}
&=-\lambda_{20}M_{4|0}
-\frac{\Omega_2}{2m}\,\operatorname{tr}(\pi^2),
\label{eq:maxwell_M40_evolution}\\
\partial_t M_{2|\alpha\beta}
&=-\lambda_{12}M_{2|\alpha\beta}
-\frac{\Omega_2}{2m}
(\pi^2)_{\langle\alpha\beta\rangle},
\label{eq:maxwell_M2ab_evolution}\\
\partial_t M_{0|\alpha\beta\gamma\delta}
&=-\lambda_{04}M_{0|\alpha\beta\gamma\delta}
+\frac{3(35\Omega_4-16\Omega_2)}{64m}
\pi_{\langle\alpha\beta}\pi_{\gamma\delta\rangle}.
\label{eq:maxwell_M0abcd_evolution}
\end{align}
In the rank-two source we use the explicit three-dimensional STF contraction
\begin{equation}
(\pi^2)_{\langle\alpha\beta\rangle}
:=(\pi^2)_{\alpha\beta}
-\frac13\delta_{\alpha\beta}\,
\operatorname{tr}(\pi^2).
\label{eq:maxwell_pi_squared_STF}
\end{equation}
The relevant entries of the complete spectrum are
\begin{equation}
\lambda_{02}=\lambda_\pi=\frac34n\Omega_2,
\qquad
\lambda_{20}=\frac12n\Omega_2,
\qquad
\lambda_{12}=\frac78n\Omega_2,
\qquad
\lambda_{04}=\frac{7n}{64}(16\Omega_2-5\Omega_4).
\label{eq:maxwell_level4_rates}
\end{equation}
Thus the first nonlinear compatible interaction has the concrete content
\begin{equation}
\text{stress}\times\text{stress}
\longrightarrow
\text{fourth-order non-Gaussian Hermite moments}.
\label{eq:maxwell_stress_squared_to_fourth}
\end{equation}
The triangularity is now visible directly in physical variables: stress relaxes autonomously, its quadratic self-interaction forces the three irreducible fourth-order sectors, and level four cannot feed back into level two.  Equations~\eqref{eq:maxwell_stress_homogeneous_evolution}--\eqref{eq:maxwell_M0abcd_evolution} are the first explicit physical instance of the general strictly triangular structure proved next.

\section{Exact triangularity and nonlinear relaxation of finite Fock jets}
\label{subsec:maxwell_finite_jet_relaxation}

The compatible low-order hierarchy above is the beginning of a general
strictly triangular structure.  Triangular spectral systems for the
homogeneous Maxwell equation have been used previously, most explicitly in
the Laguerre--spherical construction of
Glangetas, Li, and Xu \cite{GlangetasLiXu2016Triangular}, where control of
the complete infinite hierarchy requires weighted estimates and a
smallness condition in the chosen function space.  Here we isolate a
different consequence of the intrinsic grading: at any fixed Fock order,
the hierarchy closes exactly from below.

For Maxwell molecules this gives a nonperturbative finite-dimensional
relaxation statement stronger than the exact linearity of the second- and
third-level production laws: every finite collection of compatible low
Fock levels forms an exact autonomous factor of the nonlinear collision
dynamics and relaxes globally to the corresponding Maxwellian jet.  No
smallness assumption is needed at fixed jet order.

We formulate this statement for the spatially homogeneous collision
dynamics in three dimensions, where the complete one-vacuum spectrum was
obtained in Sec.~\ref{subsec:maxwell_linearized_spectrum}.  The argument
itself uses only the exact grading, compatibility, and a positive spectral
gap on the non-invariant sectors.

\paragraph{Compatible perturbation and exact triangular hierarchy.}

In the homogeneous problem the physical density, mean velocity, and
temperature are constants of the collision dynamics.  Hence, once the
moment-adapted frame has been fixed, the corresponding Maxwellian vacuum
does not move in time.  Write
\begin{equation}
|\chi(t)\rangle
=
n|0\rangle+|h(t)\rangle .
\label{eq:maxwell_jet_perturbation}
\end{equation}
The moment-adapted compatibility conditions and the collision invariants
imply
\begin{equation}
|h^{[0]}(t)\rangle=0,
\qquad
|h^{[1]}(t)\rangle=0,
\label{eq:maxwell_jet_compatibility_01}
\end{equation}
while the scalar trace component of
$|h^{[2]}(t)\rangle$ vanishes.  Thus
\begin{equation}
|h^{[2]}(t)\rangle
\in
\mathscr H_{2,\mathrm{tl}}
=
\mathscr K_{02}.
\label{eq:maxwell_jet_compatibility_2}
\end{equation}

Since
\begin{equation}
\widehat{\mathcal Q}(n|0\rangle,n|0\rangle)=0,
\end{equation}
bilinearity gives
\begin{equation}
\frac{d}{dt}|h\rangle
=
\widehat{\mathcal L}|h\rangle
+
\widehat{\mathcal Q}(|h\rangle,|h\rangle),
\label{eq:maxwell_jet_full_perturbation_equation}
\end{equation}
where $\widehat{\mathcal L}$ is the exact one-vacuum operator defined in
Eq.~\eqref{eq:maxwell_full_linearized_Fock_operator}.

Projecting Eq.~\eqref{eq:maxwell_jet_full_perturbation_equation} onto
$\mathscr H_r$ and using the exact grading
\eqref{eq:maxwell_exact_grading} gives
\begin{equation}
\boxed{
\frac{d}{dt}|h^{[r]}\rangle
=
\widehat{\mathcal L}_r|h^{[r]}\rangle
+
\sum_{\substack{p+q=r\\p,q\geq2}}
\widehat{\mathcal Q}_{p,q}
\left(
|h^{[p]}\rangle,
|h^{[q]}\rangle
\right),
\qquad r\geq2.
}
\label{eq:maxwell_jet_exact_triangular_hierarchy}
\end{equation}
The vacuum contributions $0+r$ and $r+0$ have been collected into
$\widehat{\mathcal L}_r$, whereas all terms containing a first-level leg
vanish by Eq.~\eqref{eq:maxwell_jet_compatibility_01}.

The decisive observation is that every nonlinear term on the right-hand
side of Eq.~\eqref{eq:maxwell_jet_exact_triangular_hierarchy} involves
strictly lower levels.  Indeed,
\begin{equation}
p+q=r,
\qquad
p,q\geq2
\end{equation}
implies
\begin{equation}
p,q\leq r-2.
\label{eq:maxwell_jet_strict_triangularity}
\end{equation}
Thus the compatible Maxwell hierarchy is not merely graded; after the
vacuum contribution is separated, it is strictly triangular with respect
to the total Fock degree.

For example,
\begin{align}
\frac{d}{dt}|h^{[2]}\rangle
&=
\widehat{\mathcal L}_2|h^{[2]}\rangle,
\label{eq:maxwell_jet_level2}
\\
\frac{d}{dt}|h^{[3]}\rangle
&=
\widehat{\mathcal L}_3|h^{[3]}\rangle,
\label{eq:maxwell_jet_level3}
\\
\frac{d}{dt}|h^{[4]}\rangle
&=
\widehat{\mathcal L}_4|h^{[4]}\rangle
+
\widehat{\mathcal Q}_{2,2}
\left(
|h^{[2]}\rangle,
|h^{[2]}\rangle
\right),
\label{eq:maxwell_jet_level4}
\\
\frac{d}{dt}|h^{[5]}\rangle
&=
\widehat{\mathcal L}_5|h^{[5]}\rangle
+
\widehat{\mathcal Q}_{2,3}
\left(
|h^{[2]}\rangle,
|h^{[3]}\rangle
\right)
+
\widehat{\mathcal Q}_{3,2}
\left(
|h^{[3]}\rangle,
|h^{[2]}\rangle
\right).
\label{eq:maxwell_jet_level5}
\end{align}
Thus the $2+2\to4$ block identified in
Eq.~\eqref{eq:maxwell_first_nonlinear_vertex} is not only the first
nonlinear compatible channel; it is the first member of an entire
strictly upward nonlinear cascade.

\paragraph{Finite jets as exact autonomous factors.}

For $R\geq2$, define the compatible finite-jet space
\begin{equation}
\mathscr J_R
:=
\mathscr K_{02}
\oplus
\bigoplus_{r=3}^{R}\mathscr H_r .
\label{eq:maxwell_finite_jet_space}
\end{equation}
The quadratic level is written separately only because moment adaptation
removes its scalar trace component, leaving $\mathscr K_{02}$; all levels
of degree $r\geq3$ contain no collision-invariant directions and therefore
enter in full.

For a compatible state, define its $R$th jet by
\begin{equation}
J_Rh
:=
\left(
h^{[2]},
h^{[3]},
\ldots,
h^{[R]}
\right)
\in\mathscr J_R.
\label{eq:maxwell_finite_jet_def}
\end{equation}

It is important to distinguish this construction from an invariant
finite-dimensional truncation.  The subspace $\mathscr J_R$ is generally
not invariant under the full nonlinear collision map: interactions among
retained levels may generate levels larger than $R$.  However, by
Eq.~\eqref{eq:maxwell_exact_grading}, such higher levels can never feed
back into lower ones.  Consequently the projection of the collision
dynamics onto levels $r\leq R$ depends only on those levels.

More explicitly, if $\Pi_R$ denotes the compatible projection onto
$\mathscr J_R$, then for every compatible $|h\rangle$,
\begin{equation}
\Pi_R
\widehat{\mathcal Q}
\left(
n|0\rangle+|h\rangle,
n|0\rangle+|h\rangle
\right)
=
\Pi_R
\widehat{\mathcal Q}
\left(
n|0\rangle+\Pi_R|h\rangle,
n|0\rangle+\Pi_R|h\rangle
\right).
\label{eq:maxwell_finite_jet_exact_factor}
\end{equation}
Indeed, an incoming component of degree larger than $R$ can contribute
only to an outgoing component of at least the same degree and therefore
cannot enter the projection on the left-hand side.

Equation~\eqref{eq:maxwell_finite_jet_exact_factor} shows that a finite
Fock jet is an \emph{exact autonomous factor} of the Maxwell collision
hierarchy.  This is different from a moment closure: no assumption about
unretained higher moments is required to determine the evolution of the
retained jet.

\paragraph{Uniform linear relaxation on the compatible sectors.}

We next obtain a common positive decay rate for all compatible levels.
The irreducible decomposition
\eqref{eq:maxwell_Fock_SO3_decomposition} and the complete spectrum
\eqref{eq:maxwell_complete_Fock_spectrum} give
\begin{equation}
\widehat{\mathcal L}
|j,\ell,m\rangle
=
-\lambda_{j\ell}|j,\ell,m\rangle,
\qquad
r=2j+\ell.
\end{equation}
For positive degree,
\begin{equation}
\lambda_{j\ell}
=
n\int_{S^2}
b(\eta)
\left[
1
-
c(\eta)^rP_\ell(c(\eta))
-
s(\eta)^rP_\ell(s(\eta))
\right]
\,d\bm\sigma,
\label{eq:maxwell_jet_spectrum_positive_degree}
\end{equation}
where $c$ and $s$ are defined in
Eq.~\eqref{eq:maxwell_half_angle_cs} and satisfy
\begin{equation}
c^2+s^2=1,
\qquad
0\leq c,s\leq1.
\end{equation}

For $0\leq x\leq1$,
\begin{equation}
P_\ell(x)\leq1.
\end{equation}
Hence, for every level of degree $r\geq3$,
\begin{align}
\lambda_{j\ell}
&\geq
n\int_{S^2}
b(\eta)
\left[
1-c^r-s^r
\right]
\,d\bm\sigma
\nonumber\\
&\geq
n\int_{S^2}
b(\eta)
\left[
1-c^3-s^3
\right]
\,d\bm\sigma .
\label{eq:maxwell_jet_gap_bound}
\end{align}
Define
\begin{equation}
\lambda_{\ge 3}
:=
n\int_{S^2}
b(\eta)
\left[
1-c(\eta)^3-s(\eta)^3
\right]
\,d\bm\sigma .
\label{eq:maxwell_jet_lambda_ng}
\end{equation}
Because
\begin{equation}
c^3+s^3<1
\end{equation}
whenever $0<c,s<1$, one has
\begin{equation}
\lambda_{\ge 3}>0
\end{equation}
for any nontrivial nonnegative angular-cutoff kernel $b$.

The only compatible degree-two sector is the traceless stress sector
$\mathscr K_{02}$, whose relaxation rate was obtained in
Eq.~\eqref{eq:maxwell_lambda_pi}:
\begin{equation}
\lambda_\pi
=
\lambda_{02}
=
\frac34n\Omega_2>0.
\end{equation}
We may therefore introduce the common compatible relaxation rate
\begin{equation}
\lambda_*
:=
\min
\left\{
\lambda_\pi,
\lambda_{\ge 3}
\right\}
>0.
\label{eq:maxwell_jet_uniform_gap}
\end{equation}

Since the irreducible decomposition is orthogonal in the Fock metric and
$\widehat{\mathcal L}_r$ acts by the real scalars
$-\lambda_{j\ell}$ on its irreducible components,
Eq.~\eqref{eq:maxwell_jet_uniform_gap} gives
\begin{equation}
\left\|
e^{t\widehat{\mathcal L}_r}u
\right\|_{\mathscr H_r}
\leq
e^{-\lambda_*t}
\|u\|_{\mathscr H_r},
\qquad
t\geq0,
\label{eq:maxwell_jet_semigroup_bound}
\end{equation}
for every compatible $u$ in every level $r\geq2$.

Thus the decay rate used below can be chosen independently of the jet
order $R$.

\paragraph{Finite-level nonlinear estimates.}

For fixed $p,q\geq2$, the collision block
\begin{equation}
\widehat{\mathcal Q}_{p,q}:
\mathscr H_p\times\mathscr H_q
\longrightarrow
\mathscr H_{p+q}
\end{equation}
acts between finite-dimensional Hilbert spaces.  Define its bilinear
operator norm by
\begin{equation}
C_{p,q}
:=
\sup_{\substack{u\in\mathscr H_p,\,
v\in\mathscr H_q\\u\neq0,\;v\neq0}}
\frac{
\left\|
\widehat{\mathcal Q}_{p,q}(u,v)
\right\|_{\mathscr H_{p+q}}
}{
\|u\|_{\mathscr H_p}
\|v\|_{\mathscr H_q}
}.
\label{eq:maxwell_jet_Cpq_def}
\end{equation}
For $p,q\geq2$ the ordered loss contribution in
Eq.~\eqref{eq:maxwell_sector_collision_block_explicit} vanishes, but this
observation is not required for the estimate itself.

The constant $C_{p,q}$ is finite and can be calculated directly from the
collision tensor already constructed above.  Indeed, using
Eq.~\eqref{eq:maxwell_sector_collision_block_basis}, write
\begin{equation}
u
=
\sum_{|\bm m|=p}u_{\bm m}|\bm m\rangle,
\qquad
v
=
\sum_{|\bm l|=q}v_{\bm l}|\bm l\rangle .
\end{equation}
Then
\begin{equation}
\widehat{\mathcal Q}_{p,q}(u,v)
=
\sum_{|\bm n|=p+q}
\left(
\sum_{\substack{|\bm m|=p\\|\bm l|=q}}
\mathcal Q_{\bm n;\bm m\bm l}
u_{\bm m}v_{\bm l}
\right)
|\bm n\rangle .
\label{eq:maxwell_jet_Qpq_components}
\end{equation}
Consequently,
\begin{equation}
C_{p,q}
\leq
C_{p,q}^{\mathrm{HS}},
\label{eq:maxwell_jet_Cpq_HS_bound}
\end{equation}
where
\begin{equation}
\left(
C_{p,q}^{\mathrm{HS}}
\right)^2
:=
\sum_{\substack{|\bm n|=p+q\\
|\bm m|=p\\
|\bm l|=q}}
\left|
\mathcal Q_{\bm n;\bm m\bm l}
\right|^2 .
\label{eq:maxwell_jet_Cpq_HS}
\end{equation}
This follows by viewing the rank-three collision tensor as a linear map
from
$\mathscr H_p\otimes\mathscr H_q$ to $\mathscr H_{p+q}$ and bounding its
operator norm by its Hilbert--Schmidt norm.  Therefore
\begin{equation}
\left\|
\widehat{\mathcal Q}_{p,q}(u,v)
\right\|_{\mathscr H_{p+q}}
\leq
C_{p,q}
\|u\|_{\mathscr H_p}
\|v\|_{\mathscr H_q}
\leq
C_{p,q}^{\mathrm{HS}}
\|u\|_{\mathscr H_p}
\|v\|_{\mathscr H_q}.
\label{eq:maxwell_jet_bilinear_bound}
\end{equation}
Thus no additional continuity assumption on the full completed Fock space
is needed for the finite-jet argument.  All constants entering the proof
are finite-level quantities determined by the explicit collision
coefficients
$\mathcal Q_{\bm n;\bm m\bm l}$.

\begin{theorem}[Global nonlinear relaxation of finite compatible Fock jets]
\label{thm:maxwell_finite_jet_relaxation}

Let $b\geq0$ be a nontrivial angular-cutoff Maxwell kernel and let
$R\geq2$.  Consider homogeneous collision dynamics in the moment-adapted
frame with fixed density, mean velocity, and temperature.  Let
\begin{equation}
J_Rh(0)
=
\left(
h^{[2]}(0),\ldots,h^{[R]}(0)
\right)
\in\mathscr J_R
\end{equation}
be arbitrary.

Then:

\begin{enumerate}
\item
The equations
\eqref{eq:maxwell_jet_exact_triangular_hierarchy} for
$2\leq r\leq R$ form an exact autonomous finite-dimensional system.

\item
This finite-jet system has a unique global solution for every initial
jet, without any smallness assumption.

\item
Every level converges exponentially to zero.  More precisely, there are
explicit finite constants $A_r$, depending polynomially on the initial
jet, such that
\begin{equation}
\|h^{[r]}(t)\|_{\mathscr H_r}
\leq
A_r e^{-\lambda_*t},
\qquad
2\leq r\leq R,
\label{eq:maxwell_jet_shell_decay}
\end{equation}
where $\lambda_*>0$ is the $R$-independent rate defined in
Eq.~\eqref{eq:maxwell_jet_uniform_gap}.
\end{enumerate}

The constants $A_r$ may be defined recursively by
\begin{equation}
A_2
=
\|h^{[2]}(0)\|,
\qquad
A_3
=
\|h^{[3]}(0)\|,
\label{eq:maxwell_jet_A23}
\end{equation}
and, for $r\geq4$,
\begin{equation}
A_r
=
\|h^{[r]}(0)\|
+
\frac{1}{\lambda_*}
\sum_{\substack{p+q=r\\p,q\geq2}}
C_{p,q}A_pA_q.
\label{eq:maxwell_jet_Ar_recursion}
\end{equation}
Consequently, with the finite-jet Hilbert norm
\begin{equation}
\|J_Rh\|_{\mathscr J_R}^2
:=
\sum_{r=2}^{R}
\|h^{[r]}\|_{\mathscr H_r}^2,
\label{eq:maxwell_jet_norm}
\end{equation}
one has
\begin{equation}
\boxed{
\|J_Rh(t)\|_{\mathscr J_R}
\leq
\mathcal A_R(h(0))
e^{-\lambda_*t},
}
\label{eq:maxwell_jet_global_decay}
\end{equation}
where
\begin{equation}
\mathcal A_R(h(0))
:=
\left(
\sum_{r=2}^{R}A_r^2
\right)^{1/2}.
\label{eq:maxwell_jet_AR_def}
\end{equation}
\end{theorem}

\begin{proof}

The exact autonomous character follows from
Eqs.~\eqref{eq:maxwell_jet_exact_triangular_hierarchy} and
\eqref{eq:maxwell_jet_strict_triangularity}.  The equation for level $r$
contains its own amplitude only through the linear term
$\widehat{\mathcal L}_rh^{[r]}$; every nonlinear forcing term depends only
on levels at most $r-2$.

We prove global existence and
Eq.~\eqref{eq:maxwell_jet_shell_decay} simultaneously by induction over
the level number.

For $r=2$ and $r=3$, the nonlinear sum in
Eq.~\eqref{eq:maxwell_jet_exact_triangular_hierarchy} is empty.  Therefore
\begin{equation}
|h^{[r]}(t)\rangle
=
e^{t\widehat{\mathcal L}_r}
|h^{[r]}(0)\rangle,
\qquad r=2,3,
\end{equation}
and Eq.~\eqref{eq:maxwell_jet_semigroup_bound} gives
\begin{equation}
\|h^{[r]}(t)\|
\leq
\|h^{[r]}(0)\|e^{-\lambda_*t}
=
A_re^{-\lambda_*t}.
\end{equation}
In particular, these two levels exist globally.

Now let $r\geq4$ and assume that all levels $2,\ldots,r-1$ have already
been constructed globally and satisfy
\begin{equation}
\|h^{[p]}(t)\|
\leq
A_pe^{-\lambda_*t},
\qquad
2\leq p<r.
\label{eq:maxwell_jet_induction_hypothesis}
\end{equation}
Because of strict triangularity, every pair $(p,q)$ occurring in the
nonlinear source of the $r$th equation satisfies
\begin{equation}
p,q<r.
\end{equation}
The induction hypothesis and
Eq.~\eqref{eq:maxwell_jet_bilinear_bound} therefore imply
\begin{align}
&
\left\|
\widehat{\mathcal Q}_{p,q}
\left(
h^{[p]}(t),h^{[q]}(t)
\right)
\right\|
\leq
C_{p,q}
A_pA_q
e^{-2\lambda_*t}.
\label{eq:maxwell_jet_source_decay}
\end{align}

The $r$th equation is thus a finite-dimensional linear inhomogeneous
equation with a globally defined continuous forcing.  Its unique solution
is given by the Duhamel formula
\begin{equation}
\begin{aligned}
|h^{[r]}(t)\rangle
={}&
e^{t\widehat{\mathcal L}_r}
|h^{[r]}(0)\rangle
+
\sum_{\substack{p+q=r\\p,q\geq2}}
\int_0^t
e^{(t-s)\widehat{\mathcal L}_r}
\widehat{\mathcal Q}_{p,q}
\left(
|h^{[p]}(s)\rangle,
|h^{[q]}(s)\rangle
\right)
\,ds .
\end{aligned}
\label{eq:maxwell_jet_duhamel}
\end{equation}
This formula already proves global existence of the $r$th level once all
lower levels are known globally.

Using Eqs.~\eqref{eq:maxwell_jet_semigroup_bound} and
\eqref{eq:maxwell_jet_source_decay}, we obtain
\begin{align}
\|h^{[r]}(t)\|
\leq{}&
e^{-\lambda_*t}
\|h^{[r]}(0)\|
\nonumber\\
&+
\sum_{\substack{p+q=r\\p,q\geq2}}
C_{p,q}A_pA_q
\int_0^t
e^{-\lambda_*(t-s)}
e^{-2\lambda_*s}
\,ds .
\label{eq:maxwell_jet_duhamel_bound1}
\end{align}
The time convolution is elementary:
\begin{align}
\int_0^t
e^{-\lambda_*(t-s)}
e^{-2\lambda_*s}
\,ds
&=
e^{-\lambda_*t}
\int_0^t
e^{-\lambda_*s}\,ds
\nonumber\\
&=
\frac{
1-e^{-\lambda_*t}
}{
\lambda_*
}
e^{-\lambda_*t}
\nonumber\\
&\leq
\frac{1}{\lambda_*}
e^{-\lambda_*t}.
\label{eq:maxwell_jet_convolution_bound}
\end{align}
Substitution into
Eq.~\eqref{eq:maxwell_jet_duhamel_bound1} gives
\begin{align}
\|h^{[r]}(t)\|
&\leq
\left[
\|h^{[r]}(0)\|
+
\frac{1}{\lambda_*}
\sum_{\substack{p+q=r\\p,q\geq2}}
C_{p,q}A_pA_q
\right]
e^{-\lambda_*t}
\nonumber\\
&=
A_re^{-\lambda_*t},
\end{align}
where the last equality is precisely the recursive definition
\eqref{eq:maxwell_jet_Ar_recursion}.  This closes the induction.

Repeating the argument through $r=R$ proves global existence and
exponential relaxation of the complete finite jet.  Finally,
\begin{align}
\|J_Rh(t)\|_{\mathscr J_R}^2
&=
\sum_{r=2}^{R}
\|h^{[r]}(t)\|^2
\nonumber\\
&\leq
e^{-2\lambda_*t}
\sum_{r=2}^{R}A_r^2,
\end{align}
which is Eq.~\eqref{eq:maxwell_jet_global_decay}.
\end{proof}

\paragraph{Structure of the nonlinear prefactor.}

The recursive constants make the upward nonlinear cascade quantitative.
Writing
\begin{equation}
a_r:=\|h^{[r]}(0)\|,
\end{equation}
the first cases are
\begin{align}
A_2&=a_2,
\\
A_3&=a_3,
\\
A_4
&=
a_4
+
\frac{C_{2,2}}{\lambda_*}a_2^2,
\label{eq:maxwell_jet_A4}
\\
A_5
&=
a_5
+
\frac{C_{2,3}+C_{3,2}}{\lambda_*}
a_2a_3,
\label{eq:maxwell_jet_A5}
\\
A_6
&=
a_6
+
\frac{1}{\lambda_*}
\left[
C_{2,4}a_2A_4
+
C_{3,3}a_3^2
+
C_{4,2}A_4a_2
\right].
\label{eq:maxwell_jet_A6}
\end{align}
Thus each $A_r$ is a finite polynomial with nonnegative coefficients in
the initial level norms $a_2,\ldots,a_r$.  Its monomials encode the
possible graded collision trees whose incoming degrees add to $r$.
No denominator involving an initial amplitude occurs, and therefore no
smallness condition enters at any finite jet order.

The decay estimate
\eqref{eq:maxwell_jet_global_decay} should not be confused with
monotonicity of the unweighted finite-jet norm.  The nonlinear forcing of a
higher level can temporarily increase the norm of that level even while
all levels ultimately decay.  The theorem establishes nonlinear
dissipative stability in the stronger long-time sense of global existence
and exponential relaxation, not a pointwise Lyapunov inequality for
$\|J_Rh(t)\|_{\mathscr J_R}$.

\paragraph{Relation to the full hierarchy.}

Whenever a full compatible Fock solution exists, its first $R$ levels
necessarily satisfy the autonomous system of
Theorem~\ref{thm:maxwell_finite_jet_relaxation}.  By uniqueness of that
finite-dimensional triangular system, they coincide with the finite-jet
solution constructed above.  Hence
Eq.~\eqref{eq:maxwell_jet_global_decay} is an exact statement about the
low-order sectors of the full nonlinear Maxwell collision dynamics and
not a consequence of setting the higher levels to zero.

This should also be distinguished from the use of an infinite triangular
spectral system to construct the complete solution.  In that problem one
must control the growth of the nonlinear coefficients and prove convergence
of the infinite mode expansion; such estimates are central, for example,
to Ref.~\cite{GlangetasLiXu2016Triangular}.  For a fixed jet, those
infinite-dimensional convergence issues are absent, and strict
triangularity reduces the dynamics recursively to finitely many forced
linear equations.  This is why arbitrary finite jet amplitudes can be
treated without a perturbative smallness condition.

Conversely, the theorem is an algebraic and dynamical statement on the compatible finite Fock levels; it does not require the finite Hermite reconstruction of a jet to be pointwise positive.  Global positivity and the Boltzmann $H$-theorem concern the reconstructed physical distribution and are not needed for the finite-jet relaxation result proved here.

The finite-jet result therefore provides a natural intermediate step between the complete one-vacuum spectrum and a possible nonlinear theory on a completed graded Fock space: exact grading gives strict triangularity and global finite-jet relaxation, while control of the full hierarchy remains a separate problem of weighted infinite-level estimates.

\section{Inhomogeneous dynamics, covariance, and hydrodynamic limit}
\label{sec:maxwell_inhomogeneous}
\label{sec:maxwell_inhomogeneous_hydro}

The homogeneous construction already shows representation independence at the algebraic level: the collision vertex and its grading are defined on symmetric Fock space before any coordinate realization is chosen.  The inhomogeneous problem provides a more operational test.  A local realization moves with the hydrodynamic fields $\bm u(\bm x,t)$ and $\theta(\bm x,t)$, so propagation acquires a differential connection.  Here we use this fact to make covariance constructive.  We first extend the realization-covariance statement established for the Lebowitz--Frisch--Helfand kinetic model \cite{Karlin2026LFHFock} from a one-input local operator to the present two-input nonlinear collision vertex.  We then evaluate the same physical reconstruction connection by two apparently different routes, Hermite-function and Fourier.  The two calculations meet in the same ladder-operator connection and therefore generate the same Chapman--Enskog source and the same hydrodynamics.

This comparison is deliberately not the shortest route to the transport coefficients.  The Fourier calculation alone would suffice for that purpose.  Its role is instead to demonstrate concretely that representation independence gives freedom of computational route: a coordinate realization may be used as scaffolding to construct a convenient abstract Fock operator and then discarded.

\subsection{Covariance of the nonlinear two-leg collision vertex}
\label{subsec:maxwell_covariance}

Let $\mathcal R_i:\mathscr H\to\mathscr F_i$, $i=1,2$, be two admissible realizations and set $S=\mathcal R_2\mathcal R_1^{-1}$.  If $\widehat{\mathcal Q}:\mathscr H\times\mathscr H\to\mathscr H$ is the abstract Maxwell collision vertex, define its coordinate representatives by
\begin{equation}
\mathcal Q_i(F,G)
:=\mathcal R_i\widehat{\mathcal Q}(\mathcal R_i^{-1}F,\mathcal R_i^{-1}G).
\label{eq:maxwell_Q_coordinate_definition_cov}
\end{equation}

\begin{proposition}[Two-leg covariance of the Maxwell vertex]
The coordinate collision maps obey
\begin{equation}
\boxed{
\mathcal Q_2
=
S\circ\mathcal Q_1\circ(S^{-1}\times S^{-1}),
}
\label{eq:maxwell_coordinate_covariance_map}
\end{equation}
or equivalently
\begin{equation}
\mathcal Q_2(SF,SG)=S\,\mathcal Q_1(F,G).
\label{eq:maxwell_coordinate_covariance}
\end{equation}
Thus both incoming legs transform, while the outgoing collision state transforms once.
\end{proposition}

\begin{proof}
Using $\mathcal R_2=S\mathcal R_1$ in Eq.~\eqref{eq:maxwell_Q_coordinate_definition_cov},
\begin{align}
\mathcal Q_2(SF,SG)
&=\mathcal R_2\widehat{\mathcal Q}(\mathcal R_2^{-1}SF,\mathcal R_2^{-1}SG)\\
&=S\mathcal R_1\widehat{\mathcal Q}(\mathcal R_1^{-1}F,\mathcal R_1^{-1}G)
=S\mathcal Q_1(F,G).
\end{align}
\end{proof}

Consequently, for every pair of abstract inputs,
\begin{equation}
|\widehat{\mathcal Q}[\phi,\psi]\rangle
=
\mathcal R_i^{-1}\mathcal Q_i(\mathcal R_i|\phi\rangle,\mathcal R_i|\psi\rangle),
\qquad i=H,P,F,B.
\label{eq:maxwell_abstract_covariance}
\end{equation}
In particular,
\begin{equation}
\begin{aligned}
|\widehat{\mathcal Q}[\phi,\psi]\rangle
&=\mathcal R_H^{-1}\mathcal Q_H(\mathcal R_H|\phi\rangle,\mathcal R_H|\psi\rangle)
=\mathcal R_P^{-1}\mathcal Q_P(\mathcal R_P|\phi\rangle,\mathcal R_P|\psi\rangle)\\
&=\mathcal R_F^{-1}\mathcal Q_F(\mathcal R_F|\phi\rangle,\mathcal R_F|\psi\rangle)
=\mathcal R_B^{-1}\mathcal Q_B(\mathcal R_B|\phi\rangle,\mathcal R_B|\psi\rangle).
\end{aligned}
\label{eq:maxwell_realization_equivalence}
\end{equation}
Equation~\eqref{eq:maxwell_coordinate_covariance_map} is the bilinear analogue of ordinary similarity covariance for a one-input operator.  It is the nonlinear extension needed to combine the collision construction of the present paper with the general covariance of parameter-dependent realizations and propagation established in Ref.~\cite{Karlin2026LFHFock}.  Bobylev--Bargmann coordinates are therefore privileged for constructing the Maxwell vertex, not intrinsic to the resulting map.

\subsection{Moving local frames: one abstract connection, two coordinate routes}
\label{subsec:maxwell_fourier_connection}

For the density-scaled state, let the physical velocity reconstruction be
\begin{equation}
f=\mathcal B(\bm u,\theta)|\chi\rangle,
\qquad
\mathcal B=S_{\sqrt W}\mathcal R_H=S_W\mathcal R_P,
\qquad
\mathcal A_\mu:=\mathcal B^{-1}(\partial_\mu\mathcal B).
\label{eq:maxwell_B_physical}
\end{equation}
The abstract connection $\mathcal A_\mu$ belongs to the complete physical reconstruction.  Bare realization connectors such as $\mathcal R_H^{-1}(\partial_\mu \mathcal R_H)$ are coordinate dependent, as discussed in Ref.~\cite{Karlin2026LFHFock}; their coordinate dependence is compensated by the corresponding reconstruction factors in the complete pull-back \cite{Karlin2026LFHFock}.  We now evaluate $\mathcal A_\mu$ in two ways.

\subsubsection{Hermite-function route}
\label{subsubsec:maxwell_Hermite_connection_route}

For the Hermite-function realization, the differential connectors derived in Ref.~\cite{Karlin2026LFHFock} are
\begin{align}
\widehat{\mathcal C}^{(H)}_{u_\beta}
&=\mathcal R_H^{-1}\frac{\partial \mathcal R_H}{\partial u_\beta}
=\frac{1}{2\sqrt\theta}(\hat a_\beta^\dagger-\hat a_\beta),
\label{eq:maxwell_Hermite_Cu}
\\
\widehat{\mathcal C}^{(H)}_\theta
&=\mathcal R_H^{-1}\frac{\partial \mathcal R_H}{\partial\theta}
=\frac{1}{4\theta}
\left[(\hat a_\beta^\dagger)^2-\hat a_\beta^2\right].
\label{eq:maxwell_Hermite_Ctheta}
\end{align}
No explicit Hermite functions are needed here: their differentiation has already been encoded in these ladder-operator identities.  Write
\begin{equation}
\widehat X_\alpha:=\hat a_\alpha+\hat a_\alpha^\dagger,
\qquad
\widehat{\bm X}^{\,2}:=\widehat X_\alpha\widehat X_\alpha.
\end{equation}
The remaining factor in the physical reconstruction is multiplication by $\sqrt W$.  At fixed physical velocity,
\begin{align}
\mathcal R_H^{-1}S_{\sqrt W}^{-1}
\frac{\partial S_{\sqrt W}}{\partial u_\beta}\mathcal R_H
&=\frac{1}{2\sqrt\theta}\widehat X_\beta,
\label{eq:maxwell_Hermite_weight_u}
\\
\mathcal R_H^{-1}S_{\sqrt W}^{-1}
\frac{\partial S_{\sqrt W}}{\partial\theta}\mathcal R_H
&=\frac{1}{4\theta}
\left(\widehat{\bm X}^{\,2}-d\right).
\label{eq:maxwell_Hermite_weight_theta}
\end{align}
Adding the basis and reconstruction contributions gives the elementary connections of the complete physical map,
\begin{align}
\mathcal A^{(H)}_{u_\beta}
&=\frac{1}{2\sqrt\theta}
\left[\widehat X_\beta+\hat a_\beta^\dagger-\hat a_\beta\right]
=\frac{1}{\sqrt\theta}\hat a_\beta^\dagger,
\label{eq:maxwell_Hermite_full_u}
\\
\mathcal A^{(H)}_\theta
&=\frac{1}{4\theta}
\left[
\widehat{\bm X}^{\,2}-d
+(\hat a_\beta^\dagger)^2-\hat a_\beta^2
\right]
=\frac{1}{2\theta}
\left[\hat N+(\hat a_\beta^\dagger)^2\right].
\label{eq:maxwell_Hermite_full_theta}
\end{align}
In the last equality we used
$\widehat{\bm X}^{\,2}-d=\hat a_\beta^2+(\hat a_\beta^\dagger)^2+2\hat N$.
Thus the Hermite-function route has already left coordinate space: its output is a pair of operators in the canonical Fock algebra.

\subsubsection{Fourier route}
\label{subsubsec:maxwell_Fourier_connection_route}

The same physical reconstruction may be Fourier transformed.  Let $\mathcal B_F=\mathcal R_F$ denote the parameter-dependent Fourier reconstruction.  At fixed wave vector $\bm k$, with $\bm z=-i\sqrt\theta\,\bm k$, differentiation of the Fourier Maxwellian factor and the thermal dilation of $\bm z$ gives
\begin{align}
\mathcal B_F^{-1}\frac{\partial\mathcal B_F}{\partial u_\beta}
&=\frac{1}{\sqrt\theta}\hat a_\beta^\dagger,
\label{eq:maxwell_full_connection_u}
\\
\mathcal B_F^{-1}\frac{\partial\mathcal B_F}{\partial\theta}
&=\frac{1}{2\theta}
\left[\hat N+(\hat a_\alpha^\dagger)^2\right].
\label{eq:maxwell_full_connection_theta}
\end{align}
The number operator comes from dilation of the Bargmann coordinate and the pair-creation term from differentiation of the Fourier Maxwellian vacuum.  Comparison with Eqs.~\eqref{eq:maxwell_Hermite_full_u}--\eqref{eq:maxwell_Hermite_full_theta} gives
\begin{equation}
\boxed{
\mathcal A^{(H)}_{u_\beta}=\mathcal A^{(F)}_{u_\beta},
\qquad
\mathcal A^{(H)}_\theta=\mathcal A^{(F)}_\theta.
}
\label{eq:maxwell_HF_connection_equivalence}
\end{equation}
Hence either route yields
\begin{equation}
\boxed{
\mathcal A_\mu
=
\frac{\partial_\mu u_\beta}{\sqrt\theta}\hat a_\beta^\dagger
+
\frac{\partial_\mu\theta}{2\theta}
\left[\hat N+(\hat a_\alpha^\dagger)^2\right].
}
\label{eq:maxwell_A_mu}
\end{equation}
The equality in Eq.~\eqref{eq:maxwell_HF_connection_equivalence} should not be confused with equality of the bare coordinate connectors.  It holds for the connection of the complete physical reconstruction.  The Hermite and Fourier coordinate calculations are different; covariance makes their pulled-back Fock operator the same.

This is the practical meaning of choosing an ``economical'' realization.  If only the Chapman--Enskog coefficients were sought, the Fourier route would be shorter.  For the representation-independent formulation, however, the two-route calculation is more informative: it shows explicitly that coordinate realizations are computational devices for constructing an abstract operator, not part of the final dynamics.

\subsection{Spatially inhomogeneous Maxwell-molecule Boltzmann equation}
\label{subsec:maxwell_inhomogeneous}

Up to this point the collision analysis has been spatially homogeneous.  We now restore transport and consider
\begin{equation}
\partial_t f+v_\alpha\partial_\alpha f=Q(f,f).
\label{eq:maxwell_inhomogeneous_physical}
\end{equation}
The velocity multiplication operator is the moment-intertwined operator already obtained in Eq.~\eqref{eq:maxwell_vhat}, evaluated in the moment-adapted frame.  Pulling back the physical equation gives
\begin{equation}
\boxed{
\left[
\partial_t+\mathcal A_t+
\widehat v_\alpha(\partial_\alpha+\mathcal A_\alpha)
\right]|\chi\rangle
=
|\widehat{\mathcal Q}[\chi,\chi]\rangle.
}
\label{eq:maxwell_inhomogeneous_fock}
\end{equation}
Equation~\eqref{eq:maxwell_inhomogeneous_fock} is the common abstract Fock equation reached from either the Hermite-function or Fourier route.  Its left-hand side is the covariant moving-frame propagation structure established in Ref.~\cite{Karlin2026LFHFock}; its right-hand side is the two-legged nonlinear Maxwell vertex whose covariance was proved above.  Thus replacing a one-leg local relaxation sector by the present bilinear collision map leaves the representation-independent propagation architecture unchanged.

\subsection{Compatibility propagation and exact balance equations}
\label{subsec:maxwell_compatibility_propagation}

The moment-adapted frame is defined by the density, momentum, and energy functionals constructed in Sec.~\ref{subsec:maxwell_moment_compatibility}.  The general transport--moment theorem proved in Ref.~\cite{Karlin2026LFHFock} applies unchanged to the propagation part of Eq.~\eqref{eq:maxwell_inhomogeneous_fock}: once the complete physical reconstruction is intertwined, the conservative moment balances follow without reopening the explicit connection algebra.  The Maxwell collision vertex contributes nothing to these three projections because the corresponding collision bras annihilate it identically.

On the moment-adapted manifold the nonequilibrium stress is the Hermite coefficient already identified in Eq.~\eqref{eq:maxwell_stress_Fock_pairing}.  The corresponding heat-flux coefficient is
\begin{equation}
q_\alpha
=\frac{m}{2}\theta^{3/2}
\left\langle0\left|\hat a_\beta\hat a_\beta\hat a_\alpha\right|\chi\right\rangle .
\label{eq:maxwell_heat_Fock_pairing}
\end{equation}
The exact compatible balances are therefore
\begin{align}
\partial_t n+\partial_\alpha(nu_\alpha)&=0,
\label{eq:maxwell_continuity}
\\
\partial_t(mnu_\beta)
+\partial_\alpha\!\left(mnu_\beta u_\alpha+mn\theta\,\delta_{\beta\alpha}+\pi_{\beta\alpha}\right)&=0,
\label{eq:maxwell_momentum_balance}
\\
\frac{md}{2}nD_u\theta
+mn\theta\,\partial_\alpha u_\alpha
+\pi_{\alpha\beta}\partial_\alpha u_\beta
+\partial_\alpha q_\alpha&=0,
\qquad D_u:=\partial_t+u_\alpha\partial_\alpha.
\label{eq:maxwell_temperature_balance}
\end{align}
At local equilibrium, $\pi_{\alpha\beta}=q_\alpha=0$, and these equations reduce to the Euler compatibility system.

\subsection{One Chapman--Enskog source from two realization routes}
\label{subsec:maxwell_CE1}

We now use the two constructions above as an explicit covariance test.  Acting with the complete propagation operator on the local Maxwellian vacuum and imposing the Euler compatibility equations removes the vacuum, level-$1$, and scalar level-$2$ pieces.  Because Eqs.~\eqref{eq:maxwell_Hermite_full_u}--\eqref{eq:maxwell_Hermite_full_theta} and Eqs.~\eqref{eq:maxwell_full_connection_u}--\eqref{eq:maxwell_full_connection_theta} represent the same $\mathcal A_\mu$, the Hermite and Fourier routes give the same abstract first-order source,
\begin{equation}
|\mathcal S_{\rm kin}^{(0)}\rangle_H
=
|\mathcal S_{\rm kin}^{(0)}\rangle_F
=:
|\mathcal S_{\rm kin}^{(0)}\rangle.
\label{eq:maxwell_source_HF_equivalence}
\end{equation}
With
\begin{equation}
D_{\alpha\beta}
:=\partial_\alpha u_\beta+\partial_\beta u_\alpha
-\frac{2}{d}\delta_{\alpha\beta}\partial_\gamma u_\gamma,
\label{eq:maxwell_Dab}
\end{equation}
the common source is
\begin{equation}
\boxed{
|\mathcal S_{\rm kin}^{(0)}\rangle
=
\frac n2 D_{\alpha\beta}
\hat a_\alpha^\dagger\hat a_\beta^\dagger|0\rangle
+
\frac{n}{2\sqrt\theta}(\partial_\alpha\theta)
\hat a_\alpha^\dagger(\hat a_\beta^\dagger)^2|0\rangle
\in
\mathscr H_{2,\mathrm{tl}}\oplus\mathscr H_{3,\mathrm{vec}}.
}
\label{eq:maxwell_kinetic_source_Fock}
\end{equation}
This sector selection is kinematic and representation independent; no Maxwell collision matrix element has yet been used.  In particular, neither explicit Hermite functions nor explicit Fourier images enter the hydrodynamic calculation after the operator connections have been established.

The model-specific input is supplied by the one-vacuum spectrum derived in Sec.~\ref{subsec:maxwell_linearized_spectrum}.  The two selected sectors are eigenspaces with
\begin{equation}
\lambda_\pi=\lambda_{02}=\frac34n\Omega_2,
\qquad
\lambda_q=\lambda_{11}=\frac12n\Omega_2,
\qquad
\frac{\lambda_q}{\lambda_\pi}=\frac23.
\label{eq:maxwell_hydro_rates_compact}
\end{equation}
Hence the first kinetic correction is obtained by scalar inversion on the two sectors,
\begin{equation}
|\chi^{(1)}\rangle
=-\frac{n}{2\lambda_\pi}
D_{\alpha\beta}\hat a_\alpha^\dagger\hat a_\beta^\dagger|0\rangle
-\frac{n}{2\sqrt\theta\,\lambda_q}
(\partial_\alpha\theta)
\hat a_\alpha^\dagger(\hat a_\beta^\dagger)^2|0\rangle.
\label{eq:maxwell_chi1_compact}
\end{equation}
The level-$2$ and level-$3$ vacuum pairings are canonical Fock contractions.  Applying them to Eqs.~\eqref{eq:maxwell_stress_Fock_pairing}--\eqref{eq:maxwell_heat_Fock_pairing} gives, for $d=3$,
\begin{align}
\pi_{\mu\nu}^{(1)}
&=-\frac{mn\theta}{\lambda_\pi}D_{\mu\nu}
=-\mu D_{\mu\nu},
&
\mu&=\frac{p}{\lambda_\pi},
\label{eq:maxwell_viscosity}
\\
q_\mu^{(1)}
&=-\frac{5mn\theta}{2\lambda_q}\partial_\mu\theta
=-\kappa\partial_\mu T,
&
\kappa&=\frac52\frac{k_B}{m}\frac{p}{\lambda_q}.
\label{eq:maxwell_kappa}
\end{align}
For a monatomic gas $c_p=(5/2)k_B/m$, and therefore
\begin{equation}
\boxed{
\Pr
=\frac{c_p\mu}{\kappa}
=\frac{\lambda_q}{\lambda_\pi}
=\frac23.
}
\label{eq:maxwell_prandtl_exact}
\end{equation}
The equality of the Hermite and Fourier routes is therefore stronger than agreement of the final scalar $\Pr$: they give the same abstract Fock source before the Maxwell collision inverse is applied.  The Chapman--Enskog calculation separates into a representation-independent moving-frame source and two Maxwell-specific spectral numbers.

\subsection{Recovery of the conventional Boltzmann equation}
\label{subsec:maxwell_recovery_velocity}

The Fourier construction can be pushed back to velocity space explicitly.  Let $\mathcal B_F=\mathcal R_F$ denote the parameter-dependent Fourier reconstruction,
\begin{equation}
\widetilde f
=\mathcal B_F|\chi\rangle,
\qquad
\widetilde f(\bm k)
=\widetilde W_{\bm u,\theta}(\bm k)\,
(\mathcal R_B|\chi\rangle)(-i\sqrt{\theta}\,\bm k),
\label{eq:maxwell_BF_recalled}
\end{equation}
and define
\begin{equation}
\mathcal R_v:=\mathcal F^{-1}\mathcal B_F,
\qquad f=\mathcal R_v|\chi\rangle.
\label{eq:maxwell_velocity_realization}
\end{equation}
The abstract vertex, Bobylev's Fourier collision map $\widetilde{\mathcal Q}_M$, and the conventional Maxwell operator $\mathcal Q_M$ fit into
\begin{equation}
\begin{tikzcd}[column sep=large, row sep=large]
\mathscr H\times \mathscr H
\arrow[r, "\widehat{\mathcal Q}"]
\arrow[d, "\mathcal B_F\times\mathcal B_F"']
&
\mathscr H
\arrow[d, "\mathcal B_F"]
\\
\widetilde{\mathscr F}\times \widetilde{\mathscr F}
\arrow[r, "\widetilde{\mathcal Q}_M"]
\arrow[d, "\mathcal F^{-1}\times\mathcal F^{-1}"']
&
\widetilde{\mathscr F}
\arrow[d, "\mathcal F^{-1}"]
\\
\mathscr F_v\times \mathscr F_v
\arrow[r, "\mathcal Q_M"]
&
\mathscr F_v .
\end{tikzcd}
\label{eq:maxwell_commutative_diagram}
\end{equation}
The upper square is the nonlinear two-leg intertwining proved above, while the lower square is Bobylev's Fourier equivalence with the conventional Maxwell collision integral.  Their composition gives
\begin{equation}
\mathcal R_v\,\widehat{\mathcal Q}
=\mathcal Q_M(\mathcal R_v\otimes\mathcal R_v).
\label{eq:maxwell_velocity_intertwining}
\end{equation}

\begin{proposition}[Recovery of the conventional kinetic equation]
\label{prop:maxwell_recovery_velocity}
If
\begin{equation}
\partial_t|\chi\rangle+\widehat{\mathcal T}|\chi\rangle
=\widehat{\mathcal Q}(|\chi\rangle,|\chi\rangle),
\label{eq:maxwell_abstract_fock_eq_recalled}
\end{equation}
where $\widehat{\mathcal T}$ is the pulled-back propagation operator, then $f=\mathcal R_v|\chi\rangle$ satisfies
\begin{equation}
\partial_t f+v_\alpha\partial_{x_\alpha}f=\mathcal Q_M(f,f).
\label{eq:maxwell_conventional_boltzmann_recovered}
\end{equation}
\end{proposition}

\begin{proof}
Here the pulled-back propagation operator is
\begin{equation}
\widehat{\mathcal T}
=
\mathcal A_t+
\widehat v_\alpha(\partial_\alpha+\mathcal A_\alpha),
\qquad
\mathcal A_\mu=\mathcal R_v^{-1}(\partial_\mu\mathcal R_v),
\label{eq:maxwell_T_recovery_definition}
\end{equation}
where the second identity follows equivalently from $\mathcal R_v=\mathcal F^{-1}\mathcal B_F$ and the Fourier construction of the same complete physical connection.  Since velocity multiplication is intertwined according to
$\mathcal R_v\widehat v_\alpha=v_\alpha\mathcal R_v$, one has, on $|\chi\rangle$,
\begin{equation}
\mathcal R_v\widehat{\mathcal T}|\chi\rangle
=
(\partial_t\mathcal R_v)|\chi\rangle
+
 v_\alpha\partial_\alpha(\mathcal R_v|\chi\rangle).
\label{eq:maxwell_T_recovery_intertwining}
\end{equation}
Applying $\mathcal R_v$ to Eq.~\eqref{eq:maxwell_abstract_fock_eq_recalled} and using
\begin{equation}
\mathcal R_v\partial_t|\chi\rangle
=
\partial_t(\mathcal R_v|\chi\rangle)
-(\partial_t\mathcal R_v)|\chi\rangle,
\end{equation}
the two terms containing $\partial_t\mathcal R_v$ cancel.  The left-hand side therefore becomes
$\partial_t f+v_\alpha\partial_\alpha f$.  Equation~\eqref{eq:maxwell_velocity_intertwining} gives the right-hand side as $\mathcal Q_M(f,f)$, proving Eq.~\eqref{eq:maxwell_conventional_boltzmann_recovered}.
\end{proof}

Thus Bargmann/Fourier coordinates provide the shortest construction of the nonlinear Maxwell vertex and an economical construction of the moving-frame connection, while Hermite-function coordinates provide an independent route to the same complete propagation operator.  Once these operators are pulled back, the hydrodynamic calculation is entirely Fock-algebraic.  The agreement of the two routes is the concrete realization-level counterpart of the intrinsic grading and covariance results proved above.

\section{Outlook: Beyond Maxwell molecules}
\label{subsec:maxwell_hard_sphere_outlook}

The exact degree grading obtained above is special to Maxwell molecules, because it ultimately rests on the substitution--product structure exposed by Bobylev's identity.  For non-Maxwell kernels an additional relative-speed dependence destroys this exact level selection.

For three-dimensional hard spheres, the Carleman representation provides a potentially useful alternative starting point for a Fock-space analysis: it reorganizes the full collision integral in terms of orthogonal velocity increments and the transformed kernel simplifies substantially \cite{MouhotPareschi2006FastBoltzmann}.  Whether this formulation leads to a useful spectral or graded structure in Fock space is a separate problem and is not developed here.  We leave the linearized hard-sphere operator and its subsequent nonlinear extension to future work.

\section{Conclusion}
\label{sec:maxwell_conclusion}

The Maxwell-molecule Boltzmann collision operator admits an exact bilinear
Fock formulation.  The structural step is independent of Maxwell
kinetics: the canonical lift--fusion theorem defines the abstract map from
degree-preserving bosonic lifts and degree-additive symmetric-algebra
multiplication.  Bargmann coordinates represent this map as two linear
substitutions followed by ordinary multiplication and make uniqueness
transparent; they are computationally privileged, not constitutive of the
Fock structure.  Bobylev then supplies the particular collision maps and
the angular average.  The resulting abstract two-input, one-output
vertex has the exact grading
\begin{equation}
\widehat{\mathcal Q}(\mathscr H_p,\mathscr H_q)
\subseteq
\mathscr H_{p+q}.
\end{equation}

Triangular structures themselves have substantial precedents in Maxwell
kinetics.  Bobylev's Fourier representation, successive ordinary and
Sonine moment equations, and the later Laguerre--spherical nonlinear
spectral hierarchy all reveal versions of the same degree organization
\cite{Bobylev1975Fourier,Bobylev1984ExactRelaxation,Bobylev1988MaxwellReview,Ernst1979MaxwellMoments,GlangetasLiXu2016Triangular}.
The point of the Fock formulation is therefore not to attribute
triangularity to a new choice of coordinates, but to identify whose
property it is.  The grading belongs to the abstract collision vertex
itself.  In Fock space it follows directly from degree-preserving
one-particle lifts and degree-additive fusion; Fourier, moment/Sonine,
Hermite, and Laguerre--spherical rules are coordinate images of this
representation-independent statement.  In this sense Bargmann coordinates play for the Maxwell collision vertex the methodological role played by Hermite-function coordinates in the Fock formulation of the Lebowitz--Frisch--Helfand model: they expose the relevant operator economically without defining the underlying Fock space.

Moment adaptation turns the grading into strict triangularity on the
nonequilibrium homogeneous hierarchy.  This gives an exact dynamical
consequence: every finite compatible Fock jet is an autonomous factor of
the full nonlinear collision dynamics and relaxes globally and
exponentially to the Maxwellian jet, without a smallness assumption on the
retained amplitudes.  The result concerns exact low-order projections, not
a closure obtained by suppressing higher moments; the separate problem of
controlling the complete infinite hierarchy still requires
infinite-dimensional estimates.

The one-vacuum restriction preserves each level, and the intrinsic
trace/STF decomposition yields the complete Wang Chang--Uhlenbeck spectral
family.  The stress and heat-flux rates are the $(0,2)$ and $(1,1)$ entries
of that family.

The framework is not limited to homogeneous collision dynamics.  For the
inhomogeneous equation, the moment-adapted Maxwellian generates the
moving-frame connection and the exact compatibility balances.  At first
Chapman--Enskog order the propagation source lies only in the traceless
quadratic and contracted cubic sectors.  Their exact one-vacuum
eigenvalues therefore give the Navier--Stokes--Fourier coefficients
directly and recover $\Pr=2/3$.  The resulting picture is modular:
moving-frame Fock kinematics carries transport, while molecular physics is
encoded by the collision vertex.

The algebraic statement is broader than the Gaussian Hilbert realization.
Intertwining and grading hold on the finite-Fock core, finite jets remain
meaningful whenever the corresponding moments exist, and
characteristic-function formulations extend farther into heavy-tailed and
measure-valued regimes.

Finally, the exact grading isolates what is special about Maxwell molecules without making the representation strategy itself Maxwell-specific.  The three-dimensional Carleman representation suggests a distinct starting point for hard spheres, but its Fock-space spectral analysis is beyond the scope of the present work and is deferred to a separate study.

\section*{Acknowledgement of AI assistance}
During the development and preparation of this manuscript, the author used
ChatGPT (OpenAI, GPT-5.6 Sol) as an interactive research and writing assistant.
Its use included discussion and critical examination of mathematical
formulations, cross-checking of derivations, development of the presentation
and organization of the manuscript, identification of notational and
expository inconsistencies, and assistance with editorial revision and
bibliographic verification.  All mathematical results, arguments,
interpretations, references, and final text were reviewed and accepted by the
author, who assumes full responsibility for the content of the manuscript.

\appendix
\section{Velocity-space Chapman--Enskog check}
\label{app:maxwell_coordinate_CE_check}

This appendix gives an independent velocity-space check of the Fock derivation.  The weak Maxwell collision form below is the conventional coordinate realization of the same operator used in the main text.  Its exact stress and heat-flux moment identities are obtained from the full nonlinear collision operator; the usual first variation appears only when the Chapman--Enskog expansion is introduced.

We specialize to $d=3$ and set
\begin{equation}
\bm g=\bm c-\bm c_*,\qquad
\bm C=\frac{\bm c+\bm c_*}{2},\qquad
\bm e_g=\frac{\bm g}{g},\qquad
\eta=\bm e_g\cdot\bm\sigma,
\end{equation}
with
\begin{equation}
\Omega_2=\int_{S^2}b(\eta)(1-\eta^2)\,d\bm\sigma.
\end{equation}
The exact weak form is
\begin{equation}
\int\phi(\bm c)Q(f,f)\,d\bm c
=\frac12\int ff_*b(\eta)
\left(\phi'+\phi_*'-\phi-\phi_*\right)
\,d\bm\sigma\,d\bm c\,d\bm c_*,
\label{eq:app_maxwell_weak_form}
\end{equation}
where elastic collisions give
\begin{equation}
\bm c'=\bm C+\frac g2\bm\sigma,\qquad
\bm c_*'=\bm C-\frac g2\bm\sigma.
\end{equation}

\subsection{Exact stress and heat-flux relaxation in the weak form}

For the quadratic observable $c_\alpha c_\beta$,
\begin{equation}
c_\alpha c_\beta+c_{*\alpha}c_{*\beta}
=
2C_\alpha C_\beta+\frac12g_\alpha g_\beta,
\end{equation}
so that
\begin{equation}
\Delta_2^{\alpha\beta}
=
\frac{g^2}{2}
\left(
\sigma_\alpha\sigma_\beta
-e_{g,\alpha}e_{g,\beta}
\right).
\end{equation}
For a rank-two tensor $A_{\alpha\beta}$, the STF projection is
\begin{equation}
[A_{\alpha\beta}]_{\rm STF}
=
\frac12\left(A_{\alpha\beta}+A_{\beta\alpha}\right)
-\frac13\delta_{\alpha\beta}A_{\gamma\gamma}.
\end{equation}
Rotational symmetry gives
\begin{equation}
\left[
\int_{S^2}b(\eta)
\left(
\sigma_\alpha\sigma_\beta
-e_{g,\alpha}e_{g,\beta}
\right)d\bm\sigma
\right]_{\rm STF}
=
-\frac32\Omega_2
e_{g,\langle\alpha}e_{g,\beta\rangle}.
\end{equation}
Substitution into Eq.~\eqref{eq:app_maxwell_weak_form} gives
\begin{equation}
\int c_{\langle\alpha}c_{\beta\rangle}Q(f,f)\,d\bm c
=
-\frac38\Omega_2
\int ff_*g_{\langle\alpha}g_{\beta\rangle}
\,d\bm c\,d\bm c_*.
\end{equation}
Because the compatible frame is centered,
\begin{equation}
\int c_\alpha f\,d\bm c=0,
\end{equation}
and hence
\begin{equation}
\int ff_*g_{\langle\alpha}g_{\beta\rangle}
\,d\bm c\,d\bm c_*
=
2n\int f c_{\langle\alpha}c_{\beta\rangle}\,d\bm c.
\end{equation}
Therefore
\begin{equation}
\int c_{\langle\alpha}c_{\beta\rangle}Q(f,f)\,d\bm c
=
-\lambda_\pi
\int c_{\langle\alpha}c_{\beta\rangle}f\,d\bm c,
\qquad
\lambda_\pi=\frac34n\Omega_2.
\label{eq:app_maxwell_lambda_pi}
\end{equation}

For the cubic vector observable $\Phi_\alpha(\bm c)=c^2c_\alpha$,
\begin{equation}
c^2c_\alpha+c_*^2c_{*\alpha}
=
2\left(C^2+\frac{g^2}{4}\right)C_\alpha
+
(\bm C\cdot\bm g)g_\alpha.
\end{equation}
The first term is collision invariant, so the collisional change is
\begin{equation}
\Delta_3^\alpha
=
g^2\left[
(\bm C\cdot\bm\sigma)\sigma_\alpha
-
(\bm C\cdot\bm e_g)e_{g,\alpha}
\right].
\end{equation}
The same angular tensor gives
\begin{equation}
\int_{S^2}b(\eta)\Delta_3^\alpha\,d\bm\sigma
=
\frac{\Omega_2}{2}
\left[
g^2C_\alpha
-3(\bm C\cdot\bm g)g_\alpha
\right].
\end{equation}
Using centering,
\begin{align}
\int ff_*g^2C_\alpha\,d\bm c\,d\bm c_*
&=
n\int f c^2c_\alpha\,d\bm c,
\\
\int ff_*(\bm C\cdot\bm g)g_\alpha\,d\bm c\,d\bm c_*
&=
n\int f c^2c_\alpha\,d\bm c.
\end{align}
Therefore
\begin{equation}
\int c^2c_\alpha Q(f,f)\,d\bm c
=
-\lambda_q\int c^2c_\alpha f\,d\bm c,
\qquad
\lambda_q=\frac12n\Omega_2.
\label{eq:app_maxwell_lambda_q}
\end{equation}
These identities hold for the full nonlinear collision operator on a centered compatible state; no linearized collision model has been introduced.

\subsection{Conventional coordinate Chapman--Enskog evaluation}

For comparison, let
\begin{equation}
f=M+\varepsilon f^{(1)}+O(\varepsilon^2),
\qquad
M=n(2\pi\theta)^{-3/2}
\exp\!\left(-\frac{c^2}{2\theta}\right).
\end{equation}
Using the Euler equations, propagation of the local Maxwellian reduces to
\begin{equation}
\frac1M
(\partial_t+v_\alpha\partial_\alpha)M
=
\frac1{2\theta}c_\alpha c_\beta D_{\alpha\beta}
+
\frac{c_\alpha}{2\theta^2}(c^2-5\theta)\partial_\alpha\theta.
\label{eq:app_maxwell_CE_source}
\end{equation}
The two polynomials are precisely the velocity-coordinate images of the Fock sectors $\mathscr H_{2,\mathrm{tl}}$ and $\mathscr H_{3,\mathrm{vec}}$.  Their collision eigenvalues are the exact rates \eqref{eq:app_maxwell_lambda_pi}--\eqref{eq:app_maxwell_lambda_q}; inversion of those two scalar modes and the standard Gaussian pairings therefore gives
\begin{equation}
\mu=\frac{p}{\lambda_\pi},
\qquad
\kappa=\frac52\frac{k_B}{m}\frac{p}{\lambda_q},
\qquad
\Pr=\frac23,
\end{equation}
in agreement with Sec.~\ref{subsec:maxwell_CE1}.  The general moving-Maxwellian Chapman--Enskog organization and the corresponding Fock pairings are given in the Fock formulation of the Lebowitz--Frisch--Helfand model \cite{Karlin2026LFHFock}; the purpose of this appendix is only to verify independently the Maxwell-specific collision rates in conventional velocity coordinates.


\section{From Cartesian Hermite to Laguerre--spherical coordinates}
\label{app:maxwell_Hermite_Laguerre}

This appendix supplies the explicit coordinate derivation referred to in the
main text after Eq.~\eqref{eq:maxwell_Fock_spherical_states} and in the
paragraph following Eq.~\eqref{eq:maxwell_spectrum_Fock_labels}.  In
particular, it derives Eq.~\eqref{eq:maxwell_Fock_to_Laguerre_spherical}
starting from the Cartesian Hermite realization introduced in
Eq.~\eqref{eq:maxwell_H_basis}.  The purpose is to make clear that no
additional structure is imposed on the abstract Fock space when the
conventional Laguerre--spherical form is introduced.  The
decomposition into the labels $(j,\ell,m)$ is already intrinsic to the
trace/STF decomposition of each homogeneous Fock level.  What changes here
is only the coordinate representation of the same abstract states.

We work in three dimensions.  Recall from
Eqs.~\eqref{eq:maxwell_P_basis} and \eqref{eq:maxwell_ladder_P} that the
polynomial realization is
\begin{equation}
\mathcal R_P|\bm n\rangle
=
\frac{\operatorname{He}_{\bm n}(\bm\xi)}
{\sqrt{\bm n!}},
\qquad
\operatorname{He}_{\bm n}(\bm\xi)
=
\prod_{\alpha=1}^{3}
\operatorname{He}_{n_\alpha}(\xi_\alpha),
\label{eq:app_maxwell_RP_cartesian}
\end{equation}
and that
\begin{equation}
\mathcal R_P\hat a_\alpha^\dagger \mathcal R_P^{-1}
=
\xi_\alpha-\partial_{\xi_\alpha}.
\label{eq:app_maxwell_RP_creation}
\end{equation}
The Hermite-function realization is obtained afterwards from
\begin{equation}
\mathcal R_H=S_{\sqrt W}\mathcal R_P,
\label{eq:app_maxwell_RH_from_RP}
\end{equation}
so it is sufficient to derive the polynomial part first.

\subsection{The same Fock level in Cartesian Hermite coordinates}

The homogeneous Fock level of degree $r$ has the intrinsic rotational
decomposition
\begin{equation}
\mathscr H_r
=
\bigoplus_{\substack{j,\ell\geq0\\2j+\ell=r}}
\mathscr K_{j\ell}.
\end{equation}
As in Sec.~\ref{subsec:maxwell_linearized_spectrum}, let
$\mathcal Y_{\ell m}(\bm z)$ be a homogeneous solid harmonic,
\begin{equation}
\mathcal Y_{\ell m}(\lambda\bm z)
=
\lambda^\ell\mathcal Y_{\ell m}(\bm z),
\qquad
\Delta_{\bm z}\mathcal Y_{\ell m}(\bm z)=0.
\label{eq:app_maxwell_solid_harmonic}
\end{equation}
The corresponding homogeneous Bargmann polynomial is
\begin{equation}
\Phi_{j\ell m}(\bm z)
=
(\bm z^2)^j\mathcal Y_{\ell m}(\bm z),
\qquad
r=2j+\ell.
\label{eq:app_maxwell_Phi_jlm}
\end{equation}

For the present derivation, fix the irrelevant overall normalization by
choosing
\begin{equation}
\mathcal R_B|j,\ell,m\rangle
=
\Phi_{j\ell m}(\bm z).
\end{equation}
Since $\Phi_{j\ell m}$ is homogeneous of degree $r$, it has a Cartesian
monomial expansion
\begin{equation}
\Phi_{j\ell m}(\bm z)
=
\sum_{|\bm n|=r}
C_{\bm n}^{(j\ell m)}\,
\bm z^{\bm n},
\qquad
|\bm n|
=
n_1+n_2+n_3.
\label{eq:app_maxwell_cartesian_monomial_expansion}
\end{equation}
Using
\begin{equation}
\mathcal R_B|\bm n\rangle
=
\frac{\bm z^{\bm n}}{\sqrt{\bm n!}},
\end{equation}
the same abstract state can therefore be written as
\begin{equation}
|j,\ell,m\rangle
=
\sum_{|\bm n|=r}
C_{\bm n}^{(j\ell m)}
\sqrt{\bm n!}\,
|\bm n\rangle .
\label{eq:app_maxwell_jlm_cartesian_Fock}
\end{equation}
Applying $\mathcal R_P$ gives
\begin{equation}
\mathcal R_P|j,\ell,m\rangle
=
\sum_{|\bm n|=r}
C_{\bm n}^{(j\ell m)}
\operatorname{He}_{\bm n}(\bm\xi).
\label{eq:app_maxwell_jlm_cartesian_Hermite}
\end{equation}

Equation~\eqref{eq:app_maxwell_jlm_cartesian_Hermite} gives the direct
connection with the Cartesian Hermite representation introduced earlier:
the state $|j,\ell,m\rangle$ is a fixed linear combination of products of
Cartesian Hermite polynomials, all belonging to the same total Fock level
$r=2j+\ell$.  The remaining calculation derives a closed radial--angular
form for precisely this same linear combination.

\subsection{Solid harmonics under the Hermite polynomial realization}

From Eq.~\eqref{eq:app_maxwell_RP_creation}, for any polynomial
$P(\hat{\bm a}^\dagger)$,
\begin{equation}
\mathcal R_P
P(\hat{\bm a}^\dagger)|0\rangle
=
P(\bm\xi-\bm\nabla_{\bm\xi})\,1.
\label{eq:app_maxwell_polynomial_creation_image}
\end{equation}
It is useful first to evaluate the right-hand side in a form that makes
the role of harmonicity transparent.

Introduce an auxiliary vector $\bm t$.  Since
\begin{equation}
[
\bm t\cdot\bm\xi,
-\bm t\cdot\bm\nabla_{\bm\xi}
]
=
|\bm t|^2
\end{equation}
is a scalar, the Baker--Campbell--Hausdorff formula gives
\begin{align}
\exp\!\left[
\bm t\cdot
(\bm\xi-\bm\nabla_{\bm\xi})
\right]1
&=
\exp(\bm t\cdot\bm\xi)
\exp(-\bm t\cdot\bm\nabla_{\bm\xi})
\exp\!\left(-\frac{|\bm t|^2}{2}\right)1
\nonumber\\
&=
\exp\!\left(
\bm t\cdot\bm\xi-\frac{|\bm t|^2}{2}
\right).
\label{eq:app_maxwell_Hermite_generating}
\end{align}
On the other hand,
\begin{equation}
\exp\!\left(-\frac12\Delta_{\bm\xi}\right)
\exp(\bm t\cdot\bm\xi)
=
\exp\!\left(-\frac{|\bm t|^2}{2}\right)
\exp(\bm t\cdot\bm\xi).
\end{equation}
Comparison of the coefficients of powers of $\bm t$ yields the
multivariate Hermite, or Wick, identity
\begin{equation}
P(\bm\xi-\bm\nabla_{\bm\xi})\,1
=
\exp\!\left(-\frac12\Delta_{\bm\xi}\right)
P(\bm\xi).
\label{eq:app_maxwell_Wick_identity}
\end{equation}

For a solid harmonic, however,
\begin{equation}
\Delta_{\bm\xi}\mathcal Y_{\ell m}(\bm\xi)=0.
\end{equation}
Consequently every positive power of $\Delta_{\bm\xi}$ also annihilates
$\mathcal Y_{\ell m}$, and
Eq.~\eqref{eq:app_maxwell_Wick_identity} reduces to
\begin{equation}
\mathcal R_P\,
\mathcal Y_{\ell m}(\hat{\bm a}^\dagger)|0\rangle
=
\mathcal Y_{\ell m}(\bm\xi).
\label{eq:app_maxwell_solid_harmonic_RP}
\end{equation}

This identity is the coordinate counterpart of the STF property.  The
derivative contractions generated by the Hermite transformation correspond
to traces of the underlying symmetric tensor.  They vanish for the
harmonic, equivalently symmetric trace-free, component.

A homogeneous solid harmonic can be written in spherical coordinates as
\begin{equation}
\mathcal Y_{\ell m}(\bm\xi)
=
\mathcal N_{\ell m}\,
|\bm\xi|^\ell
Y_\ell^m(\widehat{\bm\xi}),
\qquad
\widehat{\bm\xi}
=
\frac{\bm\xi}{|\bm\xi|},
\label{eq:app_maxwell_solid_to_spherical}
\end{equation}
where $\mathcal N_{\ell m}$ depends only on the chosen normalization of the
solid harmonic.  Thus the spherical harmonic in the final coordinate
formula is already present in the solid harmonic defining the abstract
Fock state; it is not introduced as a separate basis transformation.

\subsection{Radial pair creation}

It remains to determine the coordinate action of the trace-pair creation
operator
\begin{equation}
\hat{\bm a}^\dagger\cdot\hat{\bm a}^\dagger
=
\sum_{\alpha=1}^{3}
(\hat a_\alpha^\dagger)^2.
\end{equation}
From Eq.~\eqref{eq:app_maxwell_RP_creation},
\begin{align}
\mathcal D
&:=
\mathcal R_P
\left(
\hat{\bm a}^\dagger\cdot\hat{\bm a}^\dagger
\right)
\mathcal R_P^{-1}
\nonumber\\
&=
(\bm\xi-\bm\nabla_{\bm\xi})^2.
\end{align}
Expanding the square requires some care because multiplication and
differentiation do not commute.  In three dimensions,
\begin{align}
(\bm\xi-\bm\nabla)^2
&=
\sum_{\alpha=1}^{3}
(\xi_\alpha-\partial_\alpha)
(\xi_\alpha-\partial_\alpha)
\nonumber\\
&=
|\bm\xi|^2
-
\sum_{\alpha}
\left(
\xi_\alpha\partial_\alpha
+
\partial_\alpha\xi_\alpha
\right)
+
\Delta_{\bm\xi}
\nonumber\\
&=
|\bm\xi|^2
-
2\bm\xi\cdot\bm\nabla_{\bm\xi}
-
3
+
\Delta_{\bm\xi}.
\end{align}
Hence
\begin{equation}
\mathcal D
=
|\bm\xi|^2
-
2\bm\xi\cdot\bm\nabla_{\bm\xi}
-
3
+
\Delta_{\bm\xi}.
\label{eq:app_maxwell_radial_pair_operator}
\end{equation}

To expose the radial action of this operator, set
\begin{equation}
x=\frac{|\bm\xi|^2}{2},
\qquad
\alpha=\ell+\frac12,
\label{eq:app_maxwell_radial_x_alpha}
\end{equation}
and consider a function of the form
\begin{equation}
F(\bm\xi)
=
f(x)\mathcal Y_{\ell m}(\bm\xi).
\label{eq:app_maxwell_radial_ansatz}
\end{equation}
Since
\begin{equation}
\frac{\partial x}{\partial\xi_\alpha}
=
\xi_\alpha,
\end{equation}
we have
\begin{equation}
\bm\nabla f(x)
=
f'(x)\bm\xi
\end{equation}
and
\begin{equation}
\Delta f(x)
=
|\bm\xi|^2f''(x)+3f'(x)
=
2x f''(x)+3f'(x).
\label{eq:app_maxwell_radial_scalar_laplacian}
\end{equation}

The homogeneity and harmonicity of $\mathcal Y_{\ell m}$ give
Euler's identity
\begin{equation}
\bm\xi\cdot\bm\nabla
\mathcal Y_{\ell m}
=
\ell\,\mathcal Y_{\ell m}
\label{eq:app_maxwell_Euler_solid}
\end{equation}
and
\begin{equation}
\Delta\mathcal Y_{\ell m}=0.
\end{equation}
Therefore
\begin{align}
\bm\xi\cdot\bm\nabla
\left[
f(x)\mathcal Y_{\ell m}
\right]
&=
\left[
2x f'(x)+\ell f(x)
\right]
\mathcal Y_{\ell m},
\label{eq:app_maxwell_radial_Euler}
\\
\Delta
\left[
f(x)\mathcal Y_{\ell m}
\right]
&=
\left[
\Delta f(x)
\right]\mathcal Y_{\ell m}
+
2\bm\nabla f(x)\cdot
\bm\nabla\mathcal Y_{\ell m}
+
f(x)\Delta\mathcal Y_{\ell m}
\nonumber\\
&=
\left[
2x f''(x)
+
3f'(x)
+
2\ell f'(x)
\right]
\mathcal Y_{\ell m}
\nonumber\\
&=
2
\left[
x f''(x)
+
\left(\ell+\frac32\right)f'(x)
\right]
\mathcal Y_{\ell m}.
\label{eq:app_maxwell_radial_laplacian}
\end{align}

Substitution of Eqs.~\eqref{eq:app_maxwell_radial_Euler} and
\eqref{eq:app_maxwell_radial_laplacian} into
Eq.~\eqref{eq:app_maxwell_radial_pair_operator} gives
\begin{align}
\mathcal D
\left[
f(x)\mathcal Y_{\ell m}
\right]
&=
2
\Bigg[
x f''(x)
+
\left(
\ell+\frac32-2x
\right)f'(x)
\nonumber\\
&\hspace{30mm}
+
\left(
x-\ell-\frac32
\right)f(x)
\Bigg]
\mathcal Y_{\ell m}.
\end{align}
With $\alpha=\ell+1/2$, this becomes
\begin{equation}
\mathcal D
\left[
f(x)\mathcal Y_{\ell m}
\right]
=
2
\left[
x f''
+
(\alpha+1-2x)f'
+
(x-\alpha-1)f
\right]
\mathcal Y_{\ell m}.
\label{eq:app_maxwell_radial_pair_reduced}
\end{equation}

Equation~\eqref{eq:app_maxwell_radial_pair_reduced} is the radial
differential operator whose repeated action generates the generalized
Laguerre polynomials.

\subsection{Identification with generalized Laguerre polynomials}

Let
\begin{equation}
f(x)=L_j^{(\alpha)}(x),
\end{equation}
where $L_j^{(\alpha)}$ is the generalized Laguerre polynomial.  We shall use
three standard identities.  First, the Laguerre differential equation is
\begin{equation}
x\frac{d^2L_j^{(\alpha)}}{dx^2}
+
(\alpha+1-x)
\frac{dL_j^{(\alpha)}}{dx}
+
jL_j^{(\alpha)}
=
0.
\label{eq:app_maxwell_Laguerre_ODE}
\end{equation}
For $j\geq1$, the derivative identity is
\begin{equation}
x\frac{dL_j^{(\alpha)}}{dx}
=
jL_j^{(\alpha)}
-
(j+\alpha)L_{j-1}^{(\alpha)},
\label{eq:app_maxwell_Laguerre_derivative}
\end{equation}
and the three-term recurrence is
\begin{equation}
(j+1)L_{j+1}^{(\alpha)}
=
(2j+\alpha+1-x)L_j^{(\alpha)}
-
(j+\alpha)L_{j-1}^{(\alpha)}.
\label{eq:app_maxwell_Laguerre_recurrence}
\end{equation}

We now substitute $f=L_j^{(\alpha)}$ into the square bracket in
Eq.~\eqref{eq:app_maxwell_radial_pair_reduced}.  Using the differential
equation \eqref{eq:app_maxwell_Laguerre_ODE},
\begin{align}
&
x L_j^{(\alpha)\prime\prime}
+
(\alpha+1-2x)L_j^{(\alpha)\prime}
+
(x-\alpha-1)L_j^{(\alpha)}
\nonumber\\
&=
-
(\alpha+1-x)L_j^{(\alpha)\prime}
-
jL_j^{(\alpha)}
+
(\alpha+1-2x)L_j^{(\alpha)\prime}
+
(x-\alpha-1)L_j^{(\alpha)}
\nonumber\\
&=
-xL_j^{(\alpha)\prime}
+
(x-\alpha-1-j)L_j^{(\alpha)}.
\label{eq:app_maxwell_Laguerre_step1}
\end{align}
For $j\geq1$, Eq.~\eqref{eq:app_maxwell_Laguerre_derivative} then gives
\begin{align}
&
-xL_j^{(\alpha)\prime}
+
(x-\alpha-1-j)L_j^{(\alpha)}
\nonumber\\
&=
-jL_j^{(\alpha)}
+
(j+\alpha)L_{j-1}^{(\alpha)}
+
(x-\alpha-1-j)L_j^{(\alpha)}
\nonumber\\
&=
(x-\alpha-1-2j)L_j^{(\alpha)}
+
(j+\alpha)L_{j-1}^{(\alpha)}.
\label{eq:app_maxwell_Laguerre_step2}
\end{align}
Finally, the recurrence
\eqref{eq:app_maxwell_Laguerre_recurrence} shows that
\begin{equation}
(x-\alpha-1-2j)L_j^{(\alpha)}
+
(j+\alpha)L_{j-1}^{(\alpha)}
=
-(j+1)L_{j+1}^{(\alpha)}.
\label{eq:app_maxwell_Laguerre_step3}
\end{equation}
Thus
\begin{equation}
\mathcal D
\left[
L_j^{(\alpha)}(x)
\mathcal Y_{\ell m}(\bm\xi)
\right]
=
-2(j+1)
L_{j+1}^{(\alpha)}(x)
\mathcal Y_{\ell m}(\bm\xi).
\label{eq:app_maxwell_Laguerre_raising}
\end{equation}
For $j=0$, the same result follows directly from
$L_0^{(\alpha)}=1$ and
\begin{equation}
L_1^{(\alpha)}(x)
=
-x+\alpha+1.
\end{equation}

Repeated application of
Eq.~\eqref{eq:app_maxwell_Laguerre_raising}, beginning with
$L_0^{(\alpha)}=1$, therefore yields
\begin{equation}
\mathcal D^j
\mathcal Y_{\ell m}(\bm\xi)
=
(-2)^j j!\,
L_j^{(\ell+1/2)}
\left(
\frac{|\bm\xi|^2}{2}
\right)
\mathcal Y_{\ell m}(\bm\xi).
\label{eq:app_maxwell_pair_to_Laguerre}
\end{equation}

This equation is the desired closed form for the repeated radial-pair
creation in Hermite polynomial coordinates.

\subsection{A low-order check}

The first nontrivial radial state gives a useful direct check of the
calculation.  Take $j=1$ and $\ell=0$.  In abstract Fock notation the state
is, up to normalization,
\begin{equation}
\left(
\hat{\bm a}^\dagger\cdot
\hat{\bm a}^\dagger
\right)|0\rangle
=
\sum_{\alpha=1}^{3}
(\hat a_\alpha^\dagger)^2|0\rangle .
\end{equation}
The Cartesian Hermite realization gives
\begin{align}
\mathcal R_P
\left(
\hat{\bm a}^\dagger\cdot
\hat{\bm a}^\dagger
\right)|0\rangle
&=
\sum_{\alpha=1}^{3}
\operatorname{He}_2(\xi_\alpha)
\nonumber\\
&=
\sum_{\alpha=1}^{3}
(\xi_\alpha^2-1)
\nonumber\\
&=
|\bm\xi|^2-3.
\label{eq:app_maxwell_radial_example_cartesian}
\end{align}
On the other hand,
\begin{equation}
\alpha=\frac12,
\qquad
x=\frac{|\bm\xi|^2}{2},
\end{equation}
and
\begin{equation}
L_1^{(1/2)}(x)
=
-x+\frac32.
\end{equation}
Therefore
\begin{equation}
-2L_1^{(1/2)}
\left(
\frac{|\bm\xi|^2}{2}
\right)
=
|\bm\xi|^2-3,
\label{eq:app_maxwell_radial_example_Laguerre}
\end{equation}
in exact agreement with
Eq.~\eqref{eq:app_maxwell_radial_example_cartesian}.

This elementary example exhibits explicitly what happens in the general
case: a sum of Cartesian Hermite products of fixed total Hermite order is
reorganized into a radial Laguerre polynomial multiplied by an irreducible
angular factor.

\subsection{Hermite-function coordinate representation}

We may now apply the result to the abstract state
\begin{equation}
|j,\ell,m\rangle
\propto
\left(
\hat{\bm a}^\dagger\cdot
\hat{\bm a}^\dagger
\right)^j
\mathcal Y_{\ell m}(\hat{\bm a}^\dagger)|0\rangle .
\end{equation}
Equations~\eqref{eq:app_maxwell_solid_harmonic_RP} and
\eqref{eq:app_maxwell_pair_to_Laguerre} give
\begin{equation}
\mathcal R_P|j,\ell,m\rangle
\propto
L_j^{(\ell+1/2)}
\left(
\frac{|\bm\xi|^2}{2}
\right)
\mathcal Y_{\ell m}(\bm\xi).
\label{eq:app_maxwell_RP_jlm_Laguerre_solid}
\end{equation}
Using Eq.~\eqref{eq:app_maxwell_solid_to_spherical},
\begin{equation}
\mathcal R_P|j,\ell,m\rangle
\propto
|\bm\xi|^\ell
L_j^{(\ell+1/2)}
\left(
\frac{|\bm\xi|^2}{2}
\right)
Y_\ell^m(\widehat{\bm\xi}).
\label{eq:app_maxwell_RP_jlm_Laguerre_spherical}
\end{equation}
Since only the overall normalization is relevant here, a factor
$2^{-\ell/2}$ may be absorbed into it, giving the equivalent form
\begin{equation}
\mathcal R_P|j,\ell,m\rangle
\propto
\left(
\frac{|\bm\xi|}{\sqrt2}
\right)^\ell
L_j^{(\ell+1/2)}
\left(
\frac{|\bm\xi|^2}{2}
\right)
Y_\ell^m(\widehat{\bm\xi}).
\label{eq:app_maxwell_RP_jlm_final}
\end{equation}

Finally, using
$\mathcal R_H=S_{\sqrt W}\mathcal R_P$,
\begin{equation}
\mathcal R_H|j,\ell,m\rangle
\propto
\sqrt W\,
\left(
\frac{|\bm\xi|}{\sqrt2}
\right)^\ell
L_j^{(\ell+1/2)}
\left(
\frac{|\bm\xi|^2}{2}
\right)
Y_\ell^m
\left(
\frac{\bm\xi}{|\bm\xi|}
\right).
\label{eq:app_maxwell_RH_jlm_Laguerre_spherical}
\end{equation}
This is the Laguerre--spherical expression quoted in
Eq.~\eqref{eq:maxwell_Fock_to_Laguerre_spherical}.

\subsection{Radial orthogonality and the meaning of the index \texorpdfstring{$\alpha=\ell+1/2$}{alpha=l+1/2}}

The same calculation also explains the parameter
$\alpha=\ell+1/2$ appearing in the generalized Laguerre polynomial.
The square of the Hermite-function coordinate contains the Gaussian factor
\begin{equation}
W\propto
\exp\!\left(-\frac{|\bm\xi|^2}{2}\right)
=
e^{-x}.
\end{equation}
In spherical velocity coordinates,
\begin{equation}
d^3\xi
=
|\bm\xi|^2\,d|\bm\xi|\,d\Omega.
\end{equation}
Since the solid harmonic contributes the radial factor
$|\bm\xi|^\ell$, its squared modulus contributes
$|\bm\xi|^{2\ell}$.  With
\begin{equation}
|\bm\xi|=\sqrt{2x},
\qquad
d|\bm\xi|
=
\frac{dx}{\sqrt{2x}},
\end{equation}
the radial part of the norm is therefore proportional to
\begin{align}
&
e^{-|\bm\xi|^2/2}
|\bm\xi|^{2\ell}
|\bm\xi|^2\,d|\bm\xi|
\nonumber\\
&\qquad\propto
e^{-x}
x^\ell x^{1/2}\,dx
\nonumber\\
&\qquad=
x^{\ell+1/2}e^{-x}\,dx
=
x^\alpha e^{-x}\,dx.
\end{align}
Thus the radial orthogonality is exactly the standard generalized
Laguerre orthogonality
\begin{equation}
\int_0^\infty
x^\alpha e^{-x}
L_j^{(\alpha)}(x)
L_{j'}^{(\alpha)}(x)\,dx
=
\frac{\Gamma(j+\alpha+1)}{j!}
\delta_{jj'}.
\label{eq:app_maxwell_Laguerre_orthogonality}
\end{equation}
The value $\alpha=\ell+1/2$ is therefore fixed simultaneously by the
three-dimensional radial measure and by the angular rank $\ell$.

\subsection{Interpretation}

Equations~\eqref{eq:app_maxwell_jlm_cartesian_Hermite} and
\eqref{eq:app_maxwell_RH_jlm_Laguerre_spherical} are two coordinate
descriptions of the same abstract Fock state.  The first displays it as a
linear combination of Cartesian Hermite products,
\begin{equation}
\operatorname{He}_{n_1}(\xi_1)
\operatorname{He}_{n_2}(\xi_2)
\operatorname{He}_{n_3}(\xi_3),
\qquad
n_1+n_2+n_3=2j+\ell,
\end{equation}
whereas the second displays the same state in radial--angular velocity
coordinates as
\begin{equation}
L_j^{(\ell+1/2)}
\left(
\frac{|\bm\xi|^2}{2}
\right)
|\bm\xi|^\ell
Y_\ell^m(\widehat{\bm\xi}).
\end{equation}

It is important that the latter polynomial contains ordinary powers of
$|\bm\xi|$ below the maximal degree $2j+\ell$.  These lower ordinary
polynomial degrees do not correspond to lower Fock levels.  They are the
trace, or Wick-contraction, terms inherent in the Hermite coordinate
realization.  Equation~\eqref{eq:app_maxwell_jlm_cartesian_Hermite} makes
this explicit: every Cartesian Hermite product entering the state has the
same total Hermite order
\begin{equation}
r=2j+\ell.
\end{equation}

Thus the chain of identifications is
\begin{equation}
\begin{aligned}
\mathscr H_r
&\xrightarrow{\ \mathrm{trace/STF}\ }
|j,\ell,m\rangle,
\\[1mm]
|j,\ell,m\rangle
&\xrightarrow{\ \mathcal R_P,\,\mathcal R_H\ }
\text{Laguerre--spherical Hermite coordinates}.
\end{aligned}
\end{equation}
The generalized Laguerre polynomial and the spherical harmonic therefore
arise from the coordinate realization of the intrinsic Fock state; they do
not define an additional or adapted Fock structure.


\begin{thebibliography}{99}
\bibitem{Grad1949Hermite}
Harold Grad. ``Note on N-Dimensional Hermite Polynomials.'' \emph{Communications on Pure and Applied Mathematics} \textbf{2}(4), 325--330 (1949). DOI: \url{https://doi.org/10.1002/cpa.3160020402}.

\bibitem{Grad1949KineticTheory}
Harold Grad. ``On the Kinetic Theory of Rarefied Gases.'' \emph{Communications on Pure and Applied Mathematics} \textbf{2}(4), 331--407 (1949). DOI: \url{https://doi.org/10.1002/cpa.3160020403}.

\bibitem{Maxwell1867DynamicalTheory}
James Clerk Maxwell. ``On the Dynamical Theory of Gases.'' \emph{Philosophical Transactions of the Royal Society of London} \textbf{157}, 49--88 (1867). DOI: \url{https://doi.org/10.1098/rstl.1867.0004}.

\bibitem{IkenberryTruesdell1956I}
E. Ikenberry and C. Truesdell. ``On the Pressures and the Flux of Energy in a Gas according to Maxwell's Kinetic Theory, I.'' \emph{Journal of Rational Mechanics and Analysis} \textbf{5}, 1--54 (1956). DOI: \url{https://doi.org/10.1512/iumj.1956.5.55001}.

\bibitem{Truesdell1956II}
C. Truesdell. ``On the Pressures and the Flux of Energy in a Gas according to Maxwell's Kinetic Theory, II.'' \emph{Journal of Rational Mechanics and Analysis} \textbf{5}, 55--128 (1956). DOI: \url{https://doi.org/10.1512/iumj.1956.5.55002}.

\bibitem{TruesdellMuncaster1980Maxwell}
C. Truesdell and R. G. Muncaster. \emph{Fundamentals of Maxwell's Kinetic Theory of a Simple Monatomic Gas: Treated as a Branch of Rational Mechanics}. Academic Press, New York (1980).

\bibitem{Bobylev1975Fourier}
A. V. Bobylev. ``The Method of the Fourier Transform in the Theory of the Boltzmann Equation for Maxwell Molecules.'' \emph{Doklady Akademii Nauk SSSR} \textbf{225}(5), 1041--1044 (1975). In Russian; English translation: Soviet Physics Doklady 20 (1975), 820--822.

\bibitem{Desvillettes2003FourierBoltzmann}
Laurent Desvillettes. ``About the Use of the Fourier Transform for the Boltzmann Equation.'' \emph{Rivista di Matematica della Universit\`a di Parma} \textbf{2*}, 1--99 (2003).

\bibitem{Bobylev1988MaxwellReview}
A. V. Bobylev. ``The Theory of the Nonlinear Spatially Uniform Boltzmann Equation for Maxwell Molecules.'' \emph{Soviet Scientific Reviews C: Mathematical Physics} \textbf{7}, 111--233 (1988).

\bibitem{Bobylev1984ExactRelaxation}
A. V. Bobylev. ``Exact Solutions of the Nonlinear Boltzmann Equation and the Theory of Relaxation of a Maxwellian Gas.'' \emph{Theoretical and Mathematical Physics} \textbf{60}(2), 820--841 (1984). DOI: \url{https://doi.org/10.1007/BF01018983}.

\bibitem{Ernst1979MaxwellMoments}
M. H. Ernst. ``Exact Solution of the Non-Linear Boltzmann Equation for Maxwell Models.'' \emph{Physics Letters A} \textbf{69}(6), 390--392 (1979). DOI: \url{https://doi.org/10.1016/0375-9601(79)90385-2}.

\bibitem{GlangetasLiXu2016Triangular}
L{\'e}o Glangetas, Hao-Guang Li, and Chao-Jiang Xu. ``Sharp Regularity Properties for the Non-Cutoff Spatially Homogeneous Boltzmann Equation.'' \emph{Kinetic and Related Models} \textbf{9}(2), 299--371 (2016). DOI: \url{https://doi.org/10.3934/krm.2016.9.299}.

\bibitem{WangChangUhlenbeck1952}
C. S. Wang Chang and G. E. Uhlenbeck. ``On the Propagation of Sound in Monatomic Gases.'' Technical report Project M999, Engineering Research Institute, University of Michigan, Ann Arbor, Michigan (1952). Technical report; Handle 2027.42/4098; later reprinted in Studies in Statistical Mechanics, Vol. V.

\bibitem{WangChangUhlenbeck1970}
C. S. Wang Chang and G. E. Uhlenbeck. ``The Kinetic Theory of Gases.'' In \emph{Studies in Statistical Mechanics}, edited by J. de Boer and G. E. Uhlenbeck, Vol. 5, pp. 43--75, North-Holland, Amsterdam (1970). Reprint of the 1952 University of Michigan report.

\bibitem{AltermanFrankowskiPekeris1962}
Z. Alterman, K. Frankowski, and C. L. Pekeris. ``Eigenvalues and Eigenfunctions of the Linearized Boltzmann Collision Operator for a Maxwell Gas and for a Gas of Rigid Spheres.'' \emph{Astrophysical Journal Supplement Series} \textbf{7}, 291--331 (1962).

\bibitem{LernerMorimotoPravdaStarovXu2013}
Nicolas Lerner, Yoshinori Morimoto, Karel Pravda-Starov, and Chao-Jiang Xu. ``Phase Space Analysis and Functional Calculus for the Linearized Landau and Boltzmann Operators.'' \emph{Kinetic and Related Models} \textbf{6}(3), 625--648 (2013). DOI: \url{https://doi.org/10.3934/krm.2013.6.625}.

\bibitem{Grad1965HTheorem}
Harold Grad. ``On Boltzmann's H-Theorem.'' \emph{Journal of the Society for Industrial and Applied Mathematics} \textbf{13}(1), 259--277 (1965). DOI: \url{https://doi.org/10.1137/0113016}.

\bibitem{CercignaniIllnerPulvirenti1994}
Carlo Cercignani, Reinhard Illner, and Mario Pulvirenti. \emph{The Mathematical Theory of Dilute Gases}. Springer, New York (1994), Applied Mathematical Sciences 106.

\bibitem{Tang1993Hermite}
Tao Tang. ``The Hermite Spectral Method for Gaussian-Type Functions.'' \emph{SIAM Journal on Scientific Computing} \textbf{14}(3), 594--606 (1993). DOI: \url{https://doi.org/10.1137/0914038}.

\bibitem{SarnaGiesselmannTorrilhon2020}
Neeraj Sarna, Jan Giesselmann, and Manuel Torrilhon. ``Convergence Analysis of Grad's Hermite Expansion for Linear Kinetic Equations.'' \emph{SIAM Journal on Numerical Analysis} \textbf{58}(2), 1164--1194 (2020). DOI: \url{https://doi.org/10.1137/19M1270884}.

\bibitem{CannoneKarch2010}
Marco Cannone and Grzegorz Karch. ``Infinite Energy Solutions to the Homogeneous Boltzmann Equation.'' \emph{Communications on Pure and Applied Mathematics} \textbf{63}(6), 747--778 (2010). DOI: \url{https://doi.org/10.1002/cpa.20298}.

\bibitem{ChoMorimotoWangYang2016}
Yong-Kum Cho, Yoshinori Morimoto, Shuaikun Wang, and Tong Yang. ``Probability Measures with Finite Moments and the Homogeneous Boltzmann Equation.'' \emph{SIAM Journal on Mathematical Analysis} \textbf{48}(4), 2399--2413 (2016). DOI: \url{https://doi.org/10.1137/15M105104X}.

\bibitem{CaiTorrilhon2015HardSphereLinearized}
Zhenning Cai and Manuel Torrilhon. ``Approximation of the Linearized Boltzmann Collision Operator for Hard-Sphere and Inverse-Power-Law Models.'' \emph{Journal of Computational Physics} \textbf{295}, 617--643 (2015). DOI: \url{https://doi.org/10.1016/j.jcp.2015.04.031}.

\bibitem{WangCai2019HermiteBoltzmann}
Yanli Wang and Zhenning Cai. ``Approximation of the Boltzmann Collision Operator Based on Hermite Spectral Method.'' \emph{Journal of Computational Physics} \textbf{397}, 108815 (2019). DOI: \url{https://doi.org/10.1016/j.jcp.2019.07.014}.

\bibitem{HiemstraKesslerAbdelmalik2026WignerEckart}
Ren{\'e} R. Hiemstra, Torsten Ke{\ss}ler, and Michael R. A. Abdelmalik. ``Wigner--Eckart Factorization of the Spectral Boltzmann Collision Operator.'' arXiv:2605.28475 (2026). Preprint. DOI: \url{https://doi.org/10.48550/arXiv.2605.28475}.

\bibitem{BobylevToscani1992GeneralizedH}
A. V. Bobylev and Giuseppe Toscani. ``On the Generalization of the Boltzmann H-Theorem for a Spatially Homogeneous Maxwell Gas.'' \emph{Journal of Mathematical Physics} \textbf{33}(7), 2578--2586 (1992). DOI: \url{https://doi.org/10.1063/1.529578}.

\bibitem{PulvirentiToscani1996Fourier}
Ada Pulvirenti and Giuseppe Toscani. ``The Theory of the Nonlinear Boltzmann Equation for Maxwell Molecules in Fourier Representation.'' \emph{Annali di Matematica Pura ed Applicata} \textbf{171}, 181--204 (1996). DOI: \url{https://doi.org/10.1007/BF01759387}.

\bibitem{Tanaka1978Probabilistic}
Hiroshi Tanaka. ``Probabilistic Treatment of the Boltzmann Equation of Maxwellian Molecules.'' \emph{Zeitschrift f{\"u}r Wahrscheinlichkeitstheorie und Verwandte Gebiete} \textbf{46}(1), 67--105 (1978). DOI: \url{https://doi.org/10.1007/BF00535689}.

\bibitem{Cook1951SecondQuantization}
J. M. Cook. ``The Mathematics of Second Quantization.'' \emph{Proceedings of the National Academy of Sciences of the United States of America} \textbf{37}(7), 417--420 (1951). DOI: \url{https://doi.org/10.1073/pnas.37.7.417}.

\bibitem{Segal1956TensorAlgebras}
Irving E. Segal. ``Tensor Algebras over Hilbert Spaces. I.'' \emph{Transactions of the American Mathematical Society} \textbf{81}, 106--134 (1956). DOI: \url{https://doi.org/10.1090/S0002-9947-1956-0076317-8}.

\bibitem{Bargmann1961HilbertAnalytic}
V. Bargmann. ``On a Hilbert Space of Analytic Functions and an Associated Integral Transform. Part I.'' \emph{Communications on Pure and Applied Mathematics} \textbf{14}(3), 187--214 (1961). DOI: \url{https://doi.org/10.1002/cpa.3160140303}.

\bibitem{Karlin2026LFHFock}
Ilya Karlin. ``Fock-Space Representation of the Lebowitz--Frisch--Helfand Kinetic Model.'' arXiv:2608.25833 (2026). \url{https://arxiv.org/abs/2608.25833}.

\bibitem{MouhotPareschi2006FastBoltzmann}
Cl{\'e}ment Mouhot and Lorenzo Pareschi. ``Fast Algorithms for Computing the Boltzmann Collision Operator.'' \emph{Mathematics of Computation} \textbf{75}(256), 1833--1852 (2006). DOI: \url{https://doi.org/10.1090/S0025-5718-06-01874-6}.

\bibitem{Fock1932}
V. A. Fock. ``Konfigurationsraum und zweite Quantelung.'' \emph{Zeitschrift f{\"u}r Physik} \textbf{75}, 622--647 (1932). DOI: \url{https://doi.org/10.1007/BF01344458}.

\bibitem{Berezin1966}
F. A. Berezin. \emph{The Method of Second Quantization}. Academic Press, New York (1966).

\bibitem{Hall2013}
Brian C. Hall. \emph{Quantum Theory for Mathematicians}. Springer, New York (2013), Graduate Texts in Mathematics 267. DOI: \url{https://doi.org/10.1007/978-1-4614-7116-5}.

\end{thebibliography}
\end{document}